\documentclass[10pt,a4paper]{article}
\usepackage{a4wide}

\usepackage{microtype}

\usepackage[utf8]{inputenc}

\usepackage{latexsym}
\usepackage{amssymb,amsmath,amsthm}
\usepackage{graphicx}
\usepackage{url}
\usepackage{hyperref}
\usepackage{xcolor}
\usepackage{enumerate}
\usepackage[capitalize]{cleveref}
\usepackage{multirow}
\usepackage{authblk}
\usepackage[normalem]{ulem}
\usepackage{mathtools} 
\usepackage{wrapfig} 

\newtheorem{theorem}{Theorem}
\newtheorem{proposition}[theorem]{Proposition}
\newtheorem{observation}[theorem]{Observation}
\newtheorem{lemma}[theorem]{Lemma}
\newtheorem{corollary}[theorem]{Corollary}

\newtheorem{claim}{Claim}[theorem]
\newtheorem*{claim*}{Claim}
\theoremstyle{definition}
\newtheorem{definition}[theorem]{Definition}
\theoremstyle{remark}

\makeatletter
\newcommand{\ProofEndBox}{{\ifhmode\unskip\nobreak\hfil\penalty50 \else
          \leavevmode\fi\quad\vadjust{}\nobreak\hfill$\Box$
            \finalhyphendemerits=0 \par}}
\makeatother

\newcommand{\R}{{\mathbb{R}}}

\newcommand\cA{\mathcal{A}}
\newcommand\cB{\mathcal{B}}
\newcommand\cC{\mathcal{C}}
\newcommand\cG{\mathcal{G}}
\newcommand\cH{\mathcal{H}}
\newcommand\cI{\mathcal{I}}
\newcommand\cK{\mathcal{K}}
\newcommand\cO{\mathcal{O}}
\newcommand\cS{\mathcal{S}}
\newcommand\cR{\mathcal{R}}
\newcommand\cT{\mathcal{T}}
\newcommand\cU{\mathcal{U}}
\newcommand\cV{\mathcal{V}}

\newcommand\vx{{\sf x}}
\newcommand\vy{{\sf y}}
\newcommand\vz{{\sf z}}

\newcommand\makevec[1]{{\bf #1}}
\def \aa {\makevec{a}}
\def \bb {\makevec{b}}
\def \cc {\makevec{c}}
\def \dd {\makevec{d}}
\def \ee {\makevec{e}}

\def \rr {\makevec{r}}
\renewcommand \ss {\makevec{s}}
\def \ttt {\makevec{t}}
\def \uu {\makevec{u}}
\def \vv {\makevec{v}}
\def \ww {\makevec{w}}
\def \xx {\makevec{x}}
\def \yy {\makevec{y}}
\def \zz {\makevec{z}}
\def \PP {\makevec{P}}
\def \RR {\makevec{R}}

\newcommand\ggamma{{\boldsymbol{\gamma}}}
\newcommand\ddelta{{\boldsymbol{\delta}}}
\newcommand\vvarepsilon{{\boldsymbol{\varepsilon}}}

\newcounter{sideremark}
\newcommand{\marrow}{\stepcounter{sideremark}\marginpar{$\boldsymbol{\longleftarrow\scriptstyle\arabic{sideremark}}$}}

\newif\ifcmts
\cmtstrue

\ifcmts
\newcommand{\pavel}[1]{{\color{teal}\vskip 5pt\textsf{*** (Pavel) \marrow #1\vskip 5pt}}}
\newcommand{\martin}[1]{{\color{blue}\vskip 5pt\textsf{*** (Martin) \marrow #1\vskip 5pt}}}
\else
\newcommand{\pavel}[1]{}
\newcommand{\martin}[1]{}
\fi

\title{Shellability is hard even for balls
}

\author[1,2]{Pavel Pat\'{a}k}
\author[1]{Martin Tancer}

\affil[1]{\small Department of Applied Mathematics, Charles University, Malostransk\'{e} n\'{a}m.
25, 118~00~~Praha~1, Czech Republic}
\affil[2]{Czech Technical University in Prague, Faculty of Information Technology, Th\'{a}kurova 2700/9, 160~00~~Praha~6, Czech Republic}

\begin{document}

\maketitle

\begin{abstract}
The main goal of this paper is to show that shellability is NP-hard for
triangulated $d$-balls (this also gives hardness for triangulated
$d$-manifolds/$d$-pseudomanifolds with boundary) as soon as $d \geq 3$. This
  extends our earlier work with Goaoc, Pat\'akov\'a and Wagner 
  on hardness
  of shellability of $2$-complexes and answers some questions implicitly raised
  by Danaraj and Klee in 1978 and explicitly mentioned by Santamar\'ia-Galvis
  and Woodroofe. Together with the main goal, we also prove that collapsibility is NP-hard for
  3-complexes embeddable in 3-space, extending an earlier work of the
  second author
  and answering an open question mentioned by Cohen, Fasy, Miller, Nayyeri,
  Peng and Walkington; and that shellability is NP-hard for 2-complexes embeddable in 3-space, answering
  another question of Santamar\'ia-Galvis and Woodroofe (in a slightly
  stronger form than what is given by the main result).
\end{abstract}
\bigskip

{\bf Keywords:} Shellability, NP-hard, collapsibility, $3$-ball recognition,
$3$-manifolds
\bigskip

\tableofcontents

\section{Introduction}
\label{s:intro}

A simplicial complex is \emph{pure} if all its \emph{facets} (i.e. maximal
faces) have the same dimension. A pure simplicial complex $K$ is shellable if there
is an ordering 
$F_1, \dots, F_m$ of all its facets such that for every $i \in
  \{2,\dots, m\}$, the complex
$K[F_i] \cap K[F_1, \dots, F_{i-1}]$
  is pure and
  $(d-1)$-dimensional where $K[\vartheta_1, \dots, \vartheta_{i-1}]$ stands for the subcomplex
  of $K$ spanned by $\vartheta_1, \dots, \vartheta_{i-1}$. Shellability of $K$ is often a great
advantage for understanding the structure of $K$.
There are however
also variants of shellings for cell/polytopal complexes, PL-manifolds, posets or even
monoids which makes this notion widely applicable, for example, 
in combinatorial topology, polytope theory, combinatorial commutative algebra,
or group theory. We will explain in more detail relations to combinatorial topology which is the most relevant for our paper; for references to other areas, we
refer to the introductions of~\cite{gpptw19, santamaria-galvis-woodroofe21}.

History of shellings in combinatorial topology traces back at least to works of
Furch~\cite{furch24}, Newman~\cite{newman26} and Frankl~\cite{frankl31} on
non-shellable balls. 
Shellings of triangulated manifolds or only pseudomanifolds are of particular
interest because a shellable pseudomanifold has to be a sphere or ball.  A
significant effort has been devoted to constructions of non-shellable balls
with various properties, especially in dimension $3$. Apart from the three
aforementioned references~\cite{furch24, newman26, frankl31}, this is also (part of) the
contents of~\cite{rudin58,bing64,ziegler98,lutz08}, for example. For more details on
various constructions (up to 1998) we refer to~\cite{ziegler98}. Non-shellable
triangulated 3-spheres also exist~\cite{lickorish91, lutz04}; however, as far as we are
able to judge they are more scarce in the literature. Shellings in the context
of balls/spheres/manifolds appear to be very useful also in the current
research---often in surprising context; see, e.g.,~\cite{benedetti-ziegler11,
dgkm16, adiprasito-benedetti17, adiprasito-benedetti20, adiprasito-liu20}.

From the computational point of view shellability was treated already in 70's in papers
of Danaraj and Klee~\cite{danaraj-klee78algo, danaraj-klee78spheres}.
In~\cite{danaraj-klee78algo} they provide an efficient algorithm for
$2$-pseudomanifolds and they point out that it is unknown whether shellability
can be tested in polynomial time for a pseudomanifold of dimension at least
$3$. In~\cite{danaraj-klee78spheres} they explicitly ask how efficiently can
shellability be tested? For general simplicial complexes this has been recently
answered in our earlier work with Goaoc, Pat\'{a}kov\'{a} and Wagner~\cite{gpptw19} by showing that this problem is NP-hard (already for
2-dimensional complexes). 
However, it is clear from the context
of~\cite{danaraj-klee78algo, danaraj-klee78spheres} that Danaraj and Klee were
also interested in shellability of pseudomanifolds and in particular spheres,
which could be in principle easier. We completely answer this question for
pseudomanifolds with boundary. The same construction also provides an answer
for manifolds with boundary or the balls. The question of the complexity status of shellability
for balls is also explicitly raised in~\cite{santamaria-galvis-woodroofe21}.
By an analogy with complexity status of the unknot recognition,
Santamar\'{i}a-Galvis and Woodroofe~\cite{santamaria-galvis-woodroofe21}
speculate that shellability of balls could belong to coNP. The answer is `no'
unless $\hbox{NP} = \hbox{coNP}$ which is widely believed not to be true.

\begin{theorem}
  \label{t:main}
  Let $d \geq 3$.
  \begin{enumerate}[(i)]
  \item
It is NP-hard to decide whether a given $d$-dimensional pseudomanifold
  with boundary is shellable.
\item
It is NP-hard to decide whether a given $d$-dimensional manifold
  with boundary is shellable.
\item
  It is NP-hard to decide whether a given $d$-dimensional triangulated ball is shellable.
  \end{enumerate}
\end{theorem}

Let us remark that the input for all variants of Theorem~\ref{t:main} is always
a simplicial complex. It is polynomial time testable whether a given simplicial
complex is a pseudomanifold which perhaps makes the statement~(i) the most
natural. 
For (ii) and (iii) we implicitly assume that the input is correctly
given as it is not know how to check in polynomial time 
whether a given simplicial complex is a $d$-manifold for $d \geq 4$ or whether
it is a $d$-ball for $d \geq 3$. Let us also remark that shellability of
simplicial complexes easily belongs to NP by guessing the shelling. Therefore
any NP-hardness result immediately implies NP-completeness (provided that the
input is correct). This is also true for collapsibility discussed later on.

We also point out that although our main result, Theorem~\ref{t:main}, is
algorithmic, it is interesting as well from purely mathematical point of view.
For a proof of the theorem, we will build triangulated balls $\cK_\phi$ such that
$\cK_\phi$ is shellable if and only if a certain logical formula $\phi$ is satisfiable.
(This is discussed in more detail in the paragraph on the proof method below.)
In particular, for non-satisfiable $\phi$, we construct non-shellable balls
where the reason for non-shellability 
is not due to presence or absence of a problematic small subcomplex of the triangulation (such as
non-existence of a free simplex or existence of a suitably `knotted' edge) but it is more global---it is deeply hidden in satisfiability of $\phi$. This complements earlier results on
non-shellable balls. 

Regarding spheres, the complexity of recognition of shellability for
spheres (or manifolds or pseudomanifolds without boundary) remains open. It is
plausible to expect that a modification of our construction would yield
NP-hardness as well. However, currently, we do not have any suitable
modification (which we would be able to analyze).

\paragraph{Relations to $3$-ball/$3$-sphere recognition.}
A partial motivation of Danaraj and Klee~\cite[page~38,
(7)(9)(12)]{danaraj-klee78spheres} to consider the complexity of shellability is
the $3$-sphere/$3$-ball recognition problem. (Note that the $3$-sphere and $3$-ball recognition are closely
related. If there is an efficient algorithm for one of them, then there is an
efficient algorithm for another one by either taking the double of the $3$-ball
or removing a tetrahedron in the $3$-sphere; see,
e.g.~\cite[Corollary~1.1]{schleimer11}.) They proposed that there may exist an
algorithm that transforms a $3$-manifold (with boundary) into another
$3$-manifold (with boundary) in such a way that if we have started with a sphere
or ball then the result would be a shellable sphere or ball. Then it would be
sufficient to check shellability. 
Although it is not known whether the suggested algorithm may be obtained, there
is some hope that it is realistic.\footnote{The $3$-sphere recognition problem is
closely related to the unknot recognition problem. For the unknot recognition
problem,
Lackenby~\cite{lackenby15} proved that there is a polynomial upper bound on the
Reidemeister moves required to untangle an unknot. The desired algorithm of
Danaraj and Klee could be based on considering a suitable subdivision of the
input manifold.
Interpreting loosely
subdivisions and shellings (for 3-sphere recognition) as an analogue of
Reidemeister moves (for unknot recognition), such an algorithm would be a counterpart of Lackenby's
result.} 
On the other hand,
Theorem~\ref{t:main} surely reveals some
problems with the second step. For completeness, let us point that the
$3$-ball/$3$-sphere recognition problem is more understood nowadays. It belongs
to NP~\cite{ivanov01,schleimer11} and modulo generalized Riemann hypothesis to
coNP~\cite{zentner18}; however no polynomial time algorithm is known.

\paragraph{Proof method and related results.}

First of all, the core of the proof of Theorem~\ref{t:main} is to show
NP-hardness for $d=3$. It is easy to check (see,
e.g.~\cite[Section~2]{gpptw19}) that a simplicial complex $K$ is shellable if
and only if the cone $c * K$ (over $K$ with apex $c$) is shellable. Because a cone
over a $k$-pseudomanifold with boundary, or a $k$-manifold with boundary, or a
$k$-ball is a $(k+1)$-pseudomanifold with boundary, or $(k+1)$-manifold with
boundary, or a $(k+1)$-ball respectively, the hardness result for $d=3$ implies
hardness for any $d \geq 3$. Thus we restrict our discussion to the interesting
case $d=3$.

Then the idea of our proof can be split into two major steps. 

The first step is related to another notion for simplification of simplicial
complexes called collapsibility. It has been proved by Malgouyres and
Franc\'{e}s~\cite{malgouyres-frances08} that it is NP-hard to decide whether a
3-dimensional simplicial complex collapses to a $1$-complex and this was further generalized
by the second author~\cite{tancer16} by showing that it is NP-hard to decide whether a
given simplicial complex is collapsible. (Even when restricted to complexes of
dimension $3$. The dimension 3 cannot be dropped as $2$-complexes can be
collapsed greedily (see~\cite[page~20]{hms93} or~\cite[Lemma~1 +
Corollary~1]{malgouyres-frances08}).) We further generalize the main result
of~\cite{tancer16} by showing hardness even when restricted to complexes
embeddable into $3$-space. This has been also mentioned as an open question
by Cohen, Fasy, Miller, Nayyeri,  Peng and Walkington~\cite{cfmnpw14}.

\begin{theorem}
\label{t:collapsibility_hard}
Collapsibility of simplicial complexes embeddable into $\R^3$ is NP-hard.
\end{theorem}

Theorem~\ref{t:collapsibility_hard} is a non-trivial generalization of the main
result of~\cite{tancer16} because the complexes used in~\cite{tancer16} are
very far from being embeddable into $3$-space. There are two types of
problems---they use gadgets which are themselves not embeddable into $3$-space
and also the local connections of the gadgets prevent embeddability into 3-space. We overcome the first problem with the aid of \emph{turbines} of
Santamar\'{\i}a-Galvis and Woodroofe~\cite{santamaria-galvis-woodroofe21} as
well as by introducing new variable gadget (called bipyramid).
Santamar\'{\i}a-Galvis and Woodroofe~\cite{santamaria-galvis-woodroofe21}
provided a modified proof of hardness of shellability of
$2$-complexes when compared with~\cite{gpptw19}. One of their changes is that they simplified some
of the gadgets into the aforementioned turbines, embeddable into 3-space.
Shellability and collapsibility are related closely enough so that the turbines can be used
also for Theorem~\ref{t:collapsibility_hard}. Regarding the second
aforementioned problem we overcome it by suitable placements (and adjustments)
of the gadgets employing a variant of a 3SAT problem called planar monotone
rectilinear 3SAT problem which has been proved to be NP-hard by de Berg and
Khosravi~\cite{deberg-khosravi10}.

As a side remark, let us point out that Santamar\'{\i}a-Galvis and 
Woodroofe~\cite{santamaria-galvis-woodroofe21} also asked whether shellability
remains NP-complete for complexes embeddable in $3$-space.
Theorem~\ref{t:main} shows that this is NP-hard, therefore NP-complete
(for $3$-complexes).
By a minor modification of the construction in the proof of
Theorem~\ref{t:collapsibility_hard}, we can show that the problem remains
NP-hard even for 2-complexes which is the best possible (and it is mainly the
context of~\cite{santamaria-galvis-woodroofe21}).
\begin{theorem}
\label{t:embedded_hard}
Shellability of simplicial $2$-complexes embeddable into $\R^3$ is NP-hard.
\end{theorem}

The second major step for the proof of Theorem~\ref{t:main} is related to thickenings of simplicial complexes to
manifolds. Thickenings are classical tools in PL topology (see
e.g.~\cite{hudson69, rourke-sanderson72}) and they allow to
thicken a simplicial complex PL embedded into $\R^d$ to a $d$-manifold with
boundary while preserving some important properties (for example preserving the
homotopy type). In the context of shellability, thickenings were used in
constructions of Frankl~\cite{frankl31} and Bing~\cite{bing64} of non-shellable
balls (see also~\cite{ziegler98} for a nice description of the ideas). 
The idea is to take a 2-dimensional simplicial complex embedded into 
3-space which is contractible but not collapsible (for example Bing's house
with 2 rooms) due to the reason that it does not have any so called free face. (It is
not possible to start with a collapse at all.) Then one can build a thick
(3-dimensional) version of such a complex in a way that the thick version again
misses the free face (this time for shelling). This is however not automatic
but it requires some careful choices to get the desired property. For example, in
case of Bing's construction~\cite{bing64}, Bing strongly uses that the Bing's
house is built of (possibly non-convex) polygons embedded in axis-aligned
planes. Finally, because the original $2$-complex was contractible, the
resulting thick complex is a $3$-ball, thereby a non-shellable $3$-ball.

Following the aforementioned ideas, Theorem~\ref{t:main} may seem as a quite
direct corollary of Theorem~\ref{t:collapsibility_hard} (perhaps requiring to
revisit the proof a bit). Indeed, Theorem~\ref{t:collapsibility_hard} is proved in a way that
for a formula $\phi$, which is an instance of planar monotone
rectilinear 3SAT, we build a complex $\cK'_\phi$ such that $\cK'_\phi$ is
collapsible if and only if $\phi$ is satisfiable which shows NP-hardness.
It also follows from the construction that $\cK'_\phi$ is always contractible.
Therefore, following the ideas above we may thicken $\cK'_\phi$ to certain
ball $\cK_\phi$ in a way that $\cK_\phi$ is shellable if and only if $\phi$
is satisfiable.

The real difficulty is however to provide the construction in such a way that 
the both implications of `$\cK_\phi$ is shellable if and only if $\phi$
is satisfiable' are satisfied simultaneously. The earlier works say nothing
about the `if' implication. 
However, intuitively, this should
work if the triangulation is fine enough (or well chosen). The `only if'
implication is even more problematic. The earlier constructions allow to block
shellability in the first step (or the last one depending on the order of
shellings) but they do not provide any recipe if the reason for
non-collapsibility is more complex (hidden in satisfiability of $\phi$). This is a huge
and critical difference for us. We have tried several attempts of constructions 
such that ideally shellings of $\cK_\phi$ would \emph{have to} follow collapses of
$\cK'_\phi$. (Thereby non-collapsibility of $\cK'_\phi$ would imply
non-shellability of $\cK_\phi$.)  However, based on our attempts, we now
doubt that it is easily possible to derive Theorem~\ref{t:main} from
Theorem~\ref{t:collapsibility_hard} by some kind of general construction.

Therefore, we use a somewhat different approach. The formula $\cK'_\phi$ is
built from various auxiliary gadgets which emulate logical dependencies in
$\phi$. We thicken each gadget simultaneously and then we try to glue them
together so that we emulate the same logical dependencies. So in principle,
formally speaking our proof of Theorem~\ref{t:main} does not use
Theorem~\ref{t:collapsibility_hard}. However, we still include
Theorem~\ref{t:collapsibility_hard} as it is an essential step for
understanding the proof of Theorem~\ref{t:main}.

This approach of thickening the gadgets simultaneously and then gluing them
together brings, however, numerous new problems that have to be resolved.
These problems are also a reason why the proof of Theorem~\ref{t:main} is
technically much more complicated then the proofs of hardness of shellability
of $2$-complexes in~\cite{gpptw19} or~\cite{santamaria-galvis-woodroofe21}.
Here we spell out only a few of the problems and a brief sketch how do we treat
them.

The most significant problem is perhaps the treatment of the free faces. The
gadgets of $\cK'_\phi$ contain numerous free faces (in collapsibility sense)
in order to encode the logical dependencies of $\phi$. These free faces are
in general edges and the logical dependency for $\cK'_\phi$ is emulated by gluing such an edge
to other gadgets. Following Bing~\cite{bing64}, a natural way how to thicken an
edge $e$ of $\cK'_\phi$ is to thicken it to a triangulated cube $C_e$ in
$\cK_\phi$. If the original edge is not free, it is somehow easier to avoid
free tetrahedra in $C_e$. However, if $e$ is free, we would ideally 
like to emulate this only by a single tetrahedron in $C_e$; otherwise $C_e$ may
admit some unwanted shellings. In our construction we solve this by ad hoc
solutions for different types of connections; see
Section~\ref{s:shellability_gadgets}. (We also need to know that no new free
tetrahedra appear even after gluing the gadgets. This is treated mainly in
Lemmas~\ref{l:1house_blocked} and~\ref{l:turbine_blocked}.)

A related problem to the previous one is the exact description of
triangulations of individual gadgets. In principle we need to specify exact set
of vertices and tetrahedra of each of the gadgets (this would also specify
edges and triangles as the gadgets are pure). However, this is
too tedious to do so directly and
it would make our paper even longer.
Thus we use an approach in between, we
first describe the gadgets first as certain polytopal complexes which are
easier to understand. Only then we specify the triangulations of individual
polytopes. (Both steps are done in Section~\ref{s:shellability_gadgets}.) But we have to make some careful choices so that the free faces
behave as we want and also that the gadgets can be glued together (in
Section~\ref{s:shellability_construction}). There is
also one more caveat: triangulations of polytopes may be in principle
non-shellable. Thus we need some general recipe so that even after such a
triangulation the gadgets are shellable as we need. On high level, this is
usually guaranteed by Lemma~\ref{l:shell_canonical} which provides a recipe how
to shell a single triangulated polytope. But then we of course have to check
for every gadget individually, that this really works; see mainly
Lemmas~\ref{l:shelling_house} and~\ref{l:shelling_turbine}.

We also need to get a $3$-ball after gluing all gadgets. This becomes more
tricky around so called variable gadgets which intersect with some other
gadget only in an edge, thus the gaps around these intersections have to be
filled by yet another gadgets. (See mainly the construction in
Subsections~\ref{ss:thick_variable}, \ref{ss:thick_splitter}
and~\ref{ss:blocker_house}.) This makes shellings somewhat complicated for
satisfiable formulas.

 There is also one more problem regarding gluing the gadgets. In $\cK'_\phi$,
 the individual gadgets are typically glued along edges (or $1$-complexes).
 This is something which is easy to treat with respect to subdivisions. In
 $\cK_\phi$ the shared parts of the gadgets thicken to $2$-dimensional pieces.
 Often we need to glue rows of squares together in our original polytopal
 complex before full triangulation.  After the full triangulation, it could in
 principle happen that the diagonals will not match. In all cases, this can be
 resolved but it makes the proof and the description (in
 Section~\ref{s:shellability_construction}) more tedious. 

\paragraph{Organization.} Theorem~\ref{t:collapsibility_hard} is proved in
Sections~\ref{s:thin_gadgets}--\ref{s:correctness}. Theorem~\ref{t:main} is
proved in Sections~\ref{s:shelling_balls}--\ref{s:shell=>sat}.
Theorem~\ref{t:embedded_hard} is proved in Section~\ref{s:embedded_hard}.

\section{Preliminaries}
\label{s:prelim}

Here we briefly explain the notions used in the introduction.

\paragraph{Simplicial complexes, polyhedra, triangulations and pseudomanifolds.}

Throughout the paper we work with \emph{geometric simplicial complexes}, that
is, collections $K$ of simplices in some $\R^d$ such that face of any simplex
in $K$ belongs to $K$ as well and if $\sigma, \tau \in K$, then $\sigma \cap
\tau$ is a face of both $\sigma$ and $\tau$ (possibly empty).  A geometric
simplicial complex can be modified to an abstract simplicial complex which
collects only the list of vertices and the list of subsets of vertices that
form a simplex. We use abstract simplicial complexes only when we consider a
simplicial complex as an input of a decision problem.  A \emph{polyhedron}
$|K|$ of a simplicial complex $K$ is the union of simplices in $K$. A
\emph{triangulation} of a topological space $X$ is a simplicial complex $K$
with $|K|$ homeomorphic to $X$. We also often use the term \emph{polyhedron}
without specifying $K$. Then we mean the polyhedron of some simplicial complex
without specifying any triangulation. In our descriptions of gadgets, we will
often specify the polyhedron first and only then we specify the triangulation.
For more details on simplicial complexes and polyhedra, we refer
to~\cite{matousek03, rourke-sanderson72}. A \emph{facet} in a
simplicial complex is an inclusion-wise maximal face. A simplicial complex is
\emph{pure} if all its facets have the same dimension.
A pure simplicial complex of dimension $d$ is a \emph{$d$-pseudomanifold} (with
boundary) if each face of dimension $(d-1)$ is contained in at most two facets
(and necessarily at least one due to pureness).

\paragraph{Collapsibility.}
A face in a simplicial complex is \emph{free} (for collapsibility) if it is not
a facet but it is contained in only one facet.\footnote{Some authors require
that the dimension of the face and the facet differ by one but this does not
affect the resulting notion of collapsibility.} A simplicial complex $K'$
arises from $K$ by an \emph{elementary collapse} if it is obtained from $K$ by
removing a free face $\sigma$ (and all faces that contain $\sigma$). In this
case we also say that \emph{$K$ collapses to $K'$ through $\sigma$}.
A
simplicial complex $K$ collapses to $L$ if $L$ can be obtained from $K$ by a
sequence of elementary collapses. A simplicial complex $K$ is
\emph{collapsible} if it collapses to a point.

\paragraph{Shellability.}
Given faces $\vartheta_1, \dots, \vartheta_m$ of a simplicial complex $K$, the
symbol $K[\vartheta_1,\dots, \vartheta_m]$ will denote the complex generated by
$\vartheta_1, \dots, \vartheta_m$, that is, the complex formed by those faces
$\sigma$ which are in some $\vartheta_i$ for $i \in \{1, \dots, m\}$.

The following definition is the standard definition of the shelling of a
simplicial complex. We call it \emph{shelling up} in order to distinguish it
from \emph{shelling down} that we will use
very heavily. 

\begin{definition}[Shelling up]
  \label{d:shell_up}
  Let $K$ be a pure $d$-dimensional simplicial complex. A \emph{shelling up} is
  an ordering $F_1, \dots, F_m$ of all its facets such that for every $i \in
  \{2,\dots, m\}$, the complex $K[F_i] \cap K[F_1, \dots, F_{i-1}]$ is pure and $(d-1)$-dimensional.
  A simplicial complex is \emph{shellable} if it admits shelling up.
\end{definition}

It will be a great advantage for us to revert the shelling order so that it
works in the same direction as for collapses. We call this \emph{shelling down}, or
simply \emph{shelling}. We stress that this is less standard. In most of the
literature, this would be called reverse shelling. However this is the default
notion for us and in some portion of literature it appears in this
direction, possibly in slightly different context; see,
e.g.,~\cite{rourke-sanderson72,adiprasito-liu20}.
We define it in slightly more general setting allowing to perform only a few initial shelling steps.

\begin{definition}[Shelling = Shelling down]
\label{d:shell_down}
  Let $K$ be a pure $d$-dimensional simplicial complex with $m$ facets. A \emph{shelling down}
  is an ordering $F_1, \dots, F_k$ for $k \leq m-1$ of some\footnote{Because we use only some
  facets, a more precise name would be \emph{incomplete shelling down}.
  However, we want to avoid using too many adjectives.} facets of $K$ such that for every $i \in
  \{1, \dots, k\}$ the complex $K[F_i] \cap K[F_{i+1}, \dots, F_{m}]$ is pure and
  $(d-1)$-dimensional where $F_{k+1}, \dots, F_m$ are the remaining facets of
  $K$ in an arbitrary order. A complex $K$ \emph{shells down} to a subcomplex
  $L$ if $L = K[F_{k+1}, \dots, F_m]$ with respect to some shelling down as
  above. 
\end{definition}
When we say shelling without up or down we implicitly mean shelling down.
The relation between the two aforementioned notions is that a pure
$d$-dimensional simplicial complex is shellable if and only if it admits a
shelling down to a $d$-simplex (i.e. one of its facets). Indeed, it is
straightforward to check that $F_1, \dots, F_m$ is a shelling up if and only if
$F_m, \dots, F_2$ is a shelling down to $F_1$.

\section{Gadgets for collapsibility}
\label{s:thin_gadgets}
In this section, we describe gadgets that will be used in the proof of
Theorem~\ref{t:collapsibility_hard}.

\subsection{Bipyramid}
First we define a bipyramid which will serve as our variable gadget. Namely, by
\emph{bipyramid} we mean the complex which is a join of the boundary of a
triangle with an edge. In more detail we take a boundary of a triangle with
vertices $u, v, w$ and an edge with vertices $a^+$ and $a^-$. The bipyramid
contains all simplices with vertex sets of a form $\sigma \cup \tau$, where
$\sigma \subsetneq \{u,v,w\}$ and $\tau \subseteq \{a^+, a^-\}$;  see
Figure~\ref{f:bipyramid}. 

\begin{figure}
  \begin{center}
    \includegraphics[page=9]{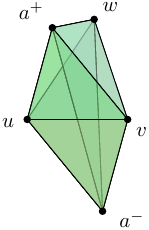}
    \caption{The bipyramid}
    \label{f:bipyramid}
  \end{center}
\end{figure}

For purposes of the following lemma, let $K$ be a complex and $L$ its
subcomplex. Following~\cite{tancer16}, we define the constraint complex of
$L$ in $K$ as a subcomplex of $L$ formed by the faces contained
in some face of $K \setminus L$.

\begin{lemma}
  \label{l:bipyramid}
\hfill

  \begin{enumerate}[(i)] 
    \item The bipyramid collapses to a subcomplex formed by the edge
      $va^+$ and the triangles $uva^-$, $uwa^-$ and $vwa^-$ (and the subfaces
      of these faces).
    \item The bipyramid collapses to a subcomplex formed by the edge
      $va^-$ and the triangles $uva^+$, $uwa^+$ and $vwa^+$ (and the subfaces
      of these faces).

    \item Assume that the bipyramid $B$ is a subcomplex of a complex $K$ such
      that the constraint complex of $B$ in $K$ is the complex
      formed by the edges $uv, uw, vw, va^+$ and $va^-$. Then in any collapsing
      of $K$ it is not possible that both $va^+$ and $va^-$ would be collapsed
      before all three edges $uv, uw$ and $vw$ via collapses removing
      $va^+$ and some triangle in $K$ outside $B$ and removing $va^-$ and some
      triangle in $K$ outside $B$.
  \end{enumerate}
\end{lemma}

\begin{proof}
 For (i), it is straightforward to check that the collapses through the faces
  $uwa^+, wa^+, ua^+, a^+a^-$ in this order are admissible and they yield the
  desired complex. Case (ii) is symmetric by swapping $a^+$ with $a^-$.

  For (iii) for contradiction assume that both $va^+$ and $va^-$ are collapsed
  before all three edges $uv, uw$ and $vw$ and via the collapses as in the
  statement.

  Without loss of generality let us
  assume that $va^+$ is collapsed before $va^-$. Let $K'$ be the intermediate
  complex obtained exactly before collapsing $va^-$. Now we refine those
  elementary  collapses yielding $K'$ which remove some face of $B$ to
  collapses which remove exactly two faces via each elementary collapse. For
  example, an elementary collapse through a vertex $a$ in a facet (triangle)
  $abc$ can be refined to two elementary collapses, first removing $ab$ and $abc$ and the
  second removing $a$
  and $ac$. In a similar way, we refine a collapse through a vertex or an edge
  in a tetrahedron. (This is related to a description of collapses via discrete
  Morse theory~\cite{forman98}. In the rest of the proof we essentially study
  the properties of the corresponding discrete Morse matching until we show
  that it cannot exist.)
  
  Let $\tau$ be a triangle
  in $B$ which is not in $K'$. 
  Such a \emph{missing} triangle has to be matched with either a
  tetrahedron or an edge with which it is removed in some collapse. In case of tetrahedron,
  we observe that this tetrahedron belongs to $B$. This follows
  from the description of the constraint complex of $B$ in
  $K$---no triangle of $B$ is in a tetrahedron outside $K$. In case of an edge
  we observe that this edge is an edge of $B$ different from $uv, uw, vw, va^+$
  and $va^-$. For $uv, uw$ and $vw$ this is excluded from the definition of $K$
  and the fact that $uv, uw$ and $vw$ are collapsed later. For $va^+$ 
  and $va^-$ this follows from the condition on their collapses. Therefore, we
  altogether have three admissible tetrahedra $uva^+a^-$, $vwa^+a^-$ and
  $uwa^+a^-$ and the five edges $a^+a^-$, $ua^+$, $ua^-$, $wa^+$ and $wa^-$
  for matching the missing triangles. We will show that this is insufficient
  obtaining the desired contradiction.

  The requirement on the collapses of $va^+$ and $va^-$ implies that
  all triangles in $B$ containing $va^+$ have to be removed before collapsing
  $va^+$ as well as all triangles in $B$ containing $va^-$ have to be removed before collapsing
  $va^-$. That is, the five triangles $va^+a^-$, $uva^+$, $uva^-$, $vwa^+$ and
  $vwa^-$ have to be missing. Considering triangles $uva^+$, $uva^-$, the only
  tetrahedron containing any of them is $uva^+a^-$, thus  at most
  one of them can be matched to this tetrahedron. This implies that either the
  former one is matched to the edge $ua^+$ or the latter one is matched to the
  edge $ua^-$. However, all triangles containing such an edge have to be
  removed before collapsing such an edge. In particular, the triangle $ua^+a^-$
  has to be missing. By a symmetric argument replacing $u$ with $w$, we get
  that $wa^+a^-$ is missing as well. Therefore we have at least seven missing
  triangles. Because we have only three tetrahedra, at least four such edges
  have to be matched with missing triangles. But now we deduce that also the
  triangles $uwa^+$ and $uwa^-$ are missing---both of them contain two edges
  (from our list of possible edges matching missing triangles) and they have to
  be removed before removing any such edge. Thus we have nine missing triangles
  which is too much for three tetrahedra and five edges.
\end{proof}
\subsection{1-house}

A \emph{1-house} is a $2$-polyhedron depicted in Figure~\ref{f:1house}. Its
construction goes back to Malgouyres, Franc\'{es} and
the second author~\cite{malgouyres-frances08,tancer16} modifying the well-known Bing's
house.\footnote{$1$-house implicitly appears
in~\cite{malgouyres-frances08} and explicitly in~\cite{tancer16} under a slightly
different name. The name $1$-house is according to~\cite{gpptw19} and refers to
the property that that there is one free edge.}
It
contains a distinguished edge $f$ and a distinguished rectangle $L$ called the
lower wall. When triangulated, it is collapsible but any collapsing must start
in an edge subdividing $f$. We will need the properties of 1-house as stated 
in Lemma~9 in~\cite{gpptw19}; however the content of this lemma also goes back
to~\cite{malgouyres-frances08,tancer16}.

Given a polyhedron $P$, a triangulation $K$ of $P$ is \emph{geometric}, if $P = |K|$.

\begin{figure}
  \begin{center}
    \includegraphics[page=1]{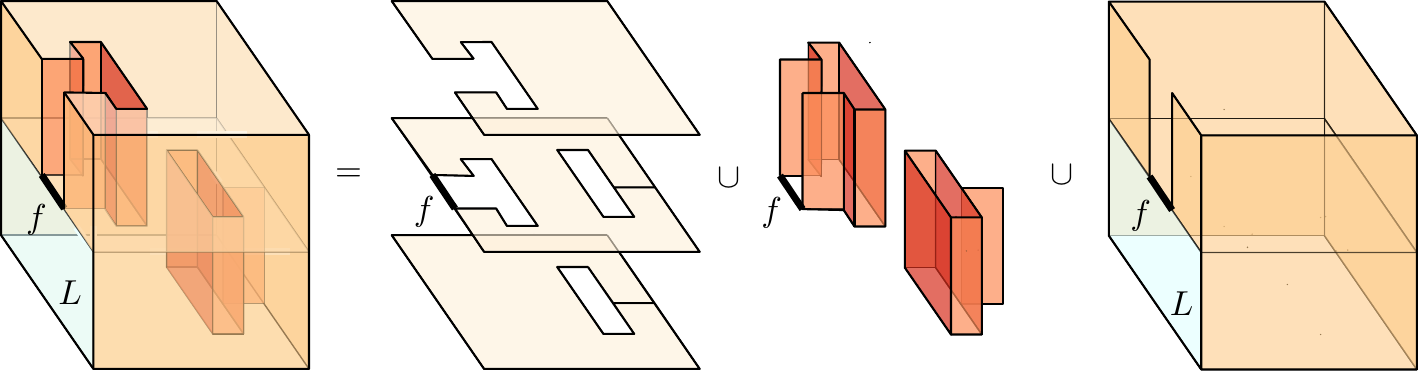}
  \end{center}
  \caption{1-house. (We are very thankful to our coauthors
  from~\cite{gpptw19} for their approval to reuse and to further modify this
  drawing.)}
  \label{f:1house}
\end{figure}

\begin{lemma}[Lemma~9 in~\cite{gpptw19}]
  \label{l:thin_1house}
  Let $H$ be a 1-house, $f$ its free edge and $L$ its lower wall. 
  In any geometric\footnote{We use geometric triangulations here because
  of the convenience of the reference to~\cite{gpptw19}---these are the
  triangulations used in~\cite{gpptw19}. By checking the proofs one can
  observe that it is possible to relax this condition to requiring only that no
  triangle of the triangulation contains two edges subdividing $f$.} triangulation of $H$, the free faces are exactly the edges that
  subdivide $f$. Moreover, $H$ collapses to any subtree of the 1-skeleton of
$H$ that is contained in $L$ and shares with the boundary of $L$ a single
  endpoint of $f$.
\end{lemma}

We will be using the lemma above only for trees of specific shape (subdivided
star). Such a tree contains a \emph{central path} (with one bend) connecting 
a vertex $a$ in $f$ (the right one) with another vertex $b$ inside $L$ but right from $f$. 
In addition it contains a \emph{splitting star} consisting of $b$ and some number
of edges emanating from $b$ to the right; see Figure~\ref{f:1house_circles}. In
addition, we will also need certain distinguished pairwise disjoint 
circles in the $1$-skeleton of some triangulation of the 1-house 
called \emph{crossing circles}. They meet $L$ in a collection of vertical
parallel segments each of which intersects the central path in exactly one
point. Then they extend to the $1$-house as in Figure~\ref{f:1house_circles}.

\begin{figure}
  \begin{center}
    \includegraphics[page=2]{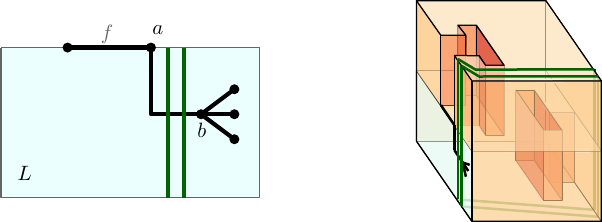}
  \end{center}
\caption{1-house with our tree and crossing circles.}
  \label{f:1house_circles}
\end{figure}

\paragraph{Encapsulating.}

In order to glue later on various gadgets together in $3$-space, it is
useful to embed a 1-house into an auxiliary 3-ball called a \emph{capsule}. In fact,
the capsule is the cube which is the convex hull of the 1-house 
as in Figure~\ref{f:1house_in_cube}. 
Note that the edge $f$, our tree and the 
crossing circles appear on the boundary of the capsule. The 1-house in this
cube is called an \emph{encapsulated 1-house}. After a homeomorphism, the
capsule can be realized as a cylinder along the central path where the crossing
circles are parallel with the bases of the cylinder; see Figure~\ref{f:1house_in_cylinder}.

\begin{figure}
  \begin{center}
    \includegraphics[page=3]{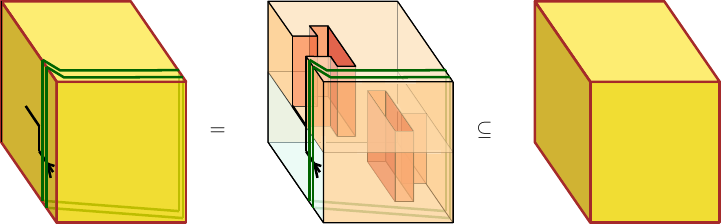}
  \end{center}
\caption{1-house with a our tree and crossing circles embedded in a cube.}
  \label{f:1house_in_cube}
\end{figure}
\begin{figure}
  \begin{center}
  \includegraphics[page=4]{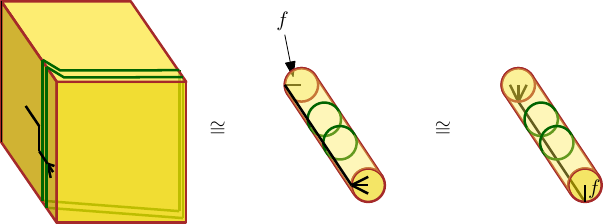}
  \end{center}
\caption{1-house with our tree and crossing circles embedded in a cube
  after a homeomorphism.}
  \label{f:1house_in_cylinder}
\end{figure}

\subsection{(Modified) turbine}
\label{ss:thin_turbine}
A gadget called turbine has been introduced by Santamar\'{\i}a-Galvis and
Woodroofe~\cite{santamaria-galvis-woodroofe21}. We will use a slightly modified
version of a turbine with 3 blades from~\cite{santamaria-galvis-woodroofe21}. 
However, in spite of the modification, we still call it just \emph{a turbine}
for simplicity of the terminology. This modification would not be necessary for
a proof of Theorem~\ref{t:collapsibility_hard}. However, it is important later
on for the proof of Theorem~\ref{t:main} when thickening the turbine. Thus we
want to be consistent. 

The turbine will be a $2$-polyhedron made from a \emph{central triangle}
and three \emph{blades}.  The central triangle is a triangle with vertices $y_1, y_2, y_3$ and 
the midpoints of the edges $m_1, m_2, m_3$ as in
Figure~\ref{f:turbine_triangle}. The barycenter of the triangle is denoted $c$.
The turbine will also contain an important subpolyhedron called \emph{central
tree}. The portion of the central tree inside the central triangle will consist
of the edges $m_1c$, $m_2c$, $m_3c$ and also one extra distinguished edge $e$
emanating from $c$ as in Figure~\ref{f:turbine_triangle}.

\begin{figure}
  \begin{center}
    \includegraphics[page=23]{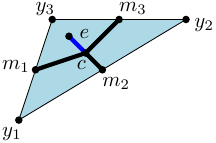}
    \caption{The central triangle of a modified thin turbine}
    \label{f:turbine_triangle}
  \end{center}
\end{figure}

Then we glue three \emph{blades} to this triangle. A single blade is the
complex depicted in 
Figure~\ref{f:blade_orthogonal}. The left picture shows
only the blade while the right one 
(partially transparent) also includes the thick path which is
the portion of the central tree in the blade. (It may also be useful to check
Figure~\ref{f:collapsing_blade_1} where the blade is gradually collapsed which
may remove doubts which faces are present.) The thick path also contains an
important distinguished edge $f_i$ (in the $i$th blade for $i \in \{1, 2,3\}$). 

\begin{figure}
  \begin{center}
    \includegraphics[page=7]{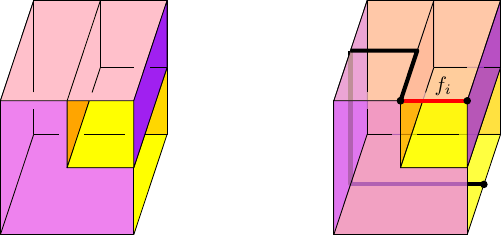}
  \end{center}
  \caption{The $i$th blade of non-thickened turbine in orthogonal scheme}
  \label{f:blade_orthogonal}
\end{figure}

The blades are glued to the central triangle as in
Figure~\ref{f:turbine_orthogonal}.
Only two blades are drawn in the picture for
simplicity but there is one blade also attached in an analogous way along the
line $y_1y_2$. After gluing all parts together, it is easy to check that the
central tree is indeed a tree. By $\pi_i$ we denote the path in this tree
between $c$ and the vertex of $f_i$ which is not a leaf of the tree.
(Therefore, $\pi_i$ does not contain $f_i$.)

The two depicted blades are intentionally drawn in a way
resembling a right angle at the vertex $y_3$.  
This will be useful for
thickening in Subsection~\ref{ss:thick_turbine}. Of course if we want the thin
turbine as a polyhedral complex embedded
in $\R^3$, then only one right angle is possible in the central triangle.
However, if we do not fix the embedding, we can locally get the right angle at
any vertex we want by an affine transform.

\begin{figure}
  \begin{center}
    \includegraphics[page=8]{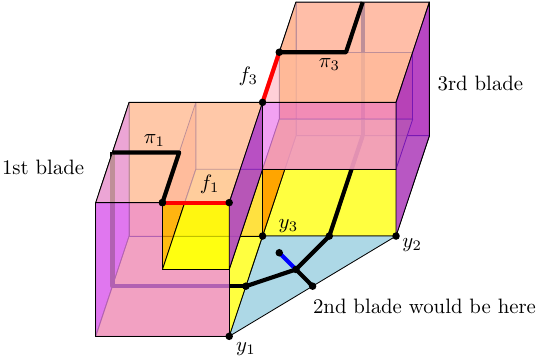}
  \end{center}
  \caption{A non-thickened variant of a turbine in orthogonal scheme (one blade
  is not drawn for simplicity of the picture).}
  \label{f:turbine_orthogonal}
\end{figure}

The following lemma is an analog of Lemma~\ref{l:thin_1house} for a turbine.
(It is also an analog of Lemma~8 in~\cite{tancer16}; however for a different
gadget.)

\begin{lemma}
  \label{l:thin_turbine}
  Let $T$ be a turbine with vertices and edges denoted as in its definition. In
    any geometric\footnote{The assumption on geometric triangulation can
    be again weakend similarly as in the case of Lemma~\ref{l:thin_1house}.} triangulation of $T$ which also restricts to a triangulation of the central
  tree, the free faces are exactly the edges that subdivide
  $f_1, f_2$ or $f_3$. Moreover such a triangulation collapses to any $1$-complex spanned by
  $e$, $\pi_1$, $\pi_2$, $\pi_3$ and any subset of $\{f_1, f_2, f_3\}$ except
  $\{f_1, f_2, f_3\}$ itself.
\end{lemma}

\begin{proof}
  The first claim regarding the free faces is straightforward to check:
  In any triangulation, 
  triangles cannot be free as the complex is $2$-dimensional. Edge which is not
  contained in $f_1, f_2$ or $f_3$ is in at least two triangles (and an edge
  contained in $f_1, f_2$ or $f_3$ is in a single triangle). Any vertex is
  contained in an edge which is in at least two triangles. (Here we use
  that the triangulation is geometric.)
  
\begin{figure}
  \begin{center}
    \includegraphics[page=47]{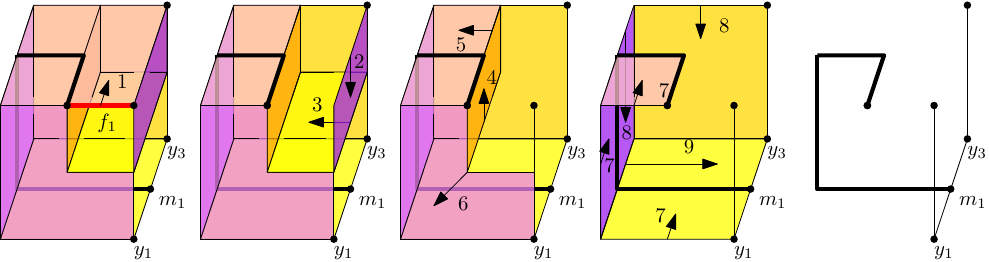}
  \end{center}
  \caption{Collapsing the first blade.}
  \label{f:collapsing_blade_1}
\end{figure}

  Then we collapse a part of the central triangle (between segments $cm_1$ and
  $cm_3$ except $e$) and the third blade by collapses
  in Figure~\ref{f:collapsing_blade_3}. The third blade is collapsed to $f_3$,
  portion of $\pi_3$ inside this blade, the segment above $y_2$ and the segment
  $m_3y_2$. 

\begin{figure}
  \begin{center}
    \includegraphics[page=48]{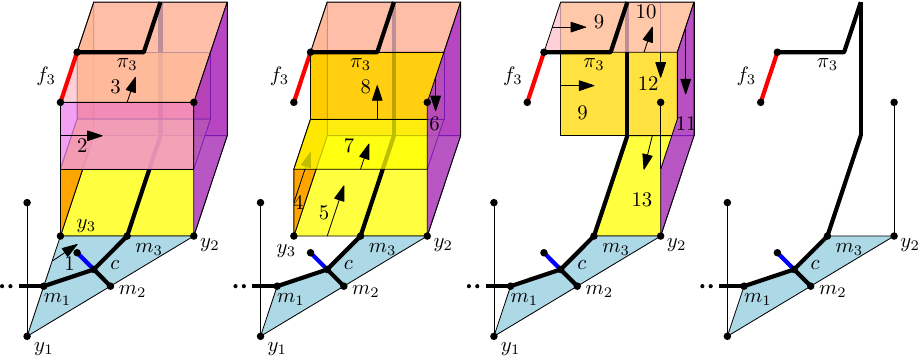}
  \end{center}
  \caption{Collapsing the third blade.}
  \label{f:collapsing_blade_3}
\end{figure}
 
  By completely analogous collapses we collapse the part of the central
  triangle between segments $cm_3$ and $cm_2$ (except avoiding $e$) 
  and then the second blade to $f_2$,
  portion of $\pi_2$ inside this blade, the segment above $y_1$ and the segment
  $y_1m_2$. After this, we collapse the segment above $y_1$ and the portion of
  the central triangle between segments $cm_1$ and $cm_2$ obtaining the central
  tree except $f_1$. If needed, we collapse $f_2$ and/or $f_3$ and we are done.
\end{proof}

\paragraph{Encapsulating.}

It is routine to check that it is possible to embed the turbine into a $3$-ball
so that the central triangle appears on boundary; up to a homeomorphism this is
as in Figure~\ref{f:turbine_embedded}. Indeed, we can add the second blade to
the drawing in Figure~\ref{f:turbine_orthogonal}  in a natural way so that after
taking the convex hull, the central tree completely appears on the boundary.
The ball containing the turbine this way is called a \emph{capsule} and the
turbine together with this ball is called an \emph{encapsulated turbine}.

\begin{figure}
  \begin{center}
  \includegraphics[page=49]{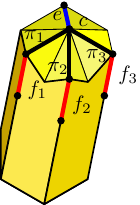}
  \end{center}
  \caption{Encapsulated turbine.}
  \label{f:turbine_embedded}
\end{figure}

\section{The construction for collapsibility}
\label{s:thin_construction}
\subsection{Planar monotone rectilinear $3$-SAT}
\label{ss:pmr3sat}

In our reduction we use the planar monotone rectilinear 3-SAT
from~\cite{deberg-khosravi10}. First, we recall the well known $3$-SAT problem and then we explain what the
adjectives `planar monotone rectilinear' mean. A \emph{literal} is a boolean
variable or its negation. A \emph{clause} is a boolean formula which is a
disjunction of literals, that is, of the form $\kappa = \ell_1 \vee \cdots \vee \ell_k$
where $\ell_1, \dots, \ell_k$ are literals. A boolean formula is a
\emph{CNF-formula}, if it is a conjunction of clauses, that is, it is of the form
$\phi = \kappa_1 \wedge \cdots \wedge \kappa_m$ where $\kappa_1, \dots, \kappa_m$ are clauses (CNF
stands for `conjunctive normal form'). An input for the 3-SAT problem is a CNF
formula $\phi$ where every clause of $\phi$ contains at most three literals.
The output states whether $\phi$ is \emph{satisfiable}, that is, whether
there is an assignment to the variables TRUE or FALSE such that the whole
formula evaluates to TRUE.

 The \emph{monotone 3-SAT} is a restricted variant of the 3-SAT problem where
 every clause is \emph{monotone}; that is, either all literals in the clause
 are variables, or all of them are negations of variables.
 The former one is called a \emph{positive} clause, the latter one is a
 \emph{negative} clause.

 The \emph{planar monotone rectilinear 3-SAT} is a further restriction of the
 monotone 3-SAT subject to the condition that the variables, clauses and their
 adjacencies can be represented via suitable drawing in the plane according to
 the following rules (see Figure~\ref{f:rectiSAT}):

\begin{figure}
  \begin{center}
    \includegraphics[page=1]{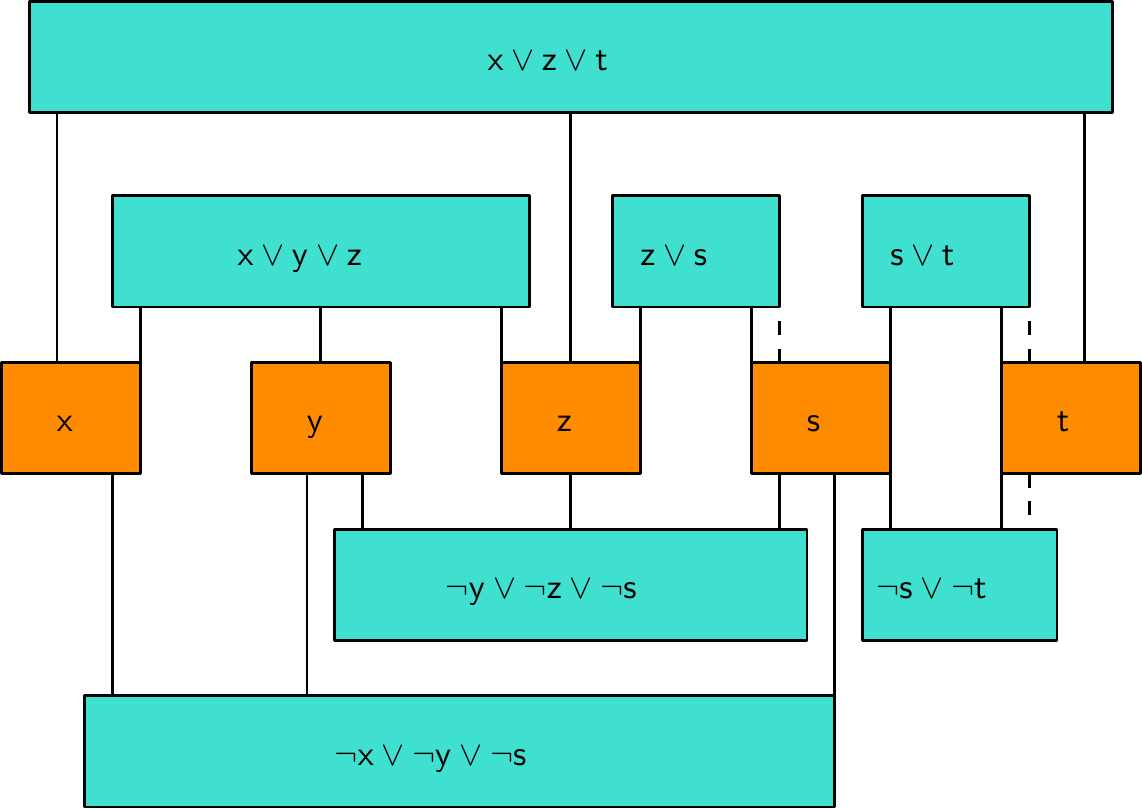}
    \caption{An example of a representation of a formula from the planar
    monotone rectilinear 3-SAT. Dashed segments appear only after the minor
    modification explained below the construction.}
    \label{f:rectiSAT}
  \end{center}
\end{figure}
\begin{itemize}
  \item
    Each variable and clause is represented by an axis-aligned rectangle (for
    simplicity) of fixed height~$1$. The rectangles are pairwise disjoint. In addition,
 the $x$-axis is the line of symmetry of the rectangles representing the
 variables. The rectangle representing a variable $\vx$ will be further denoted
    $\RR_{\vx}$; the rectangle representing a clause $\kappa$ will be denoted
    $\RR_\kappa$.
    \item A rectangle representing a \emph{positive} clause is above the
      $x$-axis. A rectangle representing a \emph{negative} clause is below the
      $x$-axis.
    \item Whenever a variable $\vx$ or its negation $\neg \vx$ is contained in a
      clause $\kappa$, then the $\RR_{\vx}$ and $\RR_\kappa$ are
      connected with a vertical segment. This vertical segment does not meet any
      other rectangle. The vertical segments are pairwise disjoint.  
\end{itemize}

In our reduction we will use that the planar monotone rectilinear 3SAT is
NP-hard assuming that the representation is a part of the input~\cite{deberg-khosravi10}.

\paragraph{Minor modification.} We emphasize that, by definition, an instance
of the planar monotone rectilinear 3SAT may also contain clauses of size 1 or
2.  However, we perform a minor modification so that
every rectangle representing a clause has exactly $3$ incoming vertical
segments. Given a clause $\kappa$ with two literals, say $\kappa =
\ell_1 \vee \ell_2$, we reach this by duplicating the vertical segment coming
from $\ell_2$. That is, between the rectangle for $\ell_2$ and the rectangle
for $\kappa$ we will have two parallel vertical segments (close to each other)
after a possible enlargement of one of the rectangles; see
Figure~\ref{f:rectiSAT}. Intuitively, this corresponds to replacing $\ell_1
\vee \ell_2$ with $\ell_1 \vee \ell_2 \vee \ell_2$ which does not affect which
assignments are satisfying.

For a clause $\kappa$ with a single literal, we could proceed similarly with three
parallel vertical segments from the rectangle of $\kappa$. On the other hand, it is
also easy to see that we can avoid clauses with a single literal without
affecting NP-hardness (and the possibility to find a suitable planar monotone
rectilinear representation).

\subsection{The template}
\label{ss:template}

Now we describe the template for our construction coming from an instance of
the planar monotone rectilinear 3SAT. For this purpose we fix a 3SAT formula $\phi$ and its planar monotone rectilinear representation. Consider an axis aligned
rectangle containing the representation of $\phi$ in the interior. We will
refer to this rectangle as the \emph{bounding box}.

As we want to preserve `above' and `below' from the planar monotone rectilinear
3SAT, we will imagine that our bounding box is situated in $3$-space so
that it is drawn on a board or displayed on the monitor. In other words, the
$x$-axis corresponds to the directions left/right, the $y$-axis corresponds to
the directions above/below, and the $z$ axis corresponds to the directions in
front of/behind. The bounding box is an axis aligned rectangle in the plane $z
= 0$. Some triangulation of the bounding box will be part of our construction. Then we will glue various gadgets always in front of the bounding box but never behind it. 

In order to glue the gadgets to the bounding box, we first create certain
template in the bounding box;  see Figure~\ref{f:template}. Then we will glue the gadgets according to the
template. 

\begin{figure}
  \begin{center}
    \includegraphics[page=2]{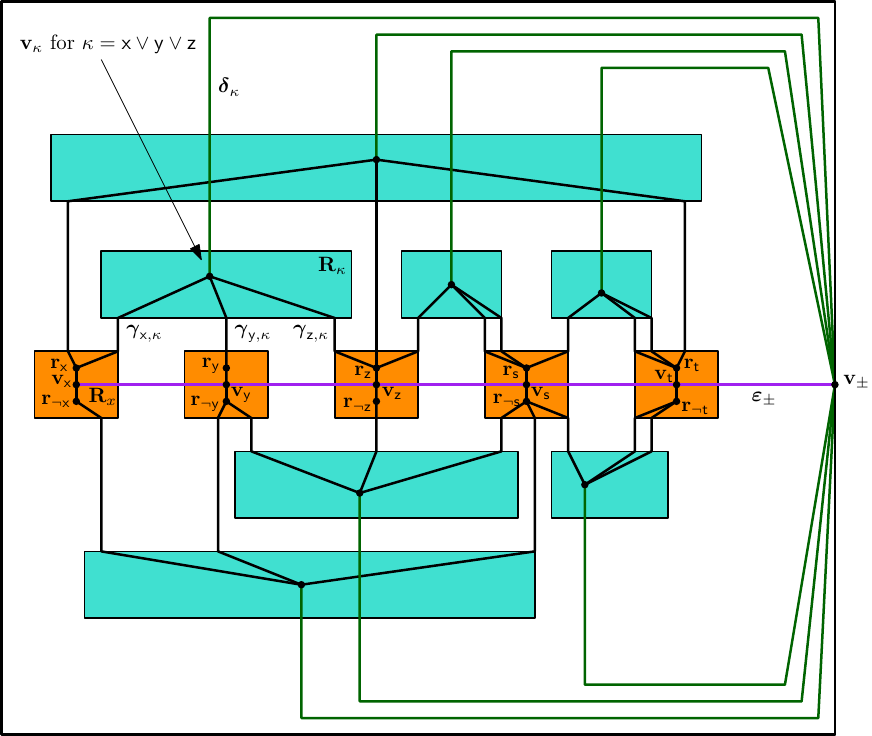}
    \caption{Creating the template for the formula from
    Figure~\ref{f:rectiSAT}.}
    \label{f:template}
  \end{center}
\end{figure}

For every variable $\vx$ we place vertices $\vv_{\vx}$, $\rr_{\vx}$ and
$\rr_{\neg \vx}$ inside $\RR_{\vx}$ so that $\vv_{\vx}$ is the midpoint of the
rectangle, $\rr_{\vx}$ is directly above $\vv_{\vx}$ and $\rr_{\neg \vx}$ is directly below it. 
 For every clause $\kappa$ we place a vertex $\vv_{\kappa}$ somewhere in the
 interior of the rectangle for $\kappa$. In addition, we require that every two
 vertices $\vv_{\kappa}$, $\vv_{\kappa'}$ for distinct clauses $\kappa$ and
 $\kappa'$ have distinct $x$-coordinates and distinct $y$-coordinates. 
On the
 intersection of the $x$-axis and the right side of the bounding box, we place
 a vertex $\vv_\pm$.

Now we define a few auxiliary piecewise linear curves. For every variable $\vx$,
we connect $\vv_{\vx}$ with $\rr_{\vx}$ and $\rr_{\neg \vx}$ with straight segments. 
For every pair $(\ell,\kappa)$
such that $\ell$ is a literal and $\kappa$ is a clause containing $\ell$
we define a piecewise linear curve $\ggamma_{\ell, \kappa}$ with at most two bends
connecting $\rr_{\ell}$ and $\vv_\kappa$. If $x$ is the variable of $\ell$ (that is $\ell
\in \{\vx, \neg \vx\}$) then we require that $\ggamma_{\ell, \kappa}$ contains the vertical segment connecting
the rectangles $\RR_{\vx}$ and $\RR_{\kappa}$ (from the rectilinear representation) and then it
connects the endpoints of this segment inside the rectangles with $\rr_\ell$ and
$\vv_{\kappa}$ by another segments. If there are more vertical segments connecting the
rectangles of $\vx$ and $\kappa$ (which may happen after our earlier `minor
modification'), then we have one such curve $\ggamma_{\ell,\kappa}$ for each vertical
segment. With slight abuse of the notation, we will not distinguish these
curves and we will refer to any of them as $\ggamma_{\ell,\kappa}$. 

For every clause $\kappa$ we also define another piecewise linear curve
$\delta_\kappa$ with at most two bends connecting $\vv_\kappa$ and $\vv_\pm$.
The curve $\ddelta_\kappa$ consists of three segments. Now suppose that $\kappa$ is positive. 
The first segment is a vertical emanating up from $\vv_\kappa$ until it almost
reaches the boundary of the bounding box. The second segment emanates right from the earlier
position until it almost reaches the right side of the bounding box. Finally,
the third segments connects the current endpoint with $\vv_\pm$. If $\kappa$ is
negative then the description is analogous with the exception that the first
segment emanates down. In addition, the curves $\ddelta_\kappa$ are mutually
positioned so that they do not intersect each other except in $\vv_\pm$; see
Figure~\ref{f:template}. On the other hand, the first segment of
$\ddelta_\kappa$ may intersect $\RR_{\kappa'}$ for some other clause $\kappa'$
and some curves
$\ggamma_{\ell,\kappa'}$ inside $\RR_{\kappa'}$. However, note that
$\ddelta_\kappa$ does not pass through any $\vv_{\kappa'}$ as the
$x$-coordinates of $\vv_{\kappa}$ and $\vv_{\kappa'}$ are distinct.

Finally, we add a 
segment $\vvarepsilon_{\pm}$ connecting $\vv_\pm$ and the rightmost
point among the points
$\vv_{\vx}$ where $\vx$ is a variable.

The role of the vertices and the curves defined above is that we will glue the
other gadgets to the template exactly in these vertices and along these curves.

\subsection{Gluing the gadgets}
\label{ss:gluing_thin}

\paragraph{Step 1: the variable gadget.} 
For every variable $\vx$ we take a copy $\cV_{\vx}$ of the bipyramid called the
\emph{variable gadget}. Denote the vertices of $\cV_{\vx}$ by $\uu_{\vx}$,
$\vv_{\vx}$, $\ww_{\vx}$, $\aa_{\vx}$ and $\aa_{\neg \vx}$ so that they correspond in this order to the 
vertices $u$, $v$, $w$, $a^+$ and $a^-$ in the original definition of the bipyramid. Note
that $\vv_x$ is already a vertex in the template, thus we also require that
$\cV_{\vx}$ intersects the template exactly in this vertex.

In order to embed the final construction into $3$-space, it already pays
off to specify the geometric position of the bipyramid: We place the bipyramid
in a sufficiently small neighborhood of the vertex $\vv_{\vx}$ so that it is fully in
front of the template, that is, in the halfspace given by $z \geq 0$. We also
assume that the vertices $\uu_{\vx}$ and $\ww_{\vx}$ lie in the plane given by
$y=0$ and $\uu_{\vx}$ is on the right from $\ww_{\vx}$. Finally, $\aa_{\vx}$ is above the
plane $y=0$ and $\aa_{\neg \vx}$ below this plane; see
Figure~\ref{f:variable_gadget_top}.

\begin{figure}
  \begin{center}
    \includegraphics[page=10]{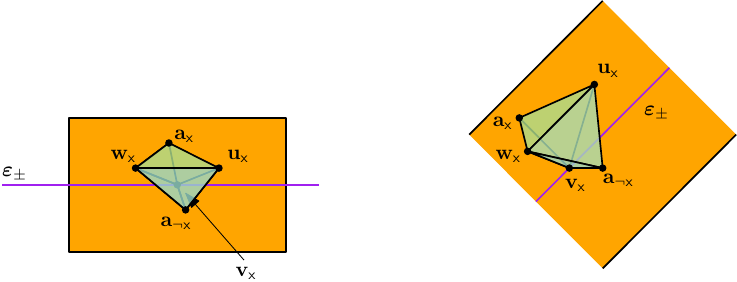}
    \caption{The variable gadget in front of the template. Left: a viewpoint
    corresponding to our conventions regarding the coordinate system. Right: a
    slightly rotated viewpoint such that all vertices are visible for a better
    visualization.}
    \label{f:variable_gadget_top}
  \end{center}
\end{figure}

\paragraph{Step 2: the clause gadget.}

 For every clause $\kappa$ we take a copy $\cC_\kappa$ of the turbine called the
 \emph{clause gadget}. We specify both the attachment of $\cC_\kappa$
 to the template and its embedding into $\R^3$. It pays off to consider
 $\cC_\kappa$ as a part of an encapsulated turbine and we will describe an
 embedding of the capsule.

 First, we explain the construction for a positive clause $\kappa = \vx
 \vee \vy \vee \vz$.
 According to our earlier minor modification, we allow that two of the
 variables coincide. Let us assume that in the planar rectilinear
 representation of our formula, the variables $\vx, \vy, \vz$ appear from left to right in this order. 

 We take a copy of an encapsulated turbine and name its vertices as in
 Figure~\ref{f:clause_gadget}, left. Thus the vertex $c$ in the original
 notation for turbine corresponds to $\vv_\kappa$; the edge $e$ corresponds to
 $\vv_\kappa  \uu_\kappa$; and the edges $f_1, f_2$ and $f_3$ correspond to
 $\vv_{\vz,\kappa}\ww_{\vz,\kappa}$, $\vv_{\vy,\kappa}\ww_{\vy,\kappa}$ and
 $\vv_{\vx,\kappa}\ww_{\vx,\kappa}$ respectively. (Compare
 Figures~\ref{f:turbine_embedded} and~\ref{f:clause_gadget}, left.) The exact
 triangulation is not so important; we only require that it is
 geometric; the edges $e, f_1,  f_2$ and $f_3$ are not subdivided and that the central tree of the turbine
 appears in the 1-skeleton (possibly subdivided in paths $\pi_1, \pi_2,
 \pi_3$). We also require that each triangle meets the central tree
 either in an edge, a vertex or it does not meet it. Note that this can always
 be achieved by subdividing problematic triangles without subdividing any edges
 of the central tree. This last assumption will guarantee that we get a
 simplicial complex when gluing the gadgets together.
 We always take the same copy of the clause gadget, thus its size is
 constant.
If some two of the variables coincide, this
 may create a notational conflict---two distinct vertices may have the same
 notation. However, with a slight abuse of notation, we ignore this notational
 conflict similarly as we did for $\ggamma_{\ell,\kappa}$. 

 \begin{figure}
  \begin{center}
    \includegraphics[page=11]{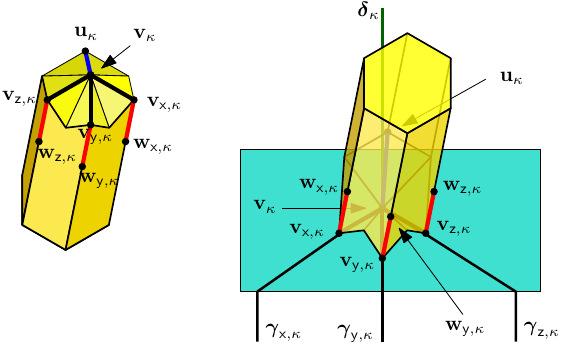}
    \caption{The clause gadget.}
   \label{f:clause_gadget}
  \end{center}
\end{figure}

 Now we rotate our encapsulated turbine upside down and we identify the vertex
 $\vv_\kappa$ with the vertex $\vv_\kappa$ of the template.
 Next, we position the vertex $\vv_{\vx,\kappa}$ on the
 curve $\ggamma_{\vx,\kappa}$ somewhere in the rectangle of $\kappa$ and we identify the
 (possibly subdivided) segment $\vv_\kappa \vv_{\vx,\kappa}$ with the corresponding segment (subcurve
 of $\ggamma_{\vx,\kappa}$)
 on the template. We proceed analogously with $\vv_{\vy,\kappa}$ and
 $\vv_{z,\kappa}$. We place $\vv_{\vx,\kappa}$, $\vv_{\vy,\kappa}$ and
 $\vv_{z,\kappa}$ close enough to $\vv_\kappa$ so that the (encapsulated) turbine does not
 intersect any curve $\ddelta_{\kappa'}$.
 Note
 that after these identifications, the encapsulated turbine and the template
 intersect exactly in the tripod formed by segments $\vv_\kappa
 \vv_{\vx,\kappa}$, $\vv_\kappa  \vv_{\vy,\kappa}$ and $\vv_\kappa
 \vv_{\vz,\kappa}$. This gives an embedding of the encapsulated
 turbine (and therefore of $\cC_\kappa$ as well) in $\R^3$ in front of the
 template. See Figure~\ref{f:clause_gadget}, right.

 In case  of notational conflict, say that $\vz = \vy$, we have two curves
 $\ggamma_{\vy,\kappa}$ and two vertices $\vv_{\vy,\kappa}$. Then we perform the identifications
 in such a way that the overall identification can be performed in 
 $3$-space (that is the cyclic orders of the identified tripods are opposite on
 the capsule and the template).

 Finally, if $\kappa = \neg \vx \vee \neg \vy \vee \neg \vz$ is a negative
 clause, we perform the same construction in a 
 mirror-symmetric fashion. The newly introduced vertices are called
 $\vv_{\neg \vx,\kappa}, \dots, \ww_{\neg \vz,\kappa}$.

\paragraph{Step 3: the splitter house.}

Another gadget we need is the splitter house $\cS_\ell$ defined for every
literal $\ell$. Let $\vx$ be the variable of the literal $\ell$. 
We take a copy of encapsulated 1-house with no crossing circles so that the
number of edges of the splitting star equals to the number of paths
$\ggamma_{\ell, \kappa}$ emanating from $\rr_\ell$. We take a \emph{geometric} triangulation such
that the free edge is not subdivided and the size of this triangulation is
linear in the number of edges of the splitting star. Similarly as for the
clause gadget, we also assume that each triangle of the triangulation meets the
union of $f$, the central path and the splitting star in an edge, a vertex, or it does not meet it. (We impose analogous
restrictions for other 1-houses used later on adding also the 
crossing circles to the union.  The size is linear either in the
number of edges in the splitting star or in the number of crossing
circles---these numbers will not be both more than one simultaneously.)
We identify the central path of the encapsulated 1-house with the segment $\vv_{\vx} \rr_\ell$ of the template. Then we identify the free
edge of encapsulated 1-house with the edge $\vv_{\vx} \aa_{\ell}$ of the variable
gadget $\cV_x$. For the actual embedding we embed this encapsulated 1-house in
front of the template in a small neighborhood of the segment $\vv_{\vx}
\rr_\ell$ (this neighborhood includes the variable
gadget considering the variable gadget sufficiently small) so that we do not
introduce any unwanted intersections with other gadgets nor their capsules; see
Figure~\ref{f:splitter}. Then the splitter
house $\cS_\ell$ is 1-house inside this encapsulated embedded 1-house (not the
whole capsule).

\begin{figure}
  \begin{center}
    \includegraphics[page=12]{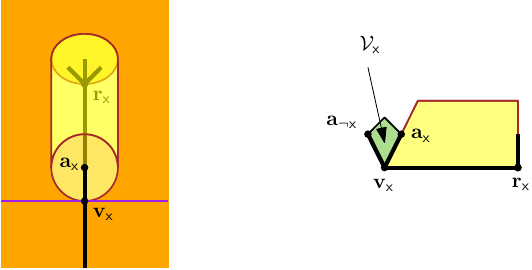}
    \caption{Left: the splitter house for $\ell = \vx$ (with the capsule) embedded in front of the template;
    compare with Figures~\ref{f:1house_in_cylinder}, right and~\ref{f:template}.
    The base of the cylinder which contains $\aa_{\vx}$ leans slightly backwards so
    that there is enough space for the variable gadget. Right: Cut through the
    plane containing $\aa_{\vx}\rr_{\vx}\vv_{\vx}$ revealing the attachment to the variable
    gadget.}
    \label{f:splitter}
  \end{center}
\end{figure}

\paragraph{Step 4: the incoming house.}

For every curve $\ggamma_{\ell,\kappa}$ we define one \emph{incoming house}
$\cI_{\ell,\kappa}$ (possibly with multiplicities if there are multiple curves 
$\ggamma_{\ell,\kappa}$).\footnote{The name `incoming' is considered relatively
to the clause gadget. Intuitively, this house transfers the information (from
variables) into clause gadgets. Later on, we will also consider `outgoing'
house that passes the information from the clause gadgets further.}
We again take a
copy of encapsulated 1-house. This time this house contains as many
crossing circles as is the number of crossings of our $\ggamma_{\ell,\kappa}$ with
curves $\ddelta_{\kappa'}$ (for arbitrary $\kappa'$). On the other hand, the
splitting star
has only one edge. Now we identify the central path of the encapsulated house with
our curve $\ggamma_{\ell,\kappa}$ except the part which is already identified with
the clause gadget. Then we identify the splitting star (edge) with the edge
$\vv_{\ell,\kappa} \ww_{\ell,\kappa}$ of the clause gadget. Finally, we identify the free edge
of this encapsulated house with $k$th edge of the splitting star of the
splitter house where $k$ is the order of $\ggamma_{\ell,\kappa}$ among the curves of
this type emanating from $\ww_\ell$. Then we embed everything in front of the
template in a small neighborhood of $\ggamma_{\ell, \kappa}$ so that there are no
unwanted intersections with other gadgets nor their capsules. (So $k$ above was
chosen in such a way that this embedding is possible.) Consult
Figure~\ref{f:incoming} in order to visualize that such an embedding is
possible. Finally, we also position the crossing circles so that they intersect
$\ggamma_{\ell,\kappa}$ exactly in the points where some $\ddelta_{\kappa'}$ meets
$\ggamma_{\ell,\kappa}$ in the template.

\begin{figure}
  \begin{center}
    \includegraphics[page=13]{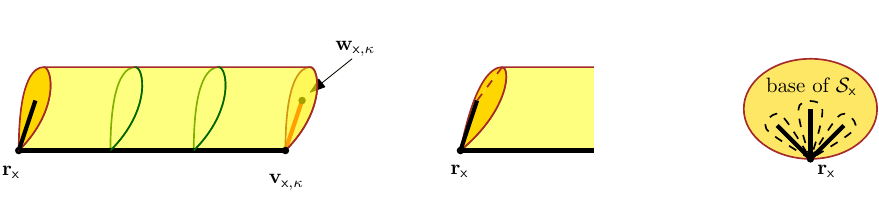}
    
    \caption{This time we realize encapsulated $\cI_{\ell, \kappa}$ (for $\ell
    = \vx$) as a prism over
    a sufficiently thin drop-shaped basis as in the left picture. After this we let the basis
    containing $\rr_{\vx}$ to lean slightly backwards except that the free edge of
    the house (the thick edge containing $\rr_{\vx}$) keeps its position. This is
    depicted in the middle picture. These two adjustments provide enough space
    to connect $\cI_{\ell, \kappa}$ to the base of the capsule of the 
    splitter house $\cS_{\vx}$ without introducing unwanted overlaps. This base is
    drawn on the right picture (compare with Figure~\ref{f:splitter}, left).
    For attachment of $\cI_{\ell, \kappa}$ to the clause gadget, we perform an
    analogous adjustment also on the base containing the edge
    $\vv_{\vx,\kappa}\ww_{\vx,\kappa}$.
    }
    \label{f:incoming}
  \end{center}
\end{figure}

\paragraph{Step 5: the outgoing house.}
For every clause $\kappa$ we define an \emph{outgoing house} $\cO_\kappa$.

We take a copy of encapsulated 1-house without crossing circles and with
only one edge in the splitting star. We identify the free edge of this 
1-house with the distinguished edge of the turbine $\cC_\kappa$, that is, with the edge
$\vv_\kappa \uu_\kappa$ so that $\vv_\kappa$ also belongs to the central path
of the 1-house.
Now we aim to glue the central path of the (encapsulated) 1-house to the
path $\ddelta_\kappa$ on the template so that the 1-house is in front of the
template. This is, however, problematic near to points where
$\ddelta_\kappa$ possibly crosses some of the paths $\ggamma_{\ell,\kappa'}$. Therefore, we
now explain how to resolve such a crossing.

Let $p$ be (temporarily) the point where $\ddelta_\kappa$ and
$\ggamma_{\ell,\kappa'}$ cross. Let $C$ be (temporarily) the crossing circle of
$\cI_{\ell, \kappa'}$ which contains $p$. For simplicity of the description assume that there is no bend on
$\ddelta_\kappa$ nor $\ggamma_{\ell,\kappa'}$ in $p$; $C$ is fully contained in the plane,
temporarily denoted $\rho$, perpendicular to the template containing (locally near $p$) a segment of
$\ddelta_\kappa$; and that $C$ is axis-symmetric around the line perpendicular to the
template passing through $p$.

We glue to the (current) construction two triangles (temporarily denoted) $\tau_1$ and $\tau_2$.
Each of them is inside $\rho$ and $p$ is the vertex of each of them. One side,
containing $p$ of $\tau_i$ is glued to a (small) subsegment (temporarily
denoted) $pq_i$ of $C$ so
that this subsegment is different for the two triangles. The other side is
glued to a subsegment (temporarily denoted) $pq_i'$ of $\ddelta_\kappa$ so that $pq_i$
projects to $pq_i'$ in the orthogonal projection to the template; see
Figure~\ref{f:extra_triangles}, left.

\begin{figure}
  \begin{center}
    \includegraphics[page=14]{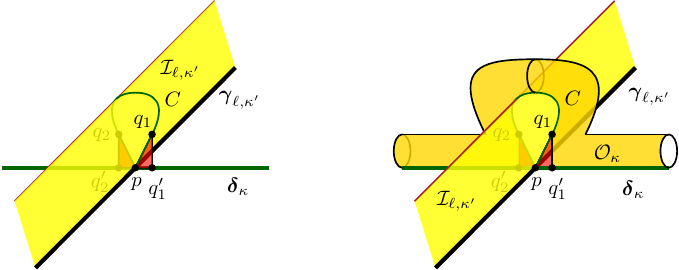}
    \caption{Resolving the crossings; the picture is rotated for simpler
    visualization.}
    \label{f:extra_triangles}
  \end{center}
\end{figure}

Now, we glue our encapsulated 1-house along $\ddelta_\kappa$ except that near
each such crossing we skip the segment $q'_1q_2'$ and instead we glue along the
path formed by segments $q'_1q_1$, $q'_2q_2$ and the portion of $C$ which is
not glued to any of the triangles.

Finally, note that the splitting star of each such 1-house consists of a
single edge with endpoint $\vv_\pm$.
We glue together the splitting star edges of $\cO_\kappa$ for all clauses
$\kappa$ and we call the second shared endpoint $\ww_\pm$.

\paragraph{Step 6: the blocking houses.}

Finally, we glue $n+1$ (encapsulated) 1-houses $\cB_0, \dots, \cB_n$ along the path
$\vvarepsilon_\pm$ and the boundaries
of the triangles $\uu_{\vx}\vv_{\vx}\ww_{\vx}$. Here $n$ is the number of
variables. For purpose
of this construction, let us assume that the variables are $\vx_1, \dots, \vx_n$
in order from left to right.
Each house we consider is again without crossing circles and with only one edge
in the splitting star.

We glue the house $\cB_0$ along the path formed by two segments
$\vv_{\vx_1}\ww_{\vx_1}$ and $\ww_{\vx_1}\uu_{\vx_1}$ 
in such a way that the free edge contains $\uu_{\vx_1}$ and
the single edge of the splitting star contains $\vv_{\vx_1}$. 

For $i \in \{1, \dots, n-1\}$ we glue the house $\cB_i$ along the path formed by
the segments $\uu_{\vx_i}\vv_{\vx_i}$, $\vv_{\vx_i}\vv_{\vx_{i+1}}$,
$\vv_{\vx_{i+1}}\ww_{\vx_{i+1}}$ and $\ww_{\vx_{i+1}}\uu_{\vx_{i+1}}$ in such a way that the free edge contains
$\uu_{\vx_{i+1}}$ and the single edge of the splitting star contains $\uu_{\vx_i}$.

The final house we add is $\cB_n$ along the path formed by the segments
$\uu_{\vx_n}\vv_{\vx_n}$ and $\vv_{\vx_n}\vv_{\pm}$.

Finally, for every $i \in \{1, \dots, n\}$ we identify the free edge of
$\cB_{i-1}$ with the splitting star edge of $\cB_i$ and we also identify the free edge of
$\cB_n$ with the edge $\vv_\pm \ww_\pm$. 

It is routine to check that these houses can be realized in close neighborhoods
of the paths to which they are glued so that we get an embedding in 
$3$-space in front of the template.

\paragraph{Triangulating the template.} We emphasize that the template is a
part of the construction. We triangulate it suitably so that the aforementioned
gluings are possible. (For example, if we glue a certain path to some edge of
the template, it is necessary to subdivide that edge.)
We also keep the size of the triangulation comparable with the sum of sizes of
individual gadgets getting the whole construction of polynomial size.

This finishes a construction of a complex $\cK'_\phi$ as described in the
introduction.

\section{Correctness of the reduction}
\label{s:correctness}

\subsection{Satisfiability implies collapsibility}
\label{ss:sat=>col}

Assume that $\phi$ is satisfiable and fix a satisfying assignment. We will show
that $\cK'_\phi$ is collapsible.

First, for a variable $\vx$ which is assigned TRUE, collapse the variable gadget
$\cV_x$ to a subcomplex formed by the edge $\vv_{\vx} \aa_{\vx}$ and triangles
$\uu_{\vx}\vv_{\vx}\aa_{\neg \vx}$, $\uu_{\vx}\ww_{\vx}\aa_{\neg \vx}$ and
$\vv_{\vx}\ww_{\vx}\aa_{\neg \vx}$. 
This is possible by Lemma~\ref{l:bipyramid}(i). (Note that
$\cV_x$ is attached to other gadgets only in the edges $\vv_{\vx} \aa_{\vx}$,
$\vv_{\vx} \aa_{\neg \vx}$, $\uu_{\vx}\vv_{\vx}$, $\uu_{\vx}\ww_{\vx}$ and
$\vv_{\vx}\ww_{\vx}$ thus the other gadgets do not block such a
collapse.) This makes the edge $\vv_{\vx} \aa_{\vx}$ free in the remaining complex.
Similarly, if $\xx$ is assigned FALSE, collapse $\cV_{\vx}$ to a subcomplex formed by
the edge $\vv_{\vx} \aa_{\neg \vx}$ and triangles 
$\uu_{\vx}\vv_{\vx}\aa_{\vx}$, $\uu_{\vx}\ww_{\vx}\aa_{\vx}$ and
$\vv_{\vx}\ww_{\vx}\aa_{\vx}$ using Lemma~\ref{l:bipyramid}(ii), This makes the
edge $\vv_{\vx} \aa_{\neg \vx}$ free.

Now assume that $\ell$ is a literal assigned TRUE. That is, either $\ell = \vx$
for $\vx$ assigned TRUE or $\ell = \neg \vx$ for $\vx$ assigned FALSE where $\vx$ is
the variable of this literal. As $\vv_{\vx} \aa_\ell$ is a free edge now, we can
collapse the splitter house $\cS_\ell$ to its central path and splitting star
which is possible due to Lemma~\ref{l:thin_1house}. (Note that at this moment,
$\cS_\ell$ meets the remainder of the complex exactly in the 
central path and splitting star, thus these collapses are not blocked by
anything else.) This makes the edges of the splitting star free in the remaining complex.

Now we would like to collapse the incoming houses. However, they may be blocked
by some outgoing houses. We explain how to make the collapses for positive
clauses. The negative clauses are analogous.

We order the positive clauses $\kappa$ by according to the $y$-coordinate of $\vv_\kappa$
starting with the lowest $y$-coordinate. (Here we use that the $y$-coordinates
are pairwise distinct.) We take a positive clause $\kappa$ and we inductively assume
that for every preceding clause $\kappa'$ the outgoing house $\cO_{\kappa'}$ has been
already collapsed to its central path and splitting star. This means that for any literal $\ell$
in $\kappa$ the incoming house $\cI_{\ell, \kappa}$ is not blocked by any
outgoing house. (There is still a pair of triangles attached to the crossing circles of 
$\cI_{\ell, \kappa}$; however we immediately collapse these triangles to the crossing circles as
each of them has a free edge.) 

As we have started with a satisfying assignment, there is a positive literal
$\ell$ in $\kappa$. This allows to collapse the incoming house
$\cI_{\ell, \kappa}$ to its central path and the single splitting star edge
$\vv_{\ell,\kappa}
\ww_{\ell, \kappa}$ of the incoming house via Lemma~\ref{l:thin_1house}. (This
edge is also an edge of $\cC_\kappa$).
Subsequently, this makes the edge $\vv_{\ell,\kappa} \ww_{\ell, \kappa}$ free. Therefore, we
can collapse the clause gadget $\cC_\kappa$ through $\vv_{\ell,\kappa}
\ww_{\ell, \kappa}$ to the remainder of its central tree via
Lemma~\ref{l:thin_turbine}. This makes the edge
$\vv_\kappa \uu_\kappa$ free. Thus we can collapse the outgoing house
$\cO_\kappa$ to its central path and the splitting star edge $\vv_\pm\ww_\pm$ (via Lemma~\ref{l:thin_1house}) 
verifying our induction assumption.

We also perform analogous collapses on negative clauses.

After these steps, the edge $\vv_\pm\ww_\pm$ becomes free. Thus we can
collapse the blocker houses $\cB_n, \dots, \cB_0$ in this
order (repeatedly using Lemma~\ref{l:thin_turbine}). This makes edges
$\uu_{\vx}\vv_{\vx}$, $\uu_{\vx}\ww_{\vx}$ and $\vv_{\vx}\ww_{\vx}$ free for every variable
$\vx$. This allows to collapse the remainder of the variable gadget
(consisting of three triangles and a pending edge) to the edge
$\vv_{\vx} \aa_{\neg \vx}$ if $\vx$ was assigned TRUE, or to the edge
$\vv_{\vx} \aa_{\vx}$ if $\vx$ was assigned FALSE. 

Consequently, we can collapse the remaining splitter houses (for literals
assigned FALSE) to their splitting stars and then the remaining incoming houses
(via Lemma~\ref{l:thin_1house}). As all the outgoing houses have been already
collapsed, the incoming houses are not blocked by them. Thus we get the
template after this step. As the template is a triangulated disk, we collapse
it to a point.

\subsection{Collapsibility implies satisfiability}
\label{ss:col=>sat}

Assume that we have a collapsing of our complex $\cK'_\phi$. Let us write $\sigma \prec
\tau$ if the simplex $\sigma$ is collapsed before $\tau$ in our collapsing.

Because $\uu_\kappa\vv_\kappa$ is the only free face of an outgoing house $\cO_\kappa$
via Lemma~\ref{l:thin_1house} and $\vv_\pm\ww_\pm$ is contained in this house, we get
$\uu_\kappa\vv_\kappa \prec \vv_\pm\ww_\pm$ for every clause $\kappa$.

Because $\uu_\kappa\vv_\kappa$ belongs to the clause gadget $\cC_\kappa$,
Lemma~\ref{l:thin_turbine} gives that each clause $\kappa$ contains a 
literal $\ell$ with $\vv_{\ell, \kappa} \ww_{\ell,\kappa} \prec
\uu_\kappa\vv_\kappa$. (The only free edges of $\cC_\kappa$ are the edges
$\vv_{\ell, \kappa} \ww_{\ell,\kappa}$ for some $\ell$.
This will be a
basis of our satisfying assignment: If $\vv_{\ell, \kappa} \ww_{\ell,\kappa} \prec \uu_\kappa\vv_\kappa$ we
set $\ell$ to TRUE (that is, if $\ell = \vx$, then we set $\vx$ to TRUE and if
$\ell = \neg \vx$, we set $\vx$ to FALSE). If $\vy$ is a variable which did not get
any assignment from this rule, we set it arbitrarily TRUE or FALSE. As soon as
we verify that there are no conflicts, that is, no literal has been assigned
both TRUE and FALSE, we get a satisfying assignment as each clause $\kappa$ contains
a literal $\ell$ with $\vv_{\ell, \kappa} \ww_{\ell,\kappa} \prec
\uu_\kappa\vv_\kappa$.

Thus it remains to check that there are no conflicts. For contradiction assume
that $\ell$ has been set both TRUE and FALSE. Without loss of generality $\ell
= \vx$ for some variable $\vx$ (otherwise we swap $\ell$ with $\neg \ell$). That is
for some clause $\kappa$ we get $\vv_{\vx, \kappa} \ww_{\vx,\kappa} \prec
\uu_\kappa\vv_\kappa$ and for another clause $\kappa'$ we get $\vv_{\neg \vx,
\kappa'} \ww_{\neg x,\kappa'} \prec \uu_{\kappa'}\vv_{\kappa'}$.

Let $\ee_{\vx}$ be the edge shared by $\cI_{\vx,\kappa}$ and $\cS_{\vx}$. This
is the only free edge of $\cI_{\vx,\kappa}$ via Lemma~\ref{l:thin_1house}, thus
we get $\ee_{\vx} \prec \vv_{\vx, \kappa} \ww_{\vx,\kappa}$. By a similar
argument (on splitter houses), we get $\vv_{\vx}\aa_{\vx} \prec \ee_{\vx}$. 
If we put earlier inequalities together, we get $\vv_{\vx}\aa_{\vx}
\prec \vv_{\pm}\ww_{\pm}$. Similarly, we deduce $\vv_{\vx}\aa_{\neg \vx} \prec
\vv_{\pm}\ww_{\pm}$. Considering the free edges of the blocker houses, we also
get that $\vv_{\vx}\aa_{\vx}$ and $\vv_{\vx}\aa_{\neg \vx}$ are collapsed
before any edge contained in any of the blocker houses. (Note that the free
edge of $\cB_i$ has to be collapsed before any edge of $\cB_{i-1}$.)

This gives a contradiction with Lemma~\ref{l:bipyramid}(iii): 
The constraint complex of $\cV_\vx$ in our complex consists of edges
$\uu_{\vx}\vv_{\vx}$, $\uu_{\vx}\ww_{\vx}$, $\vv_{\vx}\ww_{\vx}$ (where $\cV_\vx$ is
glued to some blocker) and edges $\vv_{\vx}\aa_{\vx}$ and $\vv_{\vx}\aa_{\neg
\vx}$ (where $\cV_\vx$ is glued to some splitter). The collapses removing $\vv_{\vx}\aa_{\vx}$ and $\vv_{\vx}\aa_{\neg
\vx}$ remove a triangle outside $\cV_\vx$ because $\vv_{\vx}\aa_{\vx}$ and $\vv_{\vx}\aa_{\neg
\vx}$ are the unique free edges of the corresponding splitters (there is no
other way how to collapse these splitters). Finally both $\vv_{\vx}\aa_{\vx}$ and $\vv_{\vx}\aa_{\neg
\vx}$ are collapsed before $\uu_{\vx}\vv_{\vx}$, $\uu_{\vx}\ww_{\vx}$,
$\vv_{\vx}\ww_{\vx}$ because they are contained in some blocker house(s). 

\section{Shelling $d$-balls}
\label{s:shelling_balls}

From now on we focus on the proof of Theorem~\ref{t:main}. 
First we need a
little bit more background on shellability than what is given in
Section~\ref{s:prelim}. For shelling up and shelling down = shelling, we refer
to Definitions~\ref{d:shell_up} and~\ref{d:shell_down}. We also refer to
Section~\ref{s:prelim} for the notation $K[\vartheta_1, \dots, \vartheta_m]$.

\paragraph{Shelling in PL sense.}

We will occasionally use shelling in PL sense that allows to remove a larger
subcomplex than just a $d$-simplex in a single step. We use the definition
from~\cite{rourke-sanderson72} (above Lemma~3.25) with an exception that we do
not require to shell only a manifold. It is natural to state the definition and
the auxiliary claims we need in general dimension $d$ which requires
a careful distinction of balls and PL balls. However, we remark that we will
apply our claims only in dimension at most $3$ (typically $3$) where balls and
PL balls coincide (as well as manifolds and PL manifolds). Thus the distinction
of PL and non PL structures can be ignored by the reader if restricted to $d
\leq 3$.

\begin{definition}
  Assume that $K$ is a pure $d$-dimensional simplicial complex decomposed as $K
  = L \cup B$ where $B$ is a PL $d$-ball and $L \cap B$ is a PL $(d-1)$-ball. Then
  $K$ \emph{elementarily shells in PL sense} to $L$. Next we say that $K$
  \emph{shells in PL sense} to a subcomplex $M$, if there is a sequence of
  complexes $K_0=K, K_1, \dots, K_s = M$ such that $K_{i-1}$ elementarily shells in
  PL sense to $K_i$ for $i \in \{1, \dots, s\}$.
\end{definition}

We explicitly remark that we want to follow the definition of shelling
in~\cite{rourke-sanderson72} as a special case of collapse. This means that
even not every (standard simplicial) shelling is a shelling in PL
sense---only those steps that do not modify the homotopy type are allowed.

\paragraph{Shelling and balls.}

The following lemma is a variant of Proposition~2.4(i) in~\cite{ziegler98} for PL
shellings (see also Proposition~4.7.22
in~\cite{bjorner-lasvergnas-sturmfels-white-ziegler99} and the references above
it). A \emph{$d$-pseudomanifold} is a pure $d$-dimensional simplicial complex
such that every $(d-1)$-face is contained in one or two $d$-faces. (In this
definition, we allow boundary.)

\begin{lemma}
  \label{l:pseudomanifold}
  Assume that $M$ is a $d$-pseudomanifold. If $M$ shells
  in PL sense to a PL $d$-ball then $M$ is a PL $d$-ball.
\end{lemma}

\begin{proof}
Let $B$ be the PL $d$-ball to which $M$ shells. Let $M_0=M, M_1, \dots, M_s = B$ be
  a sequence of subcomplexes of $M$ such that $M_{i-1}$ elementarily shells in
  PL sense to $M_i$ for $i \in \{1, \dots, s\}$. By the assumptions $M_s$ is a
  PL $d$-ball. Then $M_{s-1}$ is a union of two PL $d$-balls along a PL
  $(d-1)$-ball. This $(d-1)$-ball lies necessarily in the boundary of the both
  two $d$-balls otherwise we would have a $(d-1)$-simplex in at least three
  $d$-simplices which contradicts the fact that $M$ is a pseudomanifold. Thus
  $M_{s-1}$ is also a PL $d$-ball by~\cite[Corollary~3.16]{rourke-sanderson72}.
  By repeating this argument inductively, we deduce that every $M_i$, including
  $M_0$, is a PL $d$-ball.
\end{proof}

We also need a similar lemma for simplicial shelling.
Given a shelling down $F_1, \dots, F_k$ of a pure $d$-dimensional simplicial
complex $K$ and $i \in \{1,\dots,k\}$, we will also use the simplified notation $K_{F_i}$ for the
subcomplex $K[F_i, \dots, F_{m}]$ where $F_{k+1}, \dots, F_m$ are the $d$-faces
(i.e. facets) of $K$ not used in the shelling in an arbitrary order.

\begin{lemma}
  \label{l:all_balls}
  For any shelling down $F_1, \dots, F_{m-1}$ of a triangulated PL $d$-ball $B$ to
  a $d$-simplex $F_m$ the complex $B_{F_i}$ is a PL $d$-ball for every $i \in \{1,
  \dots, m-1\}$.
\end{lemma}

\begin{proof}
  The sequence $F_m, \dots, F_1$ is a shelling up of a $d$-ball.
  By~\cite[Theorem~1.3]{bjorner84}, $F_i \cap K[F_{i+1},\dots, F_m]$ cannot be
  the boundary of $F_i$ for $i \in \{1, \dots, m-1\}$ 
  (the characteristic in~\cite[Theorem~1.3]{bjorner84} is
  $0$ as $B$ is a ball), thus $F_i \cap K[F_{i+1},\dots, F_m]$ is a PL
  $(d-1)$-ball (union of some facets of the boundary of a $d$-simplex). This
  implies that the shelling down from $B_{F_i}$ to $B_{F_{i+1}}$ is also an
  elementary shelling in PL sense. Therefore, Lemma~\ref{l:pseudomanifold}
  implies that
  $B_{F_i}$ is a PL $d$-ball for every $i \in \{1,
  \dots, m-1\}$.
\end{proof}

\begin{definition}
  Let $B$ be a triangulated $d$-ball. A facet $F$ of $K$ is \emph{free
  (for shelling)} if it meets $\partial B$ in a $(d-1)$-ball.
\end{definition}

\begin{lemma}[Essentially Proposition~2.4(iv) form~\cite{ziegler98}]
\label{l:free_facet}
  Let $F_1, \dots, F_{m-1}$ be a shelling down of a triangulated $d$-ball $B$ with
  $m$ facets. Then, for every $i \in \{1, \dots, m-1\}$, the face $F_i$ is free
  in $B_{F_i}$.
\end{lemma}

\begin{proof}
  Although Proposition~2.4(iv) in~\cite{ziegler98} is not exactly stated in
  the same way as our lemma, the proof in~\cite{ziegler98} exactly describes
  which face is free. (We assume that Proposition~2.4(iv) in~\cite{ziegler98} is stated in
  counter-positive and the shelling is reverted.) In particular, it gives that
  $F_1$ is free in $B_{F_1} = K$. It also gives that $B_{F_2}$ is a $d$-ball,
  and thus $F_2, \dots, F_{m-1}$ is a shelling down of this $d$-ball. By
  repeating the previous argument we obtain that $F_2$ is free in $B_{F_2}$ and
  then by induction, we obtain the statement of the lemma.
\end{proof}

\section{Triangulating polytopal complexes}

A \emph{polytopal complex} is a finite collection $\Gamma$ of polytopes in
$\R^d$ for some $d$ satisfying the following conditions.

\begin{enumerate}[(i)]
   \item If $P \in \Gamma$ and $F$ is a face of $P$ (including the empty face),
     then $F \in \Gamma$.
   \item Any two polytopes in $\Gamma$ intersect in a face of both (possibly the
     empty face).
\end{enumerate}

One way to triangulate a polytopal complex $\Gamma$ is to order the
vertices of $\Gamma$ in an arbitrary total order. Then we inductively triangulate
$i$-dimensional faces: For $i = 0, 1$ there is nothing to do. For $i >1$, we
triangulate an $i$-face $F$ so that we pick the first vertex $v$ of $F$ in our
total order and we triangulate $F$ as a cone with apex $v$ over the (already
triangulated) faces of $F$ which do not contain $v$. We call such a
triangulation \emph{canonical}.\footnote{This
triangulation is also known as \emph{pulling triangulation} with a somewhat
different definition in the literature; see~\cite{lee-santos17} or
\cite[Lemma~1.4]{hudson69}. We do not need an equivalence of the two
definitions, thus it is easier to use another name rather than to prove the
equivalence of the two definitions.} Similarly a triangulation of a polytope
$P$ is \emph{canonical} if it can be obtained as above considering the
polytopal complex formed by all faces of $P$.

\begin{lemma}
  \label{l:shell_canonical}
    Let $K$ be a pure (simplicial) $3$-complex. Let us assume
    that $K = L \cup P$ where $L$ and $P$ are pure (simplicial) $3$-dimensional 
    subcomplexes of $K$ such that $P$ is a canonically triangulated polytope.
    In addition, let us assume that $L \cap P$ is a disk (i.e. a $2$-ball)
    fully contained in $\partial P$. 
    Then $K$ shells to $L$.
\end{lemma}

\begin{proof}
Because $P$ is canonically triangulated polytope, it is a cone with an apex $v$
over a disk. Now we decompose $\partial P$ into four subcomplexes. Subcomplex
  $W$ is formed by the triangles (and their subfaces) which contain $v$ and belong
  to $L$; $X$ is formed by the triangles which contain $v$ and do not belong to
  $L$; $Y$ is formed by the triangles which do not contain $v$ and belong to
  $L$; and $Z$ is formed by the triangles which do not contain $v$ and do not
  belong to $L$. See Fig.~\ref{fig:shell_canonical}.

  \begin{figure}
      \begin{center}
	  \includegraphics{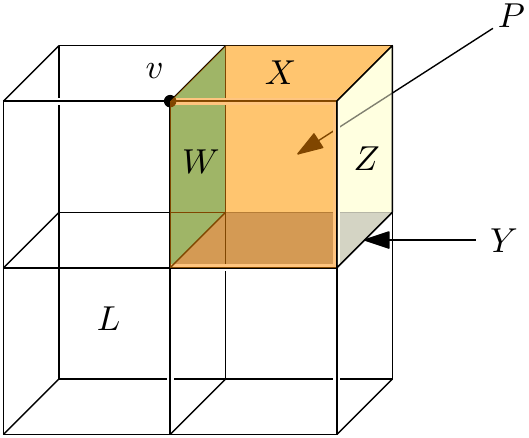}
	  \caption{Subcomplexes $W, X, Y$ and $Z$ in the proof of
	  Lemma~\ref{l:shell_canonical} (without displaying the triangulation).}\label{fig:shell_canonical}
	  \end{center}
  \end{figure}

  By assumptions, $W \cup Y$ is a disk. From this it also follows that $W$ and
  $X$ are either disks or empty by checking the neighborhood of $v$. 
  Finally,
  $W \cup Y \cup Z$ is a disk or whole $\partial P$ because $X$ is a disk or
  empty.
  
  Throughout this proof, we only use shellings up. We use an auxiliary shelling of $W \cup Y \cup Z$ obtained in
  the following way. We start with a shelling of $W$ (if $W$ is non-empty).
  Then we extend this shelling to a shelling of $W\cup Y$ by extendable
  shellability of any triangulated disk. (It is widely known and easy to prove
  that any partial shelling of a disk or sphere can be completed to a full
  shelling; see, e.g.~\cite{danaraj-klee78algo}.) Then we extend this shelling once more to $W \cup Y
  \cup Z$ using the extendable shellability again (for a disk or a $2$-sphere).  
  
  From the shelling
  that we have just obtained we aim to deduce a shelling order of tetrahedra of
  $P$ from $L$ up to $K$.
  The formula is simple, we follow our ($2$-dimensional shelling) restricted to
  the triangles of $Y \cup Z$. Then we add $v$ to each such triangle $\tau$,
  obtaining the tetrahedron $v * \tau$, and we
  claim that we get the required shelling of $P$. Because $P$ is a cone from
  $v$ over $Y \cup Z$, this covers all tetrahedra. Thus it remains to verify
  that each step is a correct shelling step.

  If $\tau$ is a triangle in $Y$, then it meets the preceding
  triangles of $W \cup Y$ in
a pure subcomplex $M$ of $\partial \tau$. This $M$ is typically one
  dimensional; however it is empty if $\tau$ is the first triangle in the
  shelling and $W$ is empty.
  Then the
  corresponding tetrahedron $v * \tau$ meets the preceding tetrahedra and $L$
  in the pure $2$-dimensional subcomplex $(v * M) \cup \tau$ as required.

If $\tau$ is a triangle in $Z$, then it meets the preceding triangles in a
  pure $1$-dimensional subcomplex $M$ of $\partial \tau$. Then the
  corresponding tetrahedron $v * \tau$ meets the preceding tetrahedra and $L$
  in the pure $2$-dimensional subcomplex $v * M$. (Note that $\tau$ does not
  belong to $L$ in this case and if $v * \tau$ meets $L$ in some face $\sigma$,
  then $\sigma$ is in some triangle $\tau'$ of $Y$, therefore $\sigma$ is
  also contained in the tetrahedron $v * \tau'$ preceding $v * \tau$.)
\end{proof}

\section{Gadgets for shellability}
\label{s:shellability_gadgets}

\subsection{Triangular prism---an analogy of the bipyramid}
Now we introduce an auxiliary complex which will serve as a variable gadget.
Our complex will be a suitably triangulated triangular prism, and we call it
simply a \emph{triangular prism}.

First we take a triangular prism with vertices $b'c'd'bcd$ as in
Figure~\ref{f:triangular_prism}, ignoring the subdivision for the moment. Then
we subdivide the triangle $bcd$ to three triangles 
by inserting a vertex $a$ into the barycentre of $bcd$ and connecting it with
$b$, $c$ and $d$. Analogously, we subdivide $b'c'd'$ by inserting $a'$ into the
barycentre. Next we subdivide the rectangles $b'c'cb$, $c'd'dc$ and 
  $d'b'bd$ by inserting one of the two possible diagonals---every combination
  of diagonals is a valid choice. Finally, we insert a vertex $e$ into the
  barycentre of the prism and triangulate the prism as a cone with the apex $e$ over
  the already triangulated boundary; see Figure~\ref{f:triangular_prism}.

The following technical lemma states a condition on order of shelling of certain
tetrahedra in the triangular prism, if this prism is suitably contained in some
$3$-ball. Although this is not obvious yet, this condition will confirm that
the triangular prism `behaves as variable gadget' with respect to our needs.

\begin{figure}
  \begin{center}
    \includegraphics[page=37]{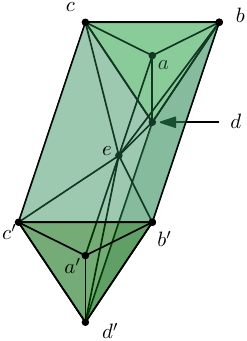}
  \end{center}
  \caption{Triangular prism. 
  The diagonals of the rectangles $b'c'cb$, $c'd'dc$ and 
  $d'b'bd$ can be chosen arbitrarily and they are not depicted.}
  \label{f:triangular_prism}
\end{figure}

\begin{lemma}
  \label{l:tp_blocked}
  Let $P$ be a triangular prism on vertices $a, b, c, d, a', b', c', d', e$ as
  in the definition above. Let us assume that $P$ is a subcomplex of a
  triangulated $3$-ball $B$ such that $P \cap \partial B$ consists exactly of
  the subcomplex formed by the four triangles $a'b'c', a'c'd', abc$ and
  $acd$. Let $L$ be another subcomplex of $B$ such that $B$ shells to $L$.
  Assume that $L$
  contains the six tetrahedra in $B$ outside $P$ which meet $\partial P$ in one of the six triangles subdividing the
  rectangles $b'c'cb$, $c'd'dc$ and $d'b'bd$.
  Then $L$ either contains all three tetrahedra $abce, abde$ and $acde$ or all three
  tetrahedra $a'b'c'e, a'b'd'e$ and $a'c'd'e$. 
\end{lemma}

\begin{proof}
  For contradiction let us assume
  that one of the tetrahedra $abce, abde$ or $acde$ does not belong to $L$ as
  well as one of the tetrahedra $a'b'c'e, a'b'd'e$ or $a'c'd'e$ does not belong
  to $L$.
  
  Now we fix a shelling of $B$ down to $L$ and without loss of generality we can assume that one of the tetrahedra $abce, abde$ or $acde$ is removed earlier than any of the tetrahedra $a'b'c'e, a'b'd'e$ or $a'c'd'e$. (If one of the tetrahedra
  $a'b'c'e, a'b'd'e$ or $a'c'd'e$ is removed earlier than any of
  the tetrahedra $abce, abde$ or $acde$, then the argument is symmetric by
  replacing $a$ with $a'$, $b$ with $b'$ and $c$ with $c'$.)

  Let $\Delta$ be the tetrahedron 
  which is the first one removed in shelling down from $B$ to $L$ among
  those tetrahedra that contain $e$ and at least one of the vertices $a'$,
  $b'$, $c'$ or $d'$. Let $L'$ be the intermediate complex in our shelling down
  from $B$ to $L$ obtained exactly before shelling $\Delta$. (This is the same as
  $B_\Delta$ in the notation of Section~\ref{s:shelling_balls} for our fixed
  shelling.) By Lemma~\ref{l:free_facet} we get that $\Delta$ is free in $L'$; that
  is, $\Delta$ intersects $\partial L'$ in a disk.

  First, we verify that $\Delta$ cannot intersect both $\{a,b,c,d\}$ and $\{a',
  b', c', d'\}$. For contradiction, we assume so. Let $f' \in \{a',  b', c',
  d'\}$ be a point of intersection of $\Delta$ and $\{a',  b', c',  d'\}$. Note that
  $f'$ is on the boundary of $L'$ (because it is on the boundary of $B$ by the
  assumptions of the lemma). We observe that no triangle of $\Delta$ containing $f'$
  may belong to $\partial L'$. Indeed, let $\tau$ be such a triangle. If $\tau$
  contains $e$ then the other tetrahedron (than $\Delta$) of $P$ containing $\tau$
  is still present in $L'$ (due to the definition of $\Delta$). Thus $\tau$ belongs
  to two tetrahedra of $L$. If $\tau$ does not contain $e$, then the other tetrahedron
  of $B$ than $\Delta$ (outside $P$) is present in $L$ a fortiori in $L'$ due to the
  assumptions of the lemma. Altogether $f'$ belongs to $\partial L'$ while no
  triangle of $\Delta$ containing $f'$ belongs to $\partial L'$ which implies that
  $\Delta$ cannot meet $\partial L'$ in a disk.
 
The only remaining options are that $\Delta$ is one of the tetrahedra $a'b'c'e$,
  $a'b'd'e$, or $a'c'd'e$. In this case $e$ is on the boundary of $L'$ because
  one of the tetrahedra $abce$, $abde$ or $acde$ was removed before $\Delta$. On the
  other hand no triangle of $\Delta$ containing $e$ belongs to $\partial L'$ because
  $\Delta$ is removed before any other tetrahedron containing
  $e$ and at least one of the vertices $a'$, $b'$, $c'$ or $d'$. Thus $\Delta$
  cannot meet $\partial L'$ in a disk. This is the desired contradiction.
\end{proof}

\subsection{Thick 1-house}
\label{ss:thick_1house}

Now we want to build a thick version of the 1-house with analogies of a distinguished
tree and
crossing circles from Figure~\ref{f:1house_circles}. In our description we
closely follow Bing~\cite{bing64}, except that we build our house
not only from bricks but we will also allow other polytopes (in the thickened
lower wall). We first describe the thick 1-house as a polytopal complex, only
later on we provide a triangulation.

As a preliminary construction, we think of the thick 1-house from
Figure~\ref{f:1house} built from axis aligned bricks (cubes) such that every
pair of cubes meets in a face of both (so that we get a polytopal complex).
The walls are built from a single layer of cubes so that each
cube (except the one called $F$ that will correspond to $f$) intersects the
boundary in at least two components. We think of the free face $f$ appropriately short so that it is
represented by a single brick $F$. This brick meets the boundary in a single
component.  See Figure~\ref{f:bing_three_layers} for the front three layers of
the construction if the starting 1-house is slightly rotated.

Note that each cube naturally corresponds either to a vertex of the thin
$1$-house, or to a relative interior of an edge, or to a relative interior of
some $2$-cell (i.e. not necessarily convex polygon). In sequel, we simplify
this by saying that the cube corresponds to a vertex, an edge or a $2$-cell
without emphasizing the relative interior.

In order to avoid some problematic cases, we also assume that the bricks are
sufficiently small in a sense that if $e$ and $e'$ are two disjoint edges of
the thin house then their thickenings (to unions of cubes) are still
disjoint. (For this purpose by edges we mean the axis-aligned segments from
Figure~\ref{f:bing_three_layers}, left, where either three faces meet or two
faces bend.)

By the (thickened) \emph{lower wall} we mean the subcomplex formed by the cubes
corresponding to the thin lower wall. According to our earlier convention this
means that this does not contain the cubes corresponding to the vertices or edges of the
thin lower wall. For example, in Figure~\ref{f:bing_three_layers} the
dimensions of the lower wall are $11 \times 4 \times 1$. We also sometimes
refer to the \emph{front side of the lower wall} which is formed by the squares
of the lower wall on the `facade' of the thick 1-house.

\begin{figure}
  \begin{center}
    \includegraphics[page=15]{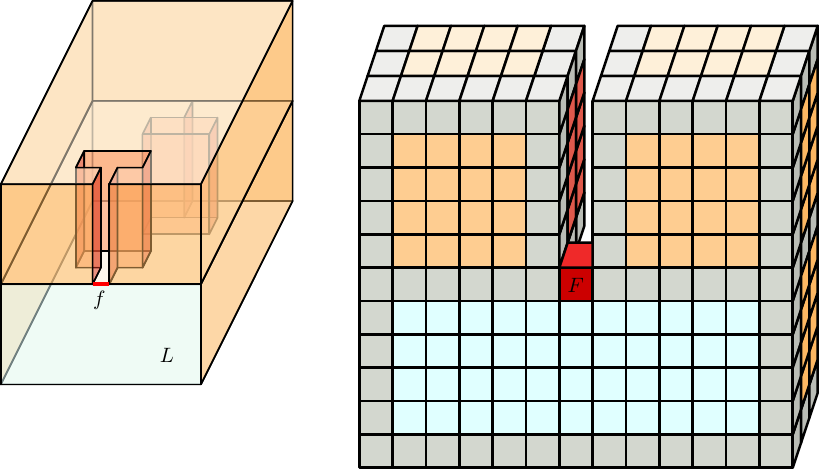}
    \caption{Left: Rotated standard 1-house, compare with Figure~\ref{f:1house}. Right:
    The first three layers of bricks of the thick 1-house. The number of
    bricks in one layer is only illustrative. In most applications, we will
    require more bricks in a layer.}
  \label{f:bing_three_layers}
  \end{center}
\end{figure}

\paragraph{Triangulation.}
Now we aim to triangulate our thick 1-house; at the moment we have only a
polytopal complex. As we also want to realize analogies of the distinguished
tree and the crossing circles from Figure~\ref{f:1house_circles}; our
triangulation will not be fixed. Later on, in some cases when we want to use the
  thick 1-house, we will allow to remove the bricks of the lower wall
  (see Figure~\ref{f:bing_cubes}, left) and to retriangulate the lower wall in
  a different way. (We also keep some flexibility regarding how many cubes we
  use when regarding the thick 1-house already as a polytopal complex.)

\begin{figure}
\begin{center}
  \includegraphics[page=16]{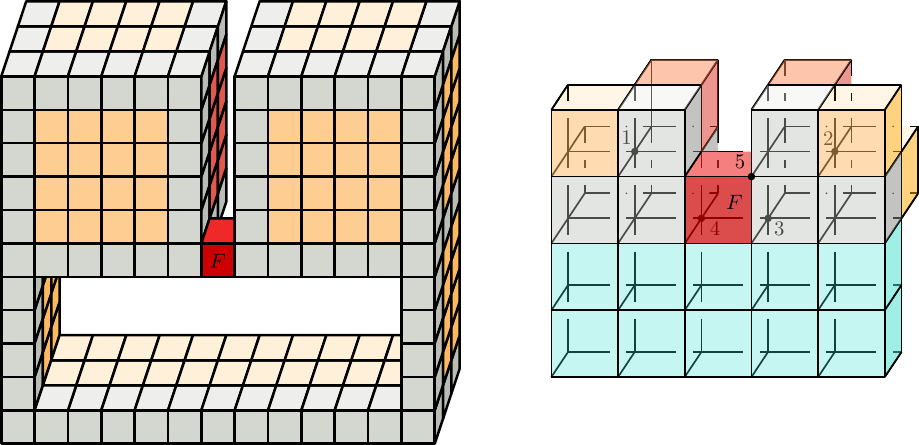}
\caption{Left: The first three layers of the thick 1-house with the bricks of the lower
  wall removed. Right: First five vertices in the order used for triangulation.}
  \label{f:bing_cubes}
\end{center}
\end{figure}

In order to obtain a triangulation of the thick 1-house, we provide
a total order on the vertices of the union of our bricks. Then we use the corresponding
canonical triangulation considering the thick 1-house as a polytopal
complex. We set up the following rules for our total order.

\begin{enumerate}[(R1)]
  \item
The first five vertices in our order are the vertices $1, \dots, 5$ from
Figure~\ref{f:bing_cubes}, right, in this order. 
    \item
    Then we put in an arbitrary order those
vertices that are a component of intersection of some cube with the boundary of
the thick 1-house. 
    \item
    Next, among the remaining ones, we put in an arbitrary
order those that are in an edge which is a component of intersection of some
cube with the boundary of the thick 1-house. 
    \item 
    Then we put the remaining
vertices in an arbitrary order. 
\end{enumerate}

\paragraph{The attachment and the shelling complexes.}
Now we aim to describe the promised analogies of the distinguished tree and
crossing circles. We do so by describing two types of subcomplexes of the boundary:
the \emph{attachment complex} and \emph{shelling complexes}. The attachment
complex is simply the piece of boundary along which we aim to glue our
thick 1-house to other gadgets. The shelling complex will be used when we shell the whole construction for a satisfiable formula. 
In general, there may be several different shelling
complexes as we will need some flexibilty to perform different types of
shellings. Usually, shelling complex will be a piece of boundary along which is our
1-house glued to the rest at the moment when we want to shell it. For this
reason, such a shelling complex will be a subcomplex of the attachment
complex. However, occasionally we will also need a 3-dimensional shelling
complex which will (of course) not be a subcomplex of the attachment complex.

We will use our thick 1-house in analogous cases as for the collapsibility
reduction, that is, as a splitter, incoming/outgoing house, or a blocker.
After the thickening, each of this cases has its own specifics, thus we discuss
each case individually. First we, however, introduce some notation.

We consider the cube $F$, the cube right of $F$ and the cup top right of $F$
and we denote their vertices $a, \dots, h, a', \dots, h'$ as in
Figure~\ref{f:F_triangulate}, left. The numbers indicate the order of the vertices
used for the triangulation; compare with Figure~\ref{f:bing_cubes}, right. In
particular, $f'$, $b'$, $e'$, $a$ in this order are the first four 
vertices among the vertices of these cubes. This induces the triangulation of
the three cubes up to the choice of the diagonal of the square $ghh'g'$. In
fact the diagonal of this square is also determined as $h'$ precedes $g$, $g'$
and $h$ due to the fact that $h'$ belongs to the edge $h'f'$  which is a
component of the intersection of the top right cube with the boundary.
Therefore $h'$ satisfies the rule (R3) of our total order while it is easy to
check that $g$, $g'$ and $h$ fall only into the (R4) case.

\begin{figure}
\begin{center}
  \includegraphics[page=17]{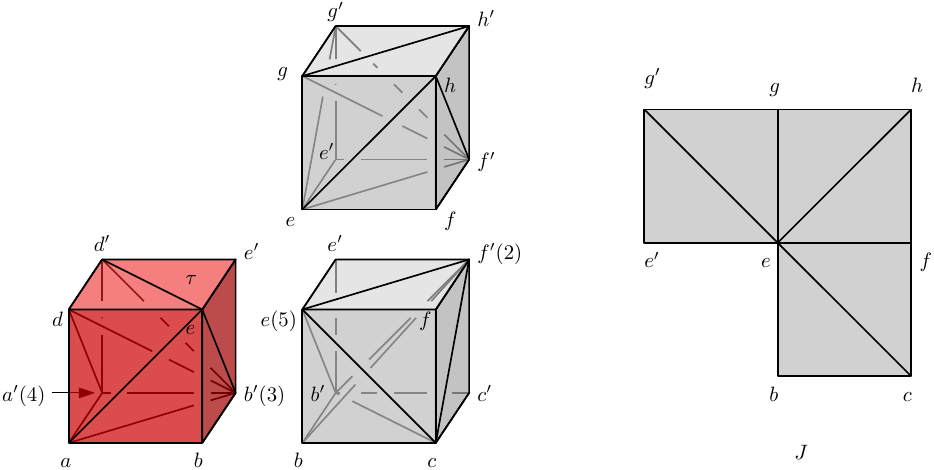}
  \caption{Left: Triangulations and the names of vertices of $F$ and two
  neighboring cubes. Right: The subcomplex $J$ after unfolding.}
  \label{f:F_triangulate}
\end{center}
\end{figure}

We also denote by $\tau$ the triangle $ee'd'$, this will be a very important
triangle, and by $J$ the subcomplex formed by the subdivided squares $bcfe$,
$efhg$ and $ee'g'g$; see Figure~\ref{f:F_triangulate}, right. The subcomplex $J$ will always be a part of the attachment
complex. The triangle $\tau$ will be a part of the attachment complex but not
the shelling complex.

\paragraph{The splitter case.}
Now we describe the attachment complex and the shelling complexes in the
\emph{splitter case}. In this case, there is a one parameter of freedom, a
positive integer $k$ which will correspond to the number of `branches' in the
splitter.

In this case, we keep the lower wall as is assuming that the dimensions of the front side are $(3k + 2) \times 6$. 
(If $k \leq 2$, we let the first dimension to be $11$ 
so that we
have enough space to finish the construction of the thick 1-house. If $k
> 2$ is even then $(3k + 2)$ is even and thus $F$ cannot appear above the `middle
column' of the lower wall. Then we let $F$ to appear above
the left one of the two middle columns. This means that our decomposition of
the thick 1-house into bricks is slightly asymmetric but this is not a problem
at all.)
We mark $k$ squares $S_1, \dots, S_k$ on the front side of the lower
wall so that the coordinates of $S_i$ are $(3i, 3)$ (assuming that the
coordinates of the bottom left square are $(1,1)$); see
Figure~\ref{f:splitter_wall} while following the construction. For further
reference, we also mark the vertices of $S_i$ by $a_i, b_i, c_i$ and $d_i$ in
the order as in the figure. (They should not be confused with vertices $a$, $b$
$c$ and $d$ in $F$ or $J$.) Each
square $S_i$ is further subdivided into two triangles, and in fact we can
choose how: When setting up our total order for triangulating the thick
1-house, all vertices of $S_i$ are subject to rule (R4). The diagonal appears
at the first one of them which we can choose independently in every $S_i$.  We
set up the diagonals in such a way that they connect $b_i$ and $d_i$.  Then we
mark a subdivided horizontal $(3k-2) \times 1$ rectangle $R_1$ which connects
the squares with coordinates $(3, 4)$ and $(3k, 4)$ if $k \geq 3$, or $(3,4)$
and $(9, 4)$ if $k \leq 2$.  Finally we mark a vertical subdivided $1 \times 2$
rectangle $R_2$ connecting $R_1$ and $J$ (i.e., it shares an edge with each of
them). 

\begin{figure}
\begin{center}
  \includegraphics[page=18]{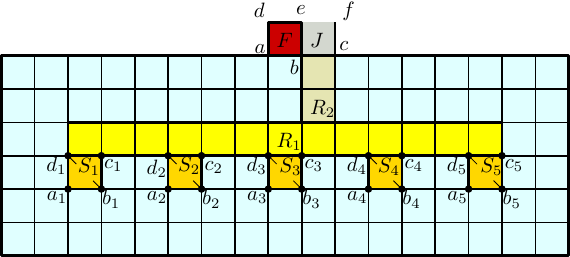}
  \caption{The lower wall in the splitter case. The squares $S_i$ and the rectangles $R_1$ and $R_2$ on the front
  side of the lower wall for $k = 5$. Each square is subdivided into two
  triangles but we depict them only for $S_i$.}
  \label{f:splitter_wall}
\end{center}
\end{figure}

Now the \emph{attachment complex} is the subcomplex of the boundary of the
thick 1-housed formed by $S_1, \dots, S_k, R_1, R_2, J$ and $\tau$. 
The are two shelling complexes. The \emph{first shelling complex} 
is formed by $S_1, \dots, S_k, R_1, R_2, J$ without $\tau$; the \emph{second
shelling complex} is formed by $S_1, \dots, S_k, R_1, R_2$ and a part of $J$ consisting of all its
triangles except $eg'e'$.

\paragraph{The incoming case.} Here we describe the attachment complex and the
shelling complex in the \emph{incoming case}. We will again use one parameter
of freedom $k \geq 0$ which will be number of \emph{crossing annuli}. They are
analogues of crossing circles used in the collapsibility reduction.

We again keep the lower wall as is assuming that the dimensions of the front
side are $(3k + 17) \times 6$. However, in this case 
we also assume that the original polytopal decomposition of the thick
1-house is very eccentric in the sense that $F$ appears above the 7th
column and also that the `chimney' is sufficiently small so that there is
enough space on the right of the chimney (this will be stated more precisely but for
the moment see Figure~\ref{f:crossing_annuli}).

Now we mark a subdivided horizontal $(3k+6)\times 1$ rectangle $R_1$ which connects the
points $(8,4)$ and $(3k+13, 4)$ and a subdivided vertical $1 \times 2$ rectangle
connecting $J$ and $R_1$. We also mark a subdivided square $S$ with coordinates
$(3k+14, 4)$ and its vertices $w, x, y, z$ as in Figure~\ref{f:incoming_wall}. 

\begin{figure}
\begin{center}
  \includegraphics[page=20]{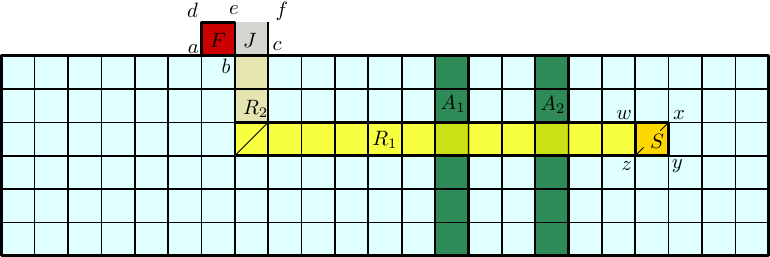}
  \caption{The lower wall in the incoming case. The squares $S$, the rectangles $R_1$ and $R_2$ and the annuli $A_i$
  on the front side of the lower wall for $k = 2$. Each square is subdivided into two
  triangles but we depict them only for $S$.}
  \label{f:incoming_wall}
\end{center}
\end{figure}

Then we mark annuli $A_1, \dots, A_k$. Each annulus $A_i$ meets not only
the lower wall but also other parts of the boundary of the thick 1-house.
Namely, $A_i$ meets the lower wall in $(3i + 11)$th column (see
Figure~\ref{f:incoming_wall}); in particular, it meets $R_i$ in the square with
coordinates $(3i+11, 4)$. Then we extend this column all the way around
the thick 1-house (so that each square of the annulus meets exactly two
other squares and they always meet in an edge, in the original decomposition of
1-house into the bricks before the triangulation); see
Figure~\ref{f:crossing_annuli} for a part of this extension. (Compare also with
Figure~\ref{f:1house_circles}; an annulus $A_i$ is just a thickening of a
crossing circle.) We also assume that the `chimney' does not touch any of the
annuli $A_i$.

\begin{figure}
\begin{center}
  \includegraphics[page=21]{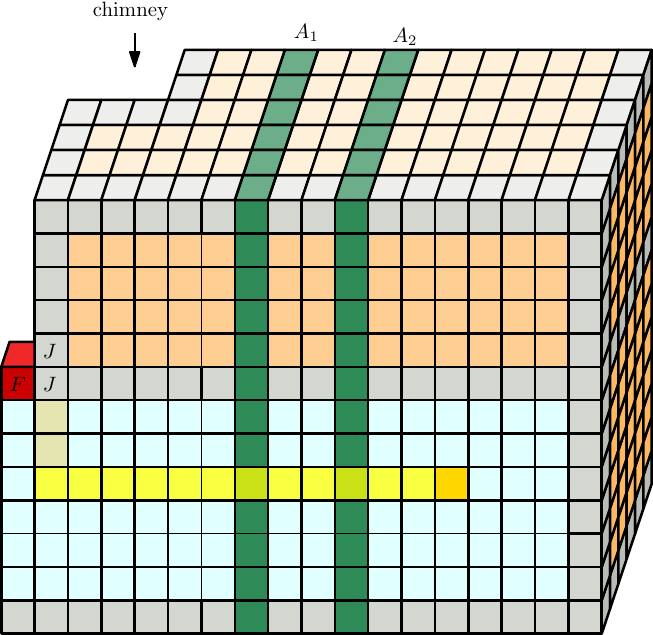}
  \caption{The crossing annuli do not touch the chimney. (Only a part of the
  house right of $F$ is drawn.)} 
\label{f:crossing_annuli}
\end{center}
\end{figure}

We also remark that we can make a choice of some diagonals of squares similarly
as in the splitter case. We will make this choice on the square with
coordinates $(8,4)$ and for the square $S$ as in Figure~\ref{f:incoming_wall}.

Finally, the \emph{attachment complex} is the subcomplex of the boundary of the
thick 1-house formed by $R_1, R_2, S, A_1, \dots, A_k, J$ and $\tau$. 
We again have two shelling complexes. The \emph{first shelling complex} is formed by $R_1, R_2, S$ and 
part of $J$ formed by all triangles triangles except $eg'e'$. The \emph{second
shelling complex} is formed by $R_1, R_2$ and
part of $J$ formed by all triangles triangles except $eg'e'$.

\paragraph{The outgoing case.} Here we describe the attachment complex and the
shelling complex in the \emph{outgoing case}. Also in this case there is
one parameter of freedom $k \geq 1$ which will be number of \emph{used squares}.

We first keep the lower wall as is (but later we may retriangulate it) assuming
that the dimensions of the front side are $11 \times  (k+3)$. The dimension 11 serves here only to keep enough
space to build the chimneys. In this case we only mark a subdivided vertical $1
\times k$ rectangle $R$ directly below $J$ and a triangle $opq$ in the square
below $R$ sharing the edge $op$ with $R$; see Figure~\ref{f:outgoing_wall}. It does not really matter whether $oq$ or $pq$ is the diagonal of the square. 

\begin{figure}
  \begin{center}
    \includegraphics[page=28]{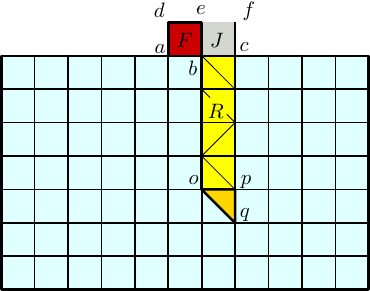}
    \caption{The lower wall in the outgoing case with $k=4$ and some choice of
    the diagonals of the squares in $R$.}
    \label{f:outgoing_wall}
  \end{center}
\end{figure}

On the other hand, we need a flexibility to choose the diagonals of $R$ as the
need arises. This can be achieved by retriangulating the lower wall:
First we decompose the 
lower wall into triangular prisms according to the
chosen diagonals (this way we can choose a diagonal in every square of the
front side of the lower wall). Note that the intersection of the lower wall with the remaining cubes is an annulus. We now set up a total
order on the vertices of the lower wall so that the
vertices of the annulus are the first (in the exactly same relative order as in
the triangulation of the whole 1-house). Then we put the remaining vertices of
the lower wall in an arbitrary order and take the corresponding
canonical triangulation. This way we obtain a
triangulation which refines the pre-chosen triangular prisms while it agrees on
the annulus. Thus it is possible to glue back the lower wall
with this triangulation.

The \emph{attachment complex} in this case consists of $R$, $J$, $\tau$ and the
triangle $opq$. The \emph{shelling complex} consists of $R$, the triangle
$opq$ and a part of $J$ formed by all its triangles except $eg'e'$.

\paragraph{The blocker case.}

Finally, we describe the attachment complex and the shelling complex in the
blocker case. In this case we take the lower wall with dimensions of the front
side $11 \times 11$ but we aim to retriangulate the lower wall quite
significantly. (The blocker is adapting to the shape of other gadgets perhaps
the most noticeably.)

We alter the triangulation of the lower wall in several
steps. First we describe what occurs on the front side; see
Figure~\ref{f:blocker_wall}. We consider  9 important squares $S_1, \dots,
S_{9}$ so that the coordinates of $S_i$ in our coordinate system from earlier
cases are $(7, 13-i)$; only six of these squares are denoted in
Figure~\ref{f:blocker_wall} due to lack of space (or possible ambiguity). We
further subdivide $S_1$, $S_6$, $S_7$, $S_8$ and $S_9$ as 
in Figure~\ref{f:blocker_wall}
and we also introduce the notation for vertices $b'_+$, $b'_-$, $c'_+$, $c'_-$,
$d'_+$, $d'_-$, $b_+$, $b_-$, $d_+$, $d_-$, $e_+$, $e_-$, $b_*$,
$c_*$, $d'_*$, $e_*$, $e'_*$, $f_*$, $g_*$, $g'_*$ and $h_*$ as depicted. The
choice of the notation is not meaningful immediately but it will be useful when gluing the blocker to the other gadgets.

\begin{figure}
  \begin{center}
    \includegraphics[page=39]{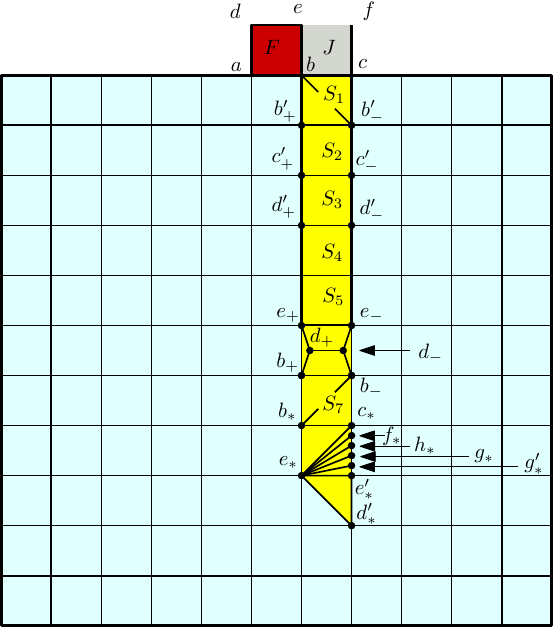}
    \caption{The front side of the lower wall in the blocker case.}
    \label{f:blocker_wall}
  \end{center}
\end{figure}

Now we subdivide the back side of the lower wall. It is
subdivided in the same way as the front side with the exception of $S_6$. We do
not subdivide the square corresponding to $S_6$ in the back side. We
temporarily remove the cube containing $S_6$ from our lower wall and we
decompose the rest of the lower wall into prisms
according to our subdivision of the front side and the back side. Regarding the
missing cube containing $S_6$, we replace it in the wall with a polytope $P_6$ 
on $10$ vertices obtained
from this cube by pushing $d_+$ and $d_-$ a little bit forward; see
Figure~\ref{f:blocker_polytope}. We also mark a point $p$ on $P_6$ behind
$b_-$ (this will be an important point for getting a suitable triangulation).
With a slight abuse of the notation, we still denote the vertices of the front
four faces of $P_6$ as $b_+, b_-, d_+, d_-, e_+, e_-$ and the union of the four
front faces as $S_6$.

\begin{figure}
  \begin{center}
    \includegraphics[page=33]{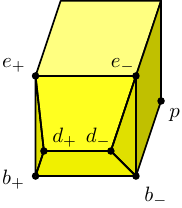}
    \caption{Pushing $d_+$ and $d_-$ a bit forward.}
    \label{f:blocker_polytope}
  \end{center}
\end{figure}

Now we triangulate the lower wall in a suitable way so that
we can glue it to the rest of the thick 1-house in the same way as it was glued
before altering the triangulation: Similarly as in the outgoing case, the lower wall meets the remaining cubes in an annulus. We set up a
total order for a canonical triangulation so that we start with the vertices of
this annulus in the same relative order as for our original triangulation of
the thick 1-house. Then we put $p$ (which ensures that $P_6$ will be
triangulated as a cone with apex $p$) and then the remaining vertices of the
lower wall in an arbitrary order. Because we have started with
the vertices of the annulus, the lower wall triangulated this
way can be glued back to the thick 1-house. 

Now the \emph{attachment complex} in this case consists of $\tau$, $J$, the subdivided
squares $S_1, \dots, S_8$ and the triangle $d'_*e'_*e_*$. 

In the blocker case, there will be actually more possible shelling complexes as
we will use several slightly different blocker houses. 

The \emph{type 0 shelling complex} consists of the subdivided squares $S_2,
\dots, S_8$ and the triangle $d'_*e'_*e_*$.

The \emph{type i shelling complex} is a (non-pure) $3$-dimensional subcomplex of
the blocker house and it consists of the subdivided squares $S_2,
\dots, S_8$, the triangle $d'_*e'_*e_*$ and $P_6$ (as a 3-dimensional piece).

Finally, the \emph{type n shelling complex} is a (non-pure) $3$-dimensional subcomplex of
the blocker house and it consists of $J$, the subdivided squares $S_1,
\dots, S_8$, the triangle $d'_*e'_*e_*$ and $P_6$ (as a 3-dimensional piece).

\bigskip

This finishes the discussion of different cases how to triangulate the lower
 wall.

Now we state and prove two lemmas on shelling the thick 1-house which will be useful in the reduction.

\begin{lemma}
  \label{l:shelling_house}
  Assume that $K$ is a pure $3$-complex. Assume that $K = H \cup L$ where
  $H$ is the thick 1-house and $H \cap L$ is the shelling complex of $H$.
  Then $K$ shells to $L$.
\end{lemma}

\begin{proof}
  The main idea is to obtain a proof by a repeated application of
  Lemma~\ref{l:shell_canonical} (with a few extra steps not using the lemma).
  That is, we want to shell most of $H$ by removing the cubes of $H$ one by one
  so that each cube intersects the remainder of the intermediate complex in a
  disk. We will essentially follow the order of collapses of $1$-house
  in~\cite{gpptw19, tancer16}; however, we have to treat various shelling
  complexes carefully. For this we need some auxiliary notation.

Let $F_\downarrow$, $F_\to$, $F_\searrow$, $F_\nearrow$
  be the closest cube in $H$ below $F$, right of $F$, diagonally right below
  $F$ and diagonally right above $F$ respectively. 
  The directions are according to
  Figure~\ref{f:bing_three_layers}.
  Next, we define
  certain auxiliary complex $A$ which will be a subcomplex of the
  (triangulated) front side of the lower wall. In the splitter, ingoing or
  outcoming case we set $A$ to be the intersection of the shelling complex with
  the front side of the lower wall. In the blocker case, we set $A$ to be the
  intersection of the attachment complex with  the front side of the lower wall.
We also note that before triangulating, the lower wall is in each case decomposed into
      prisms and possibly the exceptional polytope $P_6$ in the blocker case.
      (These prism are usually cubes of the original thick 1-house, but there may
      be triangular prisms in the outgoing case or the blocker case.)
   We define $\Omega$ as the collection of those 3-dimensional prisms such that their
      intersection with the front side of the lower wall
      meets $A$ in a at most $1$-dimensional piece and
      $\Upsilon$ as the collection of the remaining 3-dimensional prisms (and $P_6$ if
      applicable). We also define $\Upsilon^+$ as the union of $\Upsilon$ and
      $\{F_\to, F_\nearrow\}$.

By considering each case of the construction separately, we observe that the
  shelling complexes, $\Omega$ and $\Upsilon$ in various cases satisfy the
  following properties:

  \begin{enumerate}[(P1)]
    \item The polyhedron of each shelling complex is a subset of the union of
      of the prisms from $\Upsilon^+$.
    \item $F_\downarrow$ belongs to $\Omega$.
    \item The prisms subdividing $F_\searrow$ belong to $\Upsilon$.
    \item The union of the prisms of $\Omega$ meets the front side of the lower wall in a disk.
    \item The union of the polytopes of $\Upsilon$ meets the front side of the lower wall in a disk.
  \end{enumerate}

Now we describe the desired shelling. 
  We start our shelling by removing the tetrahedra inside $F$ by
  Lemma~\ref{l:shell_canonical} and then the tetrahedra inside $F_\downarrow$ by the
  same lemma. This is possible due to the fact that $H \cap L$ is the shelling complex
  and due to properties (P1) and (P2). (Thus $L$ does not impose new
  restrictions on shelling $F$ and then $F_\downarrow$.) 

 Then we continue by removing the tetrahedra inside prisms from $\Omega$ using
 Lemma~\ref{l:shell_canonical} considering these prisms one by one. This can be
 done either in greedy manner (after a small thought using extendable
 shellability of a disk and (P4)), or this can be easily done by hand in each case
 separately. See Figure~\ref{f:removing_Omega} in the blocker case. (We again
 use (P1) and the fact that $H \cap L$ is the shelling complex thus these
 shellings work not only in $H$ but also in $K = H \cup L$.)

\begin{figure}
  \begin{center}
    \includegraphics[page=51]{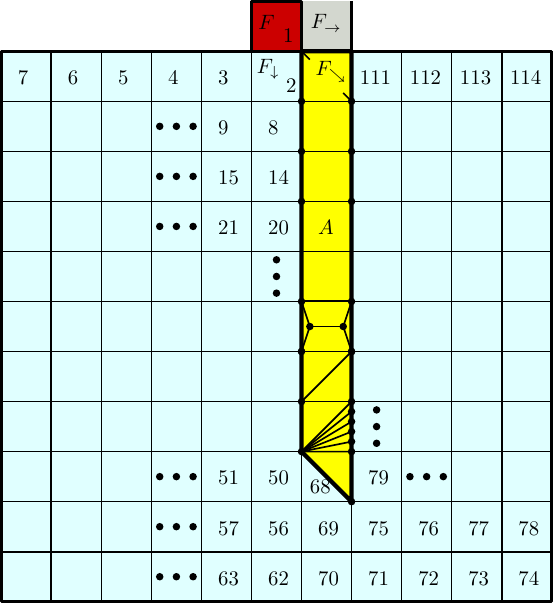}
    \caption{Order of removals of prisms in $\Omega$ (together with $F$ and
    $F_\downarrow$) in the blocker case. }
    \label{f:removing_Omega}
  \end{center}
\end{figure}

 Next, by a repeated application of Lemma~\ref{l:shell_canonical} (using (P1)
 and the fact that $H \cap L$ is the shelling complex) we remove the tetrahedra
 inside cubes of $H$ in the thickened ground floor of $H$ including the cubes
 that correspond to vertices and edges of the ground floor of the thin 1-house.
 Then we continue with (thickened) vertical walls originally touching the
 ground floor (except the lower wall, including the thickened vertical edges
 but not the horizontal ones). Next we remove the middle floor in a direction
 from the hole of the (already removed) thick bottom chimney (including the
 cubes corresponding to the vertices and edges of the thin middle floor except
 those that belong to the lower wall). Then we continue with upper vertical
 walls except those which are above the lower wall (but including the
 cubes corresponding to their vertices and edges). Next we remove the roof and
 then the remaining vertical walls, including the cubes corresponding to their
 vertices and edges except $F_\nearrow$ and $F_\to$. At this moment only the tetrahedra in the 
 prisms of $\Upsilon^+$ remain.

 Now, we remove the tetrahedra in $F_\nearrow$ by an application of
 Lemma~\ref{l:shell_canonical}. Note that $L \cap P$ (in the notation of
 Lemma~\ref{l:shell_canonical}!) is the following disk (see
 Figure~\ref{f:F_triangulate}):
\begin{itemize}
  \item The disk formed by triangles $eff', ef'e', efh, egh, egg'$ and $ee'g'$
    for the first shelling complex in the splitter case or the type $n$
    shelling complex in the blocker case.
  \item The disk formed by triangles $eff', ef'e', efh, egh$ and $egg'$ for the
    second shelling complex in the splitter case, the first or second shelling
    complex in the incoming case or the shelling complex in the outgoing case.
   \item The disk formed by triangles $eff'$ and $ef'e'$ for the type $0$ or
     type $i$ shelling complex in the blocker case.
\end{itemize}
 
 Next, we remove the tetrahedra in $F_\to$ by an application of
 Lemma~\ref{l:shell_canonical}. We again describe the disk $L \cap P$ (in the notation of
 Lemma~\ref{l:shell_canonical}):

\begin{itemize}
  \item It is the disk formed by triangles $bb'c$ and  $b'cc'$ for the type $0$ or
     type $i$ shelling complex in the blocker case.
  \item It is the disk formed $bb'c$, $b'cc'$, $bce$ and $cef$ for all other
    shelling complexes.
\end{itemize}

 It remains to remove tetrahedra in prisms in $\Upsilon$. 

 In the splitter case, incoming case or outgoing case we apply
 Lemma~\ref{l:shell_canonical} to this prisms one by one (greedily) so that the 
 disk $L \cap P$ (in the notation of Lemma~\ref{l:shell_canonical}) consists of
 the intersection of the prism with the shelling complex and the intersection
 with not yet shelled prisms. (Not every choice of a prism yields such a disk,
 but there is always a valid choice.)

 For type $0$ shelling complex in the blocker case; we first remove the prisms
 subdividing $F_\searrow$ (which contains the square $S_1$), using
 Lemma~\ref{l:shell_canonical}, and then we continue
 with other prism analogously as in the previous case. (For example, it is
 possible to shell them in topdown direction according to
 Figure~\ref{f:blocker_wall}). 

 For type $i$ and type $n$ shelling complex, we again shell using
 Lemma~\ref{l:shell_canonical} on prism of $\Upsilon$ in topdown direction
 starting with the prisms subdividing $F_\searrow$ until we reach the
 exceptional piece $P_6$. We do not remove $P_6$ as it is a part of shelling
 complex and it is supposed to be kept. Instead, we continue in bottom up
 direction starting with the prism containing the triangle $d'_*e_*e'_*$ until
 we reach $P_6$ from bottom (again we do not remove it). This finishes the
 shelling.
\end{proof}

\begin{lemma}
  \label{l:1house_blocked}
  Assume that $B$ is a triangulated $3$-ball. Assume that $H \subseteq B$ where
  $H$ is a thick 1-house. Assume that every face of $\partial H$ which is
  not in the attachment complex is also a face of $\partial B$. In addition, assume that
  the triangle $\tau$ of $H$ is not in $\partial B$. Then there is no free
  tetrahedron (for shelling) of $B$ contained in $H$.
\end{lemma}

\begin{proof}
First we exclude all tetrahedra of $H$ which are not contained in the inner
  part of the lower wall. Then we exclude the remaining ones by a case analysis
  depending on which case of the attachment complex we consider.

  Let $\Delta$ be a tetrahedron contained in a cube (brick) $C$ which is not in the
  lower wall. Note that by the construction of the attachment complex, every
  vertex of $C$ and $\Delta$ as well belongs to $\partial B$. 

  Now we distinguish several cases depending how $C$ intersects $\partial H$.

\medskip

 The first case is that $C \cap \partial H$ consists of two opposite 
  squares of $C$. This case occurs if $C$ is an `inner' cube of some thickened
  $2$-face.
 In this case $\Delta$ has to meet both the squares in some vertex and thus
  $\Delta \cap \partial B$ is
  disconnected. (Here we use that $\Delta \cap \partial B \subseteq C \cap \partial
  B \subseteq C \cap \partial H$. We will use the same inclusions also in the
  other cases  without explicit notice.)
  Thus $\Delta \cap \partial B$ cannot be a disk. Therefore $\Delta$ is not
  free. 

\medskip

  The second case is that $C \cap \partial H$ consists of two squares sharing an edge
  together with an edge $e$ which avoids the two squares. This case occurs if
  $C$ is an `inner' cube of some thickened edge where two other $2$-faces meet
  perpendicularly. Here we want to check that $\Delta$ meets both $e$ and the two
  squares. As soon as we check this we get that $\Delta \cap \partial B$ cannot be a
  disk similarly as in the previous case. It is sufficient to check that
  $\Delta$
  meets $e$ as the two squares cover six vertices of $C$.

  For checking that $\Delta$ meets $e$ we need to distinguish two subcases. The
  first subcase is that $C$ contains at least one vertex of the 
  first five vertices (in our total order inducing the triangulation of $H$). There are four such
  cubes: these are exactly the four cubes meeting $F$ in an edge outside the
  lower wall (see Figure~\ref{f:bing_cubes}, right). By inspection of the cases
  in each of the four cases the edge $e$ meets the first or the second vertex
  in our total order and this vertex $v$ is the first vertex of $C$ in the
  total order. Therefore $C$ is triangulated as a cone with apex $v$ which
  implies that $\Delta$ contains $v$.

  In the second subcase, we assume that $C$ avoids the first five vertices.
  Then the both vertices of $e$ satisfy either (R2) or (R3) of our total order.
  On the other hand the remaining six vertices of $C$ fall only into the rule 
  (R4) by an inspection of possible neighboring cubes. (Here we use our
  assumption from the early stage of the construction 
  that two disjoint edges of the original 1-house thicken to two
  disjoint unions of cubes.) We conclude that $C$ is triangulated as a cone
  where one of the vertices of $e$ is the apex. Thus $\Delta$ meets $e$.
\medskip

 The third case is that $C$ meets $\partial H$ in three components, one of them
  is a square and the remaining two are edges. This occurs if $C$ is an `inner'
  cube of some edge where three $2$-faces meet. This is an easy case as $\Delta$ has
  to meet at least two such components and thus $\Delta \cap \partial B$ is
  disconnected.

\medskip

 The forth case is that $C$ meets $\partial H$ in two or more components; 
 at least one of them is a vertex.
 This case occurs if $C$ is a thickening of a vertex $w$ of the original
 1-house where three $2$-faces meet perpendicularly and in addition the angle
 at $w$ is $\pi/2$ in each of the $2$-faces; or if four or more
 $2$-faces meet at $w$ without any condition on angles. 
 We aim to show that $C$ is triangulated as a
 cone with some apex $v$ where $v$ is one of the vertices of $C$ that form a
 single component of $C \cap \partial H$. This implies that $\Delta$ intersects (at
 least) two  disjoint boundary components of $C \cap \partial H$ and therefore
 $\Delta \cap
 \partial B$ is disconnected. We will distinguish two subcases.

 In the first subcase, we assume that $C$ contains one of the first five
 vertices of our total order. By inspecting the bricks that contain one of
 these vertices (see Figure~\ref{f:bing_cubes}, right), we deduce that $C$ is
 either the brick left of $F$ and there is only a one option for $v$ which is
 the first vertex of the total order, or $C$ is the brick right of $F$ and $v$
 is (necessarily) the second vertex of the total
 order. In both cases $C$ is triangulated as a cone with apex $v$.

 In the second subcase, we assume that $C$ does not contain any of the first
 five vertices of our total order. Then each vertex of $C$ which forms a single
 boundary component of $C \cap \partial H$ qualifies for (R2). Let $v$ be the
 first one in our total order among these vertices. In this subcase no vertex
 of $C$ qualifies for (R1). However, also by inspection of cubes of $C$
 intersecting $C$, no other vertex of $C$ qualifies for (R2) apart from those we
 already know due to $C$. (Here we again use that two
 disjoint edges of the original 1-house thicken to two disjoint unions of
 cubes. In particular, no cube intersecting $C$ is a thickening of a vertex.) 
 Therefore $C$ is triangulated as a cone with apex $v$. 
 
\medskip

 The fifth case is that $C$ meets two boundary components; one of them is a
 union of two edges on three vertices and the second one is a union of a square
 and an edge on five vertices. This occurs if $C$ is a thickening of a vertex
 $w$ of the original 1-house where three $2$-faces meet perpendicularly and in
 addition the angle at $w$ is $3\pi/2$ at one of the $2$-faces.
 If $\Delta$ meets vertices of both components, then $\Delta \cap \partial
 B$ is disconnected and we are done. If $\Delta$ meets vertices only of a single
 component, then this is necessarily the component consisting of the union of
 the square and the edge on five vertices. As $\Delta$ is $3$-dimensional, it
 necessarily contains a vertex $v$ outside the square. Then $\Delta \cap \partial B$
 is not necessarily disconnected but the star of $v$ in $\Delta \cap \partial B$ is
 at most $1$-dimensional, thus $\Delta \cap \partial B$ cannot be a disk.

\medskip

 The sixth and the last case occurs if $C = F$. Here we use the notation from
 Figure~\ref{f:F_triangulate}, left. Here we crucially use that $\tau$ does not
 belong to $\partial B$ due to assumptions of the lemma 
 (though it belongs to $\partial H$). This means that $\Delta \cap \partial B$ is a
 subcomplex of the complex $A$ formed by the triangles $abe$, $aed$, $ded'$,
 $b'e'd'$ and $a'b'd'$. On the other hand $\Delta$ has to contain $b'$ as $F$ is
 triangulated as a cone with apex $b'$. We conclude that $\Delta \cap \partial B$ is
 not a disk as $A$ does not contain a disk containing $b'$ and four vertices
 of $\Delta$ together. (The union of triangles $b'e'd'$ and $a'b'd'$ would be a disk
 containing four vertices but they are not vertices of a tetrahedron in $F$.)

\bigskip

Now let us assume that $\Delta$ is a tetrahedron in the lower wall. 

If we are in the splitter case, the incoming case, or the outgoing case
for the attachment complex, then the attachment complex is chosen in such a way
that all vertices are necessarily on $\partial B$.  Therefore, $\Delta$ is inside a cube $C$ which meets $\partial H$ in two squares and all vertices of $C$ belong
also to $\partial B$. This implies that the intersection of $\Delta$ with $\partial
B$ is disconnected. (In the outgoing case, we were changing the triangulation
of the lower wall so that $\Delta$ is in some triangular prism $P$. But this prism
$P$ is still in some cube $C$ of the original decomposition into cubes, thus
the reasoning above applies as well.)

In the blocker case, we have to be a bit more careful because the attachment
complex contains vertices $d_+$ and $d_-$ in the interior, and thus these
vertices need not be on $\partial B$. The tetrahedron $\Delta$ may be either inside
the polytope $P_6$ (from the blocker case) or inside some other cube of the lower wall. In the latter case, we conclude that $\Delta \cap
\partial B$ cannot be a disk in the same way as in the previous cases. Thus it
remains to consider the case that $\Delta$ is inside $P_6$. This means that
$\Delta$
contains $p$ as $P_6$ is triangulated as a cone with apex $p$. If $\Delta$ contains
at least two vertices from the front side of the lower wall (that is, vertices
among $b_+$, $b_-$, $d_+$, $d_-$, $e_+$ and $e_-$), then it meets the back side
in a vertex or an edge which is a component of $\Delta \cap \partial B$. Thus
$\Delta
\cap \partial B$ is not a disk. If $\Delta$ contains one vertex from the front side,
then this vertex cannot be $d_+$ or $d_-$ as $P_6$ is triangulated as a cone
with apex $p$ while there is no triangle on $\partial P_6$ containing $d_+$ or
$d_-$ and two vertices from the back side. Thus $\Delta \cap \partial B$ contains an
isolated vertex in the front side of the lower wall so it cannot be a disk.
Finally, $\Delta$ has to contain at least one vertex from each side (it contains $p$
from the back side while the whole tetrahedron would not fit into
$2$-dimensional back side). 
\end{proof}

\subsection{Thick turbine}
\label{ss:thick_turbine}

\paragraph{Thickening to a polytopal complex.}

Now we want to thicken the (thin) turbine from Subsection~\ref{ss:thin_turbine} to a polytopal complex. In order to
avoid any ambiguity, we emphasize that we thicken each blade as well
as the central triangle separately. Only after the thickening we will identify
some of the cubes to merge everything together.

We first thicken the blades as in Figure~\ref{f:thick_blade}, left 
(compare with Figure~\ref{f:blade_orthogonal}). 
When
compared with the 1-house, this time, the dimensions are fixed; that is, the
complex at the moment is the subcomplex of a cube subdivided into $9^3$ smaller
cubes (One of the cubes is bicolored in the picture. This will be used later on
when describing certain $J_i$ while it should be ignored now.) Then we further modify the decomposition as in
Figure~\ref{f:thick_blade}, right (for the moment, the names of some
distinguished vertices should be again ignored). Namely, we remove some of the cubes in the
roof and we replace them with either rectangular cuboids or some other prisms
as in the figure. The rectangular cuboid in the $i$th thick blade 
with three rectangles not attached to anything else will be denoted $F_i$ and
will play a similar role as the cube $F$ in case of the thick $1$-house.

\begin{figure}
  \begin{center}
    \includegraphics[page=9]{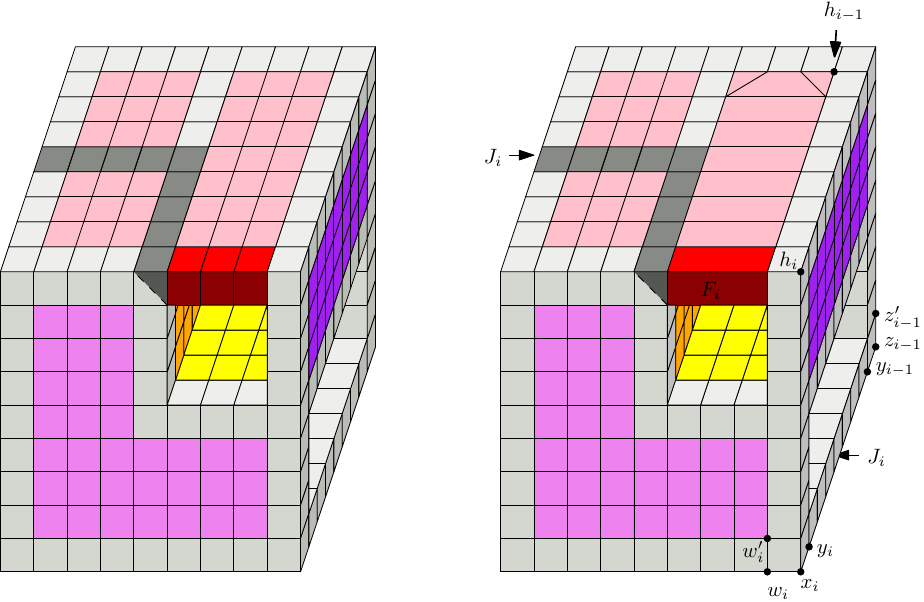}
  \end{center}
  \caption{A thick blade in orthogonal scheme}
  \label{f:thick_blade}
\end{figure}

Now we thicken the central triangle; see Figure~\ref{f:thick_triangle} (compare
with Figure~\ref{f:turbine_triangle}). First,
we build a $2$-dimensional complex as in the left picture which corresponds to
a thickening of the central triangle in $2$-space. We denote the vertices
$w_1, \dots, z_1, w_2, \dots, z_3$ as in the left picture. 
Then we take the product of this 2-complex with
the interval obtaining a $3$-dimensional complex as in the right picture. This
means that the thick central triangle is decomposed into prisms over
triangles or quadrilaterals. We think of the earlier $2$-complex as the bottom
side of the $3$-dimensional thick central triangle. On the top side we
denote vertices $w'_1, \dots, z'_1, w'_2, \dots, z'_3$ just above $w_1, \dots,
z_1, w_2, \dots, z_3$ (only a few of these vertices are marked in the right
picture).

\begin{figure}
  \begin{center}
    \includegraphics[page=10]{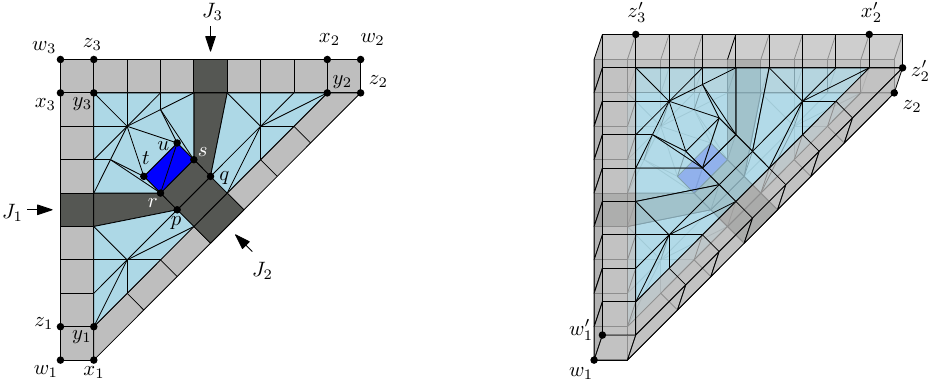}
  \end{center}
  \caption{Thickening the central triangle.}
  \label{f:thick_triangle}
\end{figure}

Now we merge the thick blades and the thick central triangle. Namely, we merge
the bottom right row (with respect to Figure~\ref{f:thick_blade}) of nine cubes
of the $i$th blade with the row of nine cubes in the thick central triangles
connecting the cubes $w_ix_iy_iz_iw'_ix'_iy'_iz'_i$ and
$w_{i-1}x_{i-1}y_{i-1}z_{i-1}w'_{i-1}x'_{i-1}y'_{i-1}z'_{i-1}$ where the
indices are considered modulo 3. Some of the merged vertices are depicted in
Figure~\ref{f:thick_blade}, right. This uniquely determines how are the two
rows of cubes merged. Note that the cube $w_ix_iy_iz_iw'_ix'_iy'_iz'_i$ belongs
to two thick blades, namely to the $i$th one and the $(i+1)$st (again modulo
$3$). We also merge the pillars above this cube in the $i$th thick blade and
the $(i+1)$st one.

\paragraph{Triangulation.} Now we aim to describe a triangulation of our
polytopal thick turbine. Our aims are similar as in the case of thick 1-house;
we want to avoid free faces in the triangulation with the
exception of the three cuboids $F_i$; for these cuboids, we still want some
control. Therefore, it should not be surprising that the description of the
triangulation will be very similar to the case of thick 1-house.
We again provide total order on the vertices of the union of our polytopes. 
Then we use the corresponding canonical triangulation with respect to this
order. 

For setting up the order, we need a small piece of notation: In the $i$th blade
we denote some of the vertices $a_i, \dots, h_i, a'_i, \dots h'_i$ is in
Figure~\ref{f:F_neighborhood_blade}, left; for the moment, the diagonals should
be ignored. For comparison with more global
picture, $h_i$ and $h_{i-1}$ are also marked in Figure~\ref{f:thick_blade}.
Now, we set up the following rules for our total order (essentially the same
as in the case of the thick 1-house).

\begin{figure}
  \begin{center}
    \includegraphics[page=24]{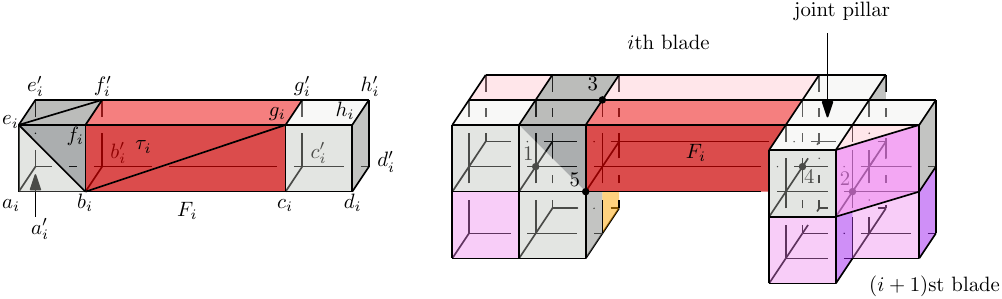}
    \caption{Left: Names of vertices of $F_i$ and the two cubes sharing a
    square with $F_i$ and also some important diagonals. Right: A slightly larger neighborhood of $F_i$ including
    some cubes of the $(i+1)$st blade. The numbers indicate the relative order
    of the first five vertices in this neighborhood with respect to our total
    order.}
    \label{f:F_neighborhood_blade}
  \end{center}
\end{figure}

\begin{enumerate}[(R1)]
  \item
The first fifteen vertices in our order are the vertices $a'_1, a'_2, a'_3,
    d_1, d_2, d_3, f'_1, f'_2, f'_3, c'_1, c'_2, c'_3, b_1, b_2$ and $b_3$ in
    this order. See Figure~\ref{f:F_neighborhood_blade}, right, for a relative
    order of these vertices in a neighborhood of $F_i$ (only this relative
    order will really matter). 
    \item
    Then we put in an arbitrary order those
vertices that are a component of intersection of some cube with the boundary of
the thick 1-house. 
    \item
    Next, among the remaining ones, we put in an arbitrary
order those that are a in an edge which is a component of intersection of some
cube with the boundary of the thick 1-house. 
    \item 
    Then we put the remaining
vertices in an arbitrary order. 
\end{enumerate}

Note that $b_ig_i$, $b_ie_i$ and $e_if_i'$ are diagonals in this triangulation. 
In particular, $b_ig_if_i$ is a triangle in this triangulation and we denote it
$\tau_i$; see Figure~\ref{f:F_neighborhood_blade}, left.

\paragraph{The attachment complex and the shelling complexes.} Similarly as in
the case of thick 1-house, we want to describe some auxiliary complexes on the
boundary of thick turbine. There will be one \emph{attachment complex} along
which the turbine will be glued to other gadgets. On the other hand, there will
be more \emph{shelling complexes}---for satisfiable formulas they will somehow
correspond to a choice of a satisfied literal in a clause.

Let $p, q, r, s, t, u$ be the six vertices on the bottom side of the thick
central triangle as in Figure~\ref{f:thick_triangle}, left. By $J_i$ we denote
the subcomplex of the boundary consisting of the triangle $b_if_ie_i$, the subdivided
square $e_if_if_i'e_i'$ and a collection of subdivided quadrilaterals (usually
squares) connecting the square $e_if_if_i'e_i'$ and the quadrilateral $pqsr$ as
in Figures~\ref{f:thick_blade} and~\ref{f:thick_triangle}. Part of this
collection is invisible in Figure~\ref{f:thick_blade}; here we believe that
Figure~\ref{f:turbine_orthogonal} clarifies which squares are chosen. The
quadrilateral $pqsr$ is also included in $J_i$. Altogether, $J_i$ consists of
$9 + 9 + 9 + 2$ subdivided quadrilaterals and one triangle. Next we set $J = J_1 \cup J_2
\cup J_3$ and the attachment complex consists of $J$, the subdivided
quadrilateral $rsut$ and the triangles $\tau_i$. See Figure~\ref{f:J_turbine}
for the attachment complex after reshaping (this shape will be useful
later on).

\begin{figure}
  \begin{center}
    \includegraphics[page=25]{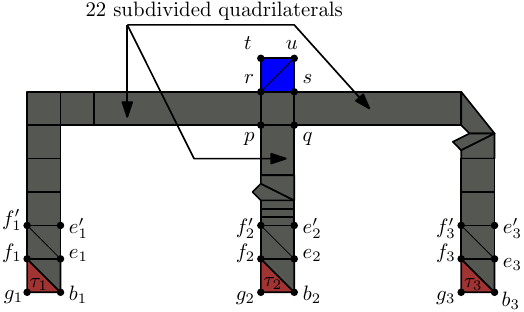}
    \caption{The attachment complex of a turbine.
    Some of the polygons are significantly reshaped in order to get this shape.}
    \label{f:J_turbine}
  \end{center}
\end{figure}

There are seven possible shelling complexes. There are obtained by removing
one, two or three pairs of vertices from the attachment complex among $\{b_1,
g_1\}, \{b_2, g_2\}$ and $\{b_3, g_3\}$.

\begin{lemma}
  \label{l:shelling_turbine}
  Assume that $K$ is a pure $3$-complex. Assume that $K = T \cup L$ where
  $T$ is the thick turbine and $T \cap L$ is one of seven possible shelling
  complexes of $T$. Then $K$ shells to $L$.
\end{lemma}

\begin{proof}
  The proof is very similar to the proof of Lemma~\ref{l:shelling_house}. We
  first consider the decomposition of the thick turbine into cubes, cuboids
  and triangular prisms as in Figures~\ref{f:thick_blade}, right
  and~\ref{f:thick_triangle}, right. We will shell the tetrahedra contained in
  these polytopes repeatedly using Lemma~\ref{l:shell_canonical} on these
  polytopes, one by one. Let $\Upsilon$ be the collection of 
  polytopes in our polytopal decomposition such that their intersection with
  the shelling complex is $2$-dimensional. First we intend to shell tetrahedra
  contained in prisms outside $\Upsilon$. Second, we intend to shell tetrahedra
  contained in prisms in $\Upsilon$.

  For the first step, we follow the
  collapses from the proof of collapsibility of the thin turbine; that is, from
  the proof of Lemma~\ref{l:thin_turbine}. See also
  Figures~\ref{f:collapsing_blade_1} and~\ref{f:collapsing_blade_3}. Note,
  due to our definition of the shelling complex, that at least one of the
  cuboids $F_1$, $F_2$ or $F_3$ is not in $\Upsilon$. Similarly
  as in the proof of Lemma~\ref{l:thin_turbine}, we can assume without loss of
  generality that $F_1$ does not belong to $\Upsilon$. In this case, the order
  of the first few polytopes for the application of
  Lemma~\ref{l:shell_canonical} is depicted in Figure~\ref{f:turbine_removals}.
  This corresponds to the first collapsing step in
  Figure~\ref{f:collapsing_blade_1} and a small part of the second one
  (removing the (subdivided) edge through which the collapse starts). The way
  how do we follow other collapses is analogous. Regarding $F_2$ and $F_3$ we
  remove them on the way if they do not belong to $\Upsilon$.

\begin{figure}
  \begin{center}
    \includegraphics[page=52]{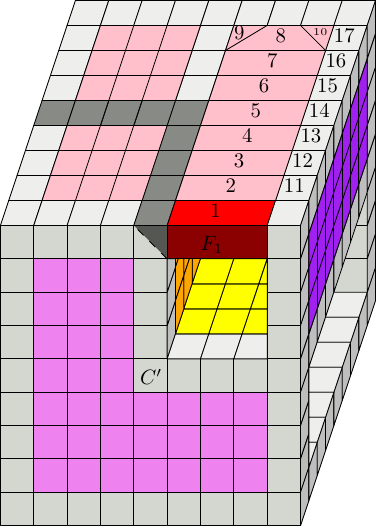}
    \caption{First few shellings of the thick turbine using
    Lemma~\ref{l:shell_canonical}. (The cube denoted $C'$ is not related to the
    proof of Lemma~\ref{l:shelling_turbine} but it is used in a proof of
    Lemma~\ref{l:turbine_blocked}.)}
    \label{f:turbine_removals}
  \end{center}
\end{figure}
  
  For the second step, when only the prisms from $\Upsilon$ remain, we again
  apply Lemma~\ref{l:shell_canonical} repeatedly. We first remove the prisms in
  directions from $F_1$, $F_2$ and $F_3$ and then the remaining three prisms
  meeting containing the quadrilateral $pqsr$, or the triangle $rsu$ or the 
  triangle
  $tru$ (see Figure~\ref{f:thick_triangle} for notation).

\end{proof}

\begin{lemma}
\label{l:turbine_blocked}
Assume that $B$ is a triangulated $3$-ball. Assume that $T \subseteq B$ where
  $T$ is a thick turbine. Assume that every face of $\partial T$ which is
  not in the attachment complex is also a face of $\partial B$. 
  In addition, assume that none of the triangles $\tau_1, \tau_2, \tau_3$ of $T$ is in $\partial B$. Then there is no free
  tetrahedron of $B$ contained in $T$.
\end{lemma}

\begin{proof} 
  As one easily checks the bricks $F_i$ do not contain any free tetrahedra.
  
  The key for checking the the tetrahedra in remaining prisms is the  following 
  observation: If $Q$ is a prism in the turbine
  and there is a face $G$ of $Q$, whose all vertices lie in different
  components of $Q\cap \partial B$ than the vertex $w$ with the smallest
  label in $Q$, then there is no free tetrahedron in $Q$. Indeed, in such a
  case any such tetrahedron contains $w$ and also some vertex of $G$ as
  any prism in the construction of the thick turbine is combinatorially equivalent
  to a triangular prism or to a cube. Therefore, the intersection of the
  tetrahedron with $\partial B$ is not connected.

  In particular, there are no free tetrahedra among the prisms containing the
  vertices 1,2,3,4,5.  In fact, the observation can be applied to almost all
  prisms; see Figures~\ref{f:F_neighborhood_blade}, \ref{f:thick_blade} and \ref{f:thick_triangle}. The only possible exception is the prism
  $C'$ of Figure~\ref{f:turbine_removals}, where $\partial B\cap C'$ consists
  of two components: a square $S$ with an attached edge $e$ and two edges
  $f_1,f_2$ sharing a vertex.  Moreover, this is an exception if and only if
  the minimal vertex lies in $S\cup e$, in which case some tetrahedra of $C'$
  intersect $\partial B$ in two components, and the remaining ones in $T\cup
  e$, where $T$ is a triangle inside $S$. In particular, the intersection is
  not a disk in either case.  
\end{proof}

\subsection{Conjunction cone}
\label{ss:conjunction_cone}

Now we describe a complex which will serve as a gadget for conjunction. It will
depend on a parameter $k$ which will indicate the number of `incoming
gadgets'.

First we take a (convex) $(k+1)$-gon $P$ with edges denoted $\varepsilon_0, \dots,
\varepsilon_k$ in clockwise order. By $w_+$ we denote the vertex shared by
$\varepsilon_0$ and $\varepsilon_1$ and by $w_-$ we denote the vertex shared by
$\varepsilon_0$ and $\varepsilon_k$. We triangulate $P$ as a cone with apex
$w_-$. For using our construction later on, it is
convenient to think of $\varepsilon_0$ as a vertical edge on the left side
and of the vertices of $P$ on a circle centered in the midpoint of
$\varepsilon_0$; see Figure~\ref{f:conjunction_cone}, left. The
\emph{conjunction cone} is then obtained as a cone with apex $d$ over $P$
triangulated as above; see Figure~\ref{f:conjunction_cone}, right. For any $i
\in \{1, \dots, k-1\}$, the unique tetrahedron containing the edge
$\varepsilon_i$ (and $d$ and $w_-$) will be denoted $\Delta_i$.

\begin{figure}
  \begin{center}
    \includegraphics[page=50]{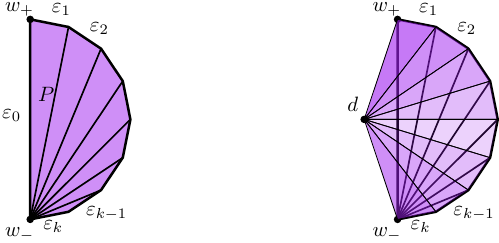}
    \caption{Left: The polygon $P$. Right: The conjunction cone.}
    \label{f:conjunction_cone}
  \end{center}
\end{figure}

The conjunction cone will be capable to perform a conjunction in the following
sense. It will be glued to the other gadgets so that all the edges containing
$d$ will be on the boundary. The triangle $d\varepsilon_0$ will be blocked (by
some other tetrahedron) and then we require that all triangles $d\varepsilon_i$ for
$i \in \{1, \dots, k\}$ have to be `unblocked' before shelling $\Delta_1$. This performs a `conjunction' on these
triangles. The precise statement is given in the following lemma.

\begin{lemma}
  \label{l:conjunction_cone}
Assume that $B$ is a triangulated $3$-ball. Assume that $C \subseteq B$ where
  $C$ is a conjunction cone (with notation as above). Assume that all the edges
  of $C$ containing $d$ belong to $\partial B$. Assume that the none of the
  triangles
  $d\varepsilon_i$ belongs to $\partial B$ for $i \in \{0, \dots, k\}$. Let
  $\Delta'_i$ be the (unique) tetrahedron of $B$ not belonging to $C$ containing
  $d\varepsilon_i$ for $i \in \{0, \dots, k\}$. 
  In any shelling down of $B$ (to a tetrahedron) before
  shelling the tetrahedron $\Delta_1$ either $\Delta'_0$ has to be shelled
  or all the tetrahedra $\Delta'_i$ for $i \in \{1, \dots, k\}$ have to be shelled. 
\end{lemma}

The lemma will be proved by a repeated application of the following
observation.

\begin{observation}
  \label{o:tetrahedron}
  Assume that $B''$ is a triangulated $3$-ball and $\Delta''$ is a tetrahedron of $B$
  with vertices $a$, $b$, $c$ and $d$. Assume that the edges $ad$, $bd$ and
  $cd$ belong to $\partial B''$. On the other hand, assume that 
  at least two of the triangles $abd$, $acd$, $bcd$ are 
  not contained in $\partial B''$.
  Then $\Delta''$ is not
  free.
\end{observation}

\begin{proof}
  Assume, without loss of generality, that $abd$ and $acd$ are not contained in
  $\partial B''$. Then $ad$ is an edge in $\partial B''$ which is not contained in any
  triangle of $\Delta''$ in $\partial B''$. Thus $\Delta''$ cannot meet $\partial B$
  in a $2$-ball.
\end{proof}

\begin{proof}[Proof of Lemma~\ref{l:conjunction_cone}]
Let us consider a shelling down of $B$ and let $B'_1$ be a $3$-ball obtained
  exactly one step before removing $\Delta_1$. (Note that $B'_1$ is a ball by
  Lemma~\ref{l:all_balls}.) If $B'_1$ does not contain
  $\Delta'_0$, then we are in one of the conclusions of the lemma. Thus it
  remains to assume that $B'_1$ contains $\Delta'_0$ and we want to deduce
  that no $\Delta'_i$ for $i \in \{1, \dots, k\}$ is contained in $B'_1$. 

  Now we apply Observation~\ref{o:tetrahedron} with ball $B'_1$ and tetrahedron
  $\Delta_1$ ($d$ remains the same). The edges of $\Delta_1$ containing $d$
  belong to $\partial B'_1$ because they belong to $\partial B$. Thus this assumption of the observation is
  satisfied. We also know that $\Delta_1$ is free in $B'_1$ as it is just about to be
  shelled. Therefore at most one of the three triangles 
  of $\Delta_1$ containing $d$ is not contained in $\partial
  B'_1$. This must be the triangle $d\varepsilon_0$ (as we assume that $B'_1$
  contains $\Delta'_0$). We conclude that the tetrahedra $\Delta'_1$ and
  $\Delta_2$ are shelled before $\Delta_1$.

  Next, let $B'_2$ be a $3$-ball obtained exactly one step before removing
  $\Delta_2$. We apply Observation~\ref{o:tetrahedron} with ball $B'_2$ and the
  tetrahedron $\Delta_2$. The edges of $\Delta_2$ containing $d$ belong to
  $B'_2$ because they belong to $B$. We also know that $\Delta_2$ is free in
  $B'_2$ as it is just about to be shelled. Therefore at most one of the three
  triangles of $\Delta_2$ containing $d$ is not contained in $\partial
  B'_2$. This must be the triangle shared by $\Delta_1$ and $\Delta_2$ as
  $\Delta_2$ is shelled before $\Delta_1$. We conclude that the tetrahedra
  $\Delta'_2$ and
  $\Delta_3$ are shelled before $\Delta_2$.
  
By repeating the argument above inductively, we obtain that $\Delta'_i$ and
  $\Delta_{i+1}$ are shelled before $\Delta_i$ for $i \in \{1, \dots, k-2\}$.
  In the final step, we obtain that $\Delta'_{k-1}$ and $\Delta'_k$ are shelled
  before $\Delta_{k-1}$. This proves the lemma.
\end{proof}

\section{The construction for shellability}
\label{s:shellability_construction}

In this section we assume that $\phi$ is an instance of the planar
monotone rectilinear 3SAT and we aim to build a triangulated $3$-ball $\cK_\phi$
such that $\cK_\phi$ is shellable if and only if $\phi$ is satisfiable. It will
follow immediately from the construction that $\cK_\phi$ is a
$3$-pseudomanifold and the $\cK_\phi$ can be built in polynomial time. We will
separately check in the following sections that $\cK_\phi$ is a $3$-ball and the equivalence `$\cK_\phi$
is shellable if and only if $\phi$ is satisfiable.' This will prove
Theorem~\ref{t:main}.

Because the construction of $\cK_\phi$ is quite complex, we collect
important objects used in the construction in Table~\ref{t:symbols}.

\begin{table}
\begin{tabular}{llll}
  symbol & introduced & in gadgets & remarks \\

  \hline
\hline
  $\aa_\ell$ & \S~\ref{ss:thick_variable} & $\cV_\vx$, $\cS_\ell$ & a vertex of
  the variable gadget \\
  $\aa_{\ell, \kappa}$ & \S~\ref{ss:thick_splitter} & $\cS_\ell$, $\cI_{\ell,
  \kappa}$ & a vertex of the splitter house \\
  $\bb_\ell$ & \S~\ref{ss:thick_variable} & $\cV_\vx$, $\cS_\ell$, $\cB_i$ & a vertex of the variable gadget \\
  $\bb_{\ell, \kappa}$ & \S~\ref{ss:thick_splitter} & $\cS_\ell$, $\cI_{\ell,
  \kappa}$ & a vertex of the splitter house \\
  $\cc_\ell$ & \S~\ref{ss:thick_variable} & $\cV_\vx$, $\cB_i$ & a vertex of the variable gadget \\
  $\cc_{\ell, \kappa}$ & \S~\ref{ss:thick_splitter} & $\cS_\ell$, $\cI_{\ell,
  \kappa}$, $\cT_0$ & a vertex of the splitter house \\
  $\dd_\ell$ & \S~\ref{ss:thick_variable} & $\cV_\vx$, $\cS_\ell$, $\cB_i$,
  $\cT_0$ & a vertex of the variable gadget \\
  $\dd_{\ell, \kappa}$ & \S~\ref{ss:thick_splitter} & $\cS_\ell$, $\cI_{\ell,
  \kappa}$, $\cT_0$ & a vertex of the splitter house \\
  $\dd$ & \S~\ref{ss:conjunction_gadget} & $\cA, \cB_n$, all $\cO_\kappa$ & the
  apex of the conjunction cone \\
  $\ee_\ell$ & \S~\ref{ss:thick_splitter} & $\cS_\ell$, $\cB_i$,
  $\cT_0$ & a vertex of the splitter house \\
  $\rr_\ell$ & \S~\ref{ss:template} & $\cS_\ell, \cT_0$ & $\rr_\vx$ and $\rr_{\neg \vx}$ are
  points in $\RR_\vx$ \\
  $\rr_\kappa$ & \S~\ref{ss:clause_gadget} & $\cC_\kappa, \cO_\kappa, \cT_0$ &
  a vertex of the clause gadget \\
  $\ss_\kappa$ & \S~\ref{ss:clause_gadget} & $\cC_\kappa, \cO_\kappa, \cT_0$ &
  a vertex of the clause gadget \\
  $\ttt_\kappa$ & \S~\ref{ss:clause_gadget} & $\cC_\kappa, \cO_\kappa$ &
  a vertex of the clause gadget \\
  $\uu_\kappa$ & \S~\ref{ss:clause_gadget} & $\cC_\kappa, \cO_\kappa$ &
  a vertex of the clause gadget \\
  $\vv_\vx$ & \S~\ref{ss:template} & none & the auxiliary midpoint of $\RR_\vx$
  (not a vertex of $\cK_\phi$) \\
  $\vv_\kappa$ & \S~\ref{ss:template} & none & an auxiliary point in
  $\RR_\kappa$ (not a vertex of $\cK_\phi)$ \\
  $\vv_\pm$ & \S~\ref{ss:template} & $\cA$, $\cT_0$, some $\cO_\kappa$ & rightmost point
  of the $x$-axis in the template\\
  $\ww_{\ell, \kappa}$ & \S~\ref{ss:thick_incoming} & $\cI_{\ell,  \kappa}$,
  $\cC_\kappa$, $\cT_0$ & a vertex of the incoming house \\
  $\ww_+$, $\ww_-$ & \S~\ref{ss:outgoing_house} & $\cA$, $\cB_n$, $\cT_0$ some
  $\cO_\kappa$ & vertices of $\PP$ \\
  $\xx_{\ell, \kappa}$ & \S~\ref{ss:thick_incoming} & $\cI_{\ell,  \kappa}$,
  $\cC_\kappa$, & a vertex of the incoming house \\
  $\yy_{\ell, \kappa}$ & \S~\ref{ss:thick_incoming} & $\cI_{\ell,  \kappa}$,
  $\cC_\kappa$, & a vertex of the incoming house \\
  $\zz_{\ell, \kappa}$ & \S~\ref{ss:thick_incoming} & $\cI_{\ell,  \kappa}$,
  $\cC_\kappa$, $\cT_0$ & a vertex of the incoming house \\
  $\ggamma_{\ell,\kappa}$ &\S~\ref{ss:template} & none & an auxiliary curve in
  the template \\
  $\ddelta_{\kappa}$ &\S~\ref{ss:template} & none & an auxiliary curve in
  the template \\
  $\vvarepsilon_\pm$ &\S~\ref{ss:template} & none & an auxiliary segment in
  the template \\
  $\PP$ & \S~\ref{ss:outgoing_house} & $\cA$, $\cT_0$ & a polygon in the
  template to which we glue $\cA$ \\
  $\PP_i$ & \S~\ref{ss:blocker_house} & $\cB_i$ & a distinguished 
  polytope inside $\cB_i$ \\	
  $\RR_\vx$ & \S~\ref{ss:pmr3sat} & none & rectangle in the template
  representing $\vx$\\
  $\RR_\kappa$ & \S~\ref{ss:pmr3sat} & none & rectangle in the template
  representing $\kappa$\\
  $\cA$ & \S~\ref{ss:conjunction_gadget} & & the conjunction gadget \\
  $\cB_i$ & \S~\ref{ss:blocker_house} & & the blocker house \\
  $\cC_\kappa$ & \S~\ref{ss:clause_gadget} & & the clause gadget \\
  $\cI_{\ell, \kappa}$ & \S~\ref{ss:thick_incoming} & & the incoming house \\
  $\cO_{\kappa}$ & \S~\ref{ss:outgoing_house} & & the outgoing house \\
  $\cS_\ell$ & \S~\ref{ss:thick_splitter} & & the splitter house \\
  $\cT_0$ & \S~\ref{ss:template_gadget} & & the template gadget \\
  $\cV_\vx$ & \S~\ref{ss:thick_variable} & & the variable gadget \\
\end{tabular}
  \caption{Symbols used throughout the construction of $\cK_\phi$. Here $\vx$ is a variable,
  $\ell$ is a literal which is either $\vx$ or $\neg \vx$ and $\kappa$ is a
  clause. If applicable, $\ell \in \kappa$. Remarks for vertices usually 
  emphasize only the gadget where the vertex was introduced.}
\label{t:symbols}
\end{table}

\subsection{The template gadget}
\label{ss:template_gadget}

We start with the template as in Subsection~\ref{ss:template} and its
bounding box and we use the notation 
$\vv_{\vx}, \vv_{\kappa}, \vv_\pm, \rr_{\ell}, \RR_{\vx}, \RR_{\kappa}, \ggamma_{\ell,
\kappa}, \ddelta_{\kappa}$ and $\vvarepsilon_\pm$ for a variable $\vx$, a literal
$\ell$ and a clause $\kappa$ in the same way as in
Subsections~\ref{ss:pmr3sat} and~\ref{ss:template}.\footnote{On the other hand, the notation used
beyond Subsection~\ref{ss:template} will be completely redefined.} (See
also Figure~\ref{f:template}.)

We will build a certain triangulation $T$ of the bounding box
(which will be gradually specified during the whole construction). Then we
thicken this by taking the product with the interval $T \times I$. By this we
mean, that we get $T \times I$ as a polytopal complex where each polytope is a
prism over a triangle. We triangulate these prisms so that we consider the
vertices of $T \times I$ in an arbitrary order and take the canonical
triangulation with respect to this order. This way we obtain the \emph{template
gadget} $\cT$. Note that $\cT$ is a triangulated $3$-ball. 

We use the same conventions as in Subsection~\ref{ss:template} regarding
positioning the template gadget. If we consider $I$ as interval $[0,1]$, then
$\cT$ is in between the planes $z=0$ and $z=1$. Considering the negative
direction of $z$ as the direction `in front of', the front side of the template
is $\cT \times \{0\}$ and we will denote it $\cT_0$. Note that $\cT_0$ is
a triangulated disk isomorphic to $T$ mentioned earlier. We will glue other
gadgets only to the front side $\cT_0$. 

\subsection{The variable gadget}
\label{ss:thick_variable}
Now for every variable $\vx$ we take a copy of the triangular prism
and we denote
it $\cV_{\vx}$. This will be our \emph{variable gadget}. We will use the new
notation for some of the vertices of the variable gadget. This notation is set
up to $a' \to \aa_{\neg \vx}$, $b' \to \bb_{\neg \vx}$, $c' \to \cc_{\neg \vx}$, $d'
\to \dd_{\neg \vx}$, $a \to \aa_{\vx}$, $b \to \bb_{\vx}$, $c \to \cc_{\vx}$ and $d \to
\dd_{\vx}$ where $y \to \yy$ stands for replacing a vertex $y$ of the diamond prism
with a vertex $\yy$ of the variable gadget. We do not introduce a new notation
for the vertex $e$ of the triangular prism as we will not use it anymore.
See Figure~\ref{f:variable_gadget}, left for the variable gadget
in current notation.

\begin{figure}
  \begin{center}
    \includegraphics[page=38]{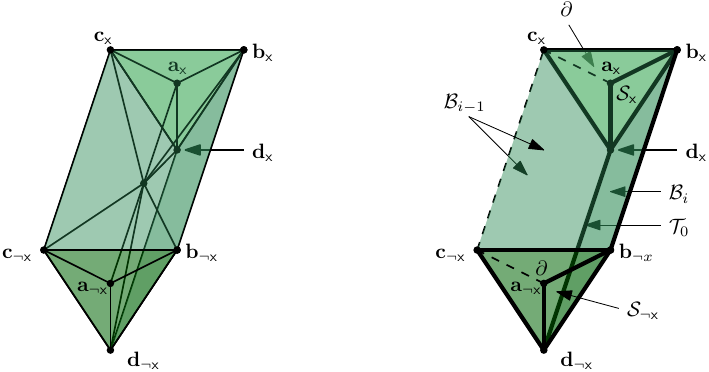}
  \end{center}
  \caption{Left: The notation for the vertices of the variable gadget. Right:
  Intended identifications of the variable gadget with other gadgets. At the
  moment, other gadgets, except $\cT_0$ are not introduced; however this
  picture will become useful when introducing the other gadgets. Dashed lines
  indicate that the neighboring pieces belong to the same gadget. The symbol
  $\partial$ stands for the piece of the boundary of the variable gadget that
  will be also a boundary of $\cK_\phi$.}
  \label{f:variable_gadget}
\end{figure}

We also recall that in the three rectangles of the variable gadget we are
flexible how to choose the diagonals of these rectangles which influences the
final construction of the variable gadget. We do not specify this now but we
will do so later on when gluing the variable gadget to other gadgets along
these rectangles.

Now we glue the variable gadget to the template, namely to $\cT_0$. We glue
$\cV_{\vx}$ to $\cT_0$ along the segment $\dd_{\vx} \dd_{\neg \vx}$; it remains to
position this segment in $\cT_0$. We place $\dd_{\vx}$ so that it is an interior
point of the segment $\vv_{\vx}\rr_{\vx}$ (of the template) and symmetrically
$\dd_{\neg \vx}$ is an interior point of $\vv_{\vx}\rr_{\neg \vx}$ (intuitively
both $\dd_{\vx}$
and $\dd_{\neg \vx}$ are sufficiently close to $\vv_{\vx}$). Note that the point
$\vv_{\vx}$ serves as an auxiliary point in the template but it will not be a part
of the triangulation of $\cT_0$.

\subsection{The splitter house}
\label{ss:thick_splitter}

Now for every literal $\ell$ we add a \emph{splitter house} $\cS_\ell$ to our
construction. Namely, $\cS_\ell$ is a copy of the thick 1-house triangulated according to the splitter case with the number of branches (parameter $k$) equal to number
of paths $\ggamma_{\ell,\kappa}$ emanating from $\rr_\ell$. Now we explain
how to glue $\cS_\ell$ to the current stage of our construction.

First we equip some of the vertices of $\cS_\ell$ with new names. These are the names
that will be used globally throughout the whole construction, and some of the
vertices may be already used. If we use a previously
used name, then we mean to identify two vertices with the same name.
Namely we set up
$d' \to \aa_\ell$, $e' \to \bb_\ell$, $e \to \dd_\ell$, $g' \to \ee_\ell$, $b
\to \rr_\ell$,
$a_i \to \aa_{\ell, \kappa_i}$, $b_i \to \bb_{\ell, \kappa_i}$, $c_i \to
\cc_{\ell, \kappa_i}$, $d_i
\to \dd_{\ell, \kappa_i}$ where $y \to \yy$ stands for renaming the vertex $y$ in the notation used
for the definition of the thick $1$-house with vertex $\yy$ in the notation of the construction
(see Figures~\ref{f:F_triangulate}
and~\ref{f:splitter_wall} for the old notation and
Figure~\ref{f:splitter_attachment}, left, for the attachment complex of $\cS_\ell$ in the
new notation). For cases such as $a_i
\to \aa_{\ell, \kappa_i}$ we temporarily denote the paths $\ggamma_{\ell,\kappa}$
emanating from $\rr_\ell$
from left to right as $\ggamma_{\ell,\kappa_1}, \dots \ggamma_{\ell,
\kappa_k}$.
We emphasize that
$\aa_\ell, \bb_\ell$ and $\dd_\ell$ are already known vertices of the variable
gadget and $\rr_\ell$ is an auxiliary point on $\cT_0$ while $\ee_\ell,
\aa_{\ell, \kappa_i},  \bb_{\ell, \kappa_i}, \cc_{\ell, \kappa_i}$ and
$\dd_{\ell, \kappa_i}$ are
completely new vertices. Finally, regarding the notation, some of the names we
use only locally in this subsection and we do not need a global notation for
them; this regards vertices $g, h, f, c$ (here the notation agrees in both
cases) and also the rectangles $R_1$ and $R_2$ in the original notation on
$\cS_\ell$.

\begin{figure}
  \begin{center}
    \includegraphics[page=43, scale=.95]{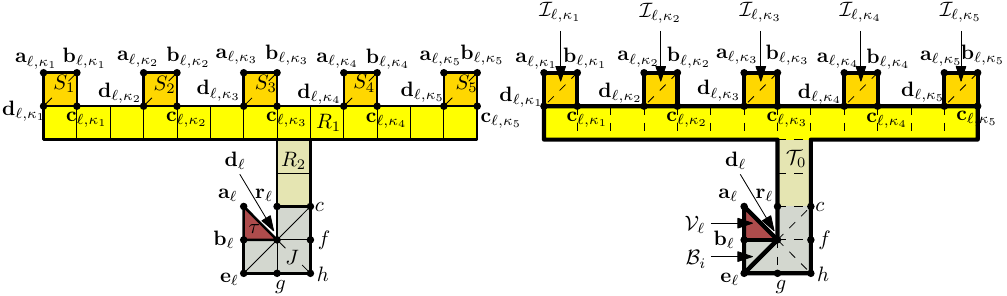}
    \caption{The attachment complex of $\cS_\ell$ (the picture is mirrored when
    compared with Figure~\ref{f:splitter_wall}). Left: The notation of
    vertices.  Right: Intended identifications with other gadgets; only $\cT_0$ and
    $\cV_\ell$
    have been introduced so far.}
    \label{f:splitter_attachment}
  \end{center}
\end{figure}

Now we prepare $\cT_0$ a little bit for gluing; see Figure~\ref{f:gluing_splitter}
(for $\ell = \vx$) while following this preparation.

\begin{figure}
  \begin{center}
    \includegraphics[page=19]{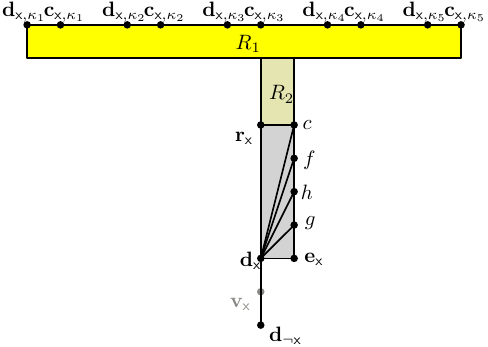}
    \caption{Gluing the splitter house to $\cT_0$ in case $\ell = \vx$. (The
    point $\vv_{\vx}$ is
    displayed in gray as it is part of $\cT_0$ but not a vertex of the
    triangulation we build.)}
    \label{f:gluing_splitter}
  \end{center}
\end{figure}

We draw an axis aligned rectangle in $\cT_0$ with vertices $\dd_{\ell},
\ee_{\ell}, c, \rr_\ell$ right of the segment $\dd_\ell \rr_\ell$. The
rectangle is sufficiently thin so that it does not intersect anything else. 
Then we further subdivide the edge $\ee_\ell c$ into a path $\ee_\ell g h f c$
of length four. Now we triangulate our rectangle as a cone with apex $\dd_\ell$ using also all vertices of the path $\ee_\ell g h f c$.

Next we draw an axis aligned copy of the rectangle $R_2$ (from the definition
of the attachment complex in the splitter case) above $\rr_\ell c$ if $\ell$ is
a positive literal or below $\rr_\ell c$ if $\ell$ is a negative literal. Then
we draw an axis aligned copy of the rectangle $R_1$ (again from the definition
of the attachment complex in the splitter case) above $R_2$ if $\ell$ is
positive and below $R_2$ if $\ell$ is negative. We require that $R_1$ and $R_2$
meet in the same way as in the attachment complex. We triangulate $R_1$ and
$R_2$ in the same way as they are triangulated in the attachment complex. 

All triangulations that we have just described (including the triangulation of
the rectangle $\dd_\ell \ee_\ell c \rr_\ell$) will be part of the triangulation
of $\cT_0$.

Now, we identify vertices of $\cS_\ell$ with the same named vertices of the
previous construction (on $\cT_0$ or the variable gadget). The triangulation
described on $\cT_0$ exactly matches the relevant part of $\cS_\ell$. This also
introduces vertices $\cc_{\ell,\kappa_i}$ and $\dd_{\ell, \kappa_i}$ on $\cT_0$ as in
Figure~\ref{f:gluing_splitter}.

Considering the attachment complex of $\cS_\ell$, at the moment, the triangle
$\tau$ (in the notation of thick 1-house) is identified with the triangle
$\aa_\ell \bb_\ell \dd_\ell$ of the variable gadget (compare
Figures~\ref{f:variable_gadget} and~\ref{f:splitter_attachment}, right). The triangle of the attachment
complex next to $\tau$ ($\bb_\ell \dd_\ell \ee_\ell$ in the notation of the
construction) as well as the subdivided squares $\aa_{\ell, \kappa_i}
\bb_{\ell, \kappa_i}
\cc_{\ell, \kappa_i} \dd_{\ell, \kappa_i}$ in the lower wall are not attached to any part of the previous construction. This will be completed
later on so that $\cS_\ell$ will meet the rest of the construction indeed in the
attachment complex.

Considering the shelling complexes of $\cS_\ell$, note that the first shelling
complex is the induced
subcomplex of the attachment complex induced by all vertices of the attachment
complex except $\aa_\ell$. In other words, it is the part of the attachment
complex which is attached to the gadgets $\cB_i$, $\cT_0$, $\cI_{\ell, \kappa_i}$
(sometimes not yet introduced) as in Figure~\ref{f:splitter_attachment}, right.
The second shelling complex is only the part of the attachment complex which is
attached to $\cT_0$ and $\cI_{\ell, c_i}$.

\subsection{The incoming house}
\label{ss:thick_incoming}

For every curve $\ggamma_{\ell,\kappa}$ we define one incoming house
$\cI_{\ell,\kappa}$
(possibly with multiplicities if there are more curves $\ggamma_{\ell,\kappa}$). We
take a thick $1$-house triangulated according to the incoming case where
the number of crossing annuli equals the number of crossings of our
$\ggamma_{\ell,\kappa}$ with curves $\ddelta_{\kappa'}$ (for arbitrary $\kappa'$). 

We start with renaming the vertices of $\cI_{\ell, \kappa}$ using the same
notation $y \to \yy$ as in the splitter case (or for the variable gadget). We
set up $d' \to \aa_{\ell,\kappa}, e' \to \bb_{\ell,\kappa}, g' \to
\cc_{\ell,\kappa}, e \to \dd_{\ell,\kappa}, w \to \ww_{\ell,\kappa}, x \to
\xx_{\ell,\kappa}, y \to \yy_{\ell,\kappa}$ and $z \to \zz_{\ell,\kappa}$.
(See Figures~\ref{f:F_triangulate}
and~\ref{f:incoming_wall} for the old notation and
Figure~\ref{f:incoming_attachment}, left, for the new notation.)
The vertices $\aa_{\ell,\kappa},  \bb_{\ell,\kappa}, \cc_{\ell,\kappa}$ and
$\dd_{\ell,\kappa}$ are vertices of the splitter house while $\ww_{\ell,\kappa},
\xx_{\ell,\kappa},
\yy_{\ell,\kappa}$ and $\zz_{\ell,\kappa}$ are newly introduced. Note that the diagonal
$\bb_{\ell,\kappa
} \dd_{\ell,\kappa}$ of the square $\aa_{\ell,\kappa} \bb_{\ell,\kappa}
\cc_{\ell,\kappa} \dd_{\ell,\kappa}$ exactly
matches the edge $ee'$ in the original notation of $\cI_{\ell,\kappa}$. 
We use the local notation for vertices $b, c, f, g, h$ of $\cI_{\ell,\kappa}$ as well
as for the rectangles $R_1$ and $R_2$ (in the triangulation of the lower wall
in the incoming case).

\begin{figure}
  \begin{center}
    \includegraphics[page=44, scale=.91]{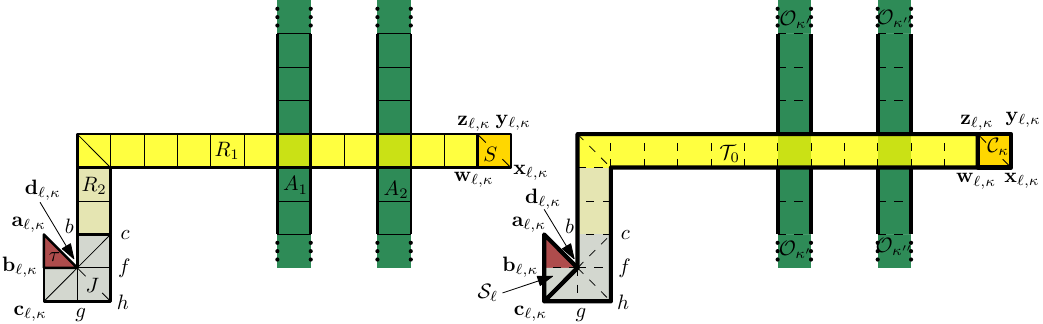}
    \caption{The attachment complex of the incoming house $\cI_{\ell,\kappa}$ (the
    picture is mirrored when compared with Figure~\ref{f:incoming_wall}. Left:
    The notation of vertices. Right: Intended identifications with other
    gadgets; only $\cT_0$ and $\cS_\ell$ have been introduced so far.}
    \label{f:incoming_attachment}
  \end{center}
\end{figure}

The choice of the notation of vertices already predetermines how $\cI_{\ell,
\kappa}$
will be glued to the splitter house $\cS_\ell$ (compare
Figures~\ref{f:splitter_attachment} and~\ref{f:incoming_attachment}, right). Therefore, it remains to explain
how to glue $\cI_{\ell,\kappa}$ to $\cT_0$; see Figure~\ref{f:gluing_incoming}.
For simplicity of the description we will assume that $\ell$ is a positive literal,
$\ell = \vx$ where $\vx$ is a variable. The case $\ell = \neg \vx$ is analogous in a
mirror symmetric fashion.

\begin{figure}
  \begin{center}
    \includegraphics[page=22]{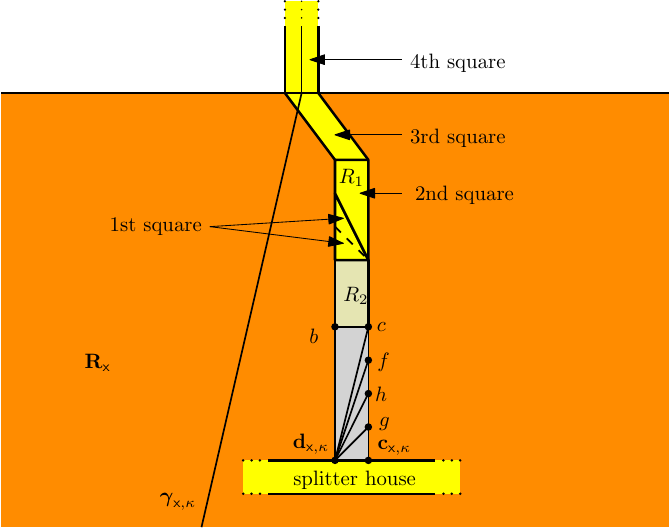}
    \caption{Gluing the incoming house to $\cT_0$, part inside $\RR_{\vx}$.}
    \label{f:gluing_incoming}
  \end{center}
\end{figure}

First we draw an axis aligned rectangle in $\cT_0$ with vertices
$\dd_{\vx,\kappa}\cc_{\vx,\kappa}cb$ above the edge
$\dd_{\vx,\kappa}\cc_{\vx,\kappa}$; we
subdivide the edge $\cc_{\vx,\kappa}c$ to a path $\cc_{\vx,\kappa} ghfc$; and
we triangulate the rectangle as a cone with apex $\dd_{\vx,\kappa}$. This is very
analogous to gluing the splitter house up to the notation. Then we glue the
rectangle $R_2$ above this rectangle as in Figure~\ref{f:gluing_incoming}. 

Now we aim to glue $R_1$. Here we describe how to glue individual squares of
$R_1$ (from left to right when referring to Figure~\ref{f:incoming_attachment}). The
first two squares are directly above $R_2$ while they straighten the bend
between $R_1$ and $R_2$. (Compare Figures~\ref{f:incoming_attachment} and
Figure~\ref{f:gluing_incoming}.) So far we assume that the drawing is in
sufficiently small neighborhood of $\rr_{\vx}$ so that we are still inside the
rectangle $\RR_{\vx}$. We draw the third square of $R_1$ as a parallelogram above
the second square with the `upper edge' inside the boundary of $\RR_{\vx}$ so that
$\ggamma_{\vx,\kappa}$ meets this `upper edge' in its midpoint. Then the forth square
of $R_1$ is drawn as a rectangle above $\RR_{\vx}$ which connects the boundary of
$\RR_{\vx}$ and the boundary of $\RR_\kappa$. (Part of the curve $\ggamma_{\vx,
\kappa}$ is an axis of symmetry of this rectangle.)

The remaining squares of $R_1$ will be glued inside $\RR_\kappa$; see
Figure~\ref{f:gluing_incoming_c}. The squares of $R_1$ shared with crossing
annuli are positioned as small axis-aligned squares (or rectangles) containing
the crossings of $\ggamma_{\vx,\kappa}$ with curves $\ddelta_{\kappa'}$; one square is used
for one such crossing. The edge $\ww_{\vx,\kappa}\zz_{\vx,\kappa}$ 
is positioned as a small 
horizontal segment slightly below $\vv_\kappa$ 
such that $\ggamma_{\vx,\kappa}$ meets it in the midpoint. Now the remaining squares
of $R_1$ are drawn as quadrilaterals which interconnect the already placed
squares and $\ww_{\vx,\kappa}\zz_{\vx,\kappa}$ as straightforwardly as possible (without
extra bends or detours; we can assume that the portion of $\ggamma_{\vx,\kappa}$
between the boundary of $\RR_\kappa$ and the segment
$\ww_{\vx,\kappa}\zz_{\vx,\kappa}$ is
covered by $R_1$). This finishes the description of the attachment of
$\cI_{\vx,\kappa}$ to $\cT_0$.

\begin{figure}
  \begin{center}
    \includegraphics[page=26]{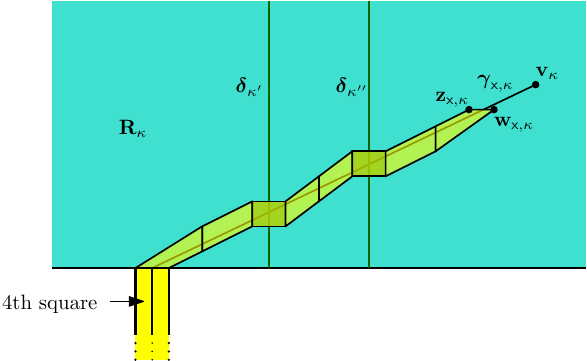}
    \caption{Gluing the incoming house to $\cT_0$, part inside $\RR_\kappa$.}
    \label{f:gluing_incoming_c}
  \end{center}
\end{figure}

We remark that the square $S$ in the 
notation of Figure~\ref{f:incoming_attachment}, left, is not yet glued to anything; it will
be used later on. In particular, the vertices $\xx_{\vx,\kappa},
\yy_{\vx,\kappa}$ 
are not part of $\cT_0$. Similarly, the remainders of crossing annuli are not
glued to anything else yet. As usual, the attachment of these objects will be
explained later on when introducing other gadgets.

For the shelling complexes, we point out that the first shelling complex is a subcomplex of the
attachment complex induced by all vertices of the attachment complex except
$\aa_{\ell,\kappa}, \bb_{\ell,\kappa}$ and the vertices of the crossing annuli outside
the rectangle $R_1$. In other words, it is the subcomplex of the attachment complex
glued to the gadgets $\cT_0$ and $\cC_\kappa$ according to
Figure~\ref{f:incoming_attachment}, right (where $\cC_\kappa$ has not been introduced
yet). The second shelling complex is the subcomplex of the attachment complex
glued to $\cT_0$.

\subsection{The clause gadget}
\label{ss:clause_gadget}

For every clause $\kappa$ we define one \emph{clause gadget} $\cC_\kappa$ which is a copy
of the thick turbine. For description of the attachment of the clause gadget to
other gadgets, we assume that $\kappa$ is positive clause, that is, $\kappa =
(\vx \vee \vy \vee \vz)$ where $\vx, \vy, \vz$ are variables. For a negative clause (with all literals
negative), the attachment is done in a mirror symmetric fashion. We recall that
it may happen that $\kappa$ is obtained from a clause with less than three
literals by duplication of variables; that is, for example, $\vx = \vy$. However in
this case, there are still three curves $\ggamma_{\vx,\kappa}$,
$\ggamma_{\vy,\kappa}$ and $\ggamma_{\vz,\kappa}$ entering the vertex
$\vv_\kappa$. With respect to our earlier slight
abuse of the notation, we distinguish these curves exactly by using $\vx, \vy,
\vz$ even in this case. We also assume that the notation is chosen so that
$\ggamma_{\vx,\kappa}$, $\ggamma_{\vy,\kappa}$ and $\ggamma_{\vz,\kappa}$ enter
$\vv_\kappa$ from left to
right as in Figure~\ref{f:template}.

Now we set up a new notation on $\cC_\kappa$; compare Figures~\ref{f:J_turbine} (old
notation) and~\ref{f:clause_attachment}, left (new notation). Namely, we set up $r \to
\rr_\kappa, s \to \ss_\kappa, t \to \ttt_\kappa, u \to \uu_\kappa, e_1 \to
\ww_{\vx,\kappa},  e_2 \to
\ww_{\vy,\kappa}, e_3 \to \ww_{\vz,\kappa}, b_1 \to \xx_{\vx,\kappa}, b_2 \to
\xx_{\vy,\kappa}, b_3 \to \xx_{\vz,\kappa}, g_1 \to \yy_{\vx,\kappa}, g_2 \to
\yy_{\vy,\kappa}, g_3 \to \yy_{\vz,\kappa}, f_1 \to \zz_{\vx,\kappa}, f_2 \to
\zz_{\vy,\kappa}, f_3 \to \zz_{\vy,\kappa}$. The vertices $\rr_\kappa,
\ss_\kappa, \ttt_\kappa, \uu_\kappa$ are newly introduced
vertices while all the other vertices are in some incoming house. This also
means that the triangles $\xx_{\vx,\kappa}\yy_{\vx,\kappa}\zz_{\vx,\kappa}$ and
$\ww_{\vx,\kappa}\xx_{\vx,\kappa}\zz_{\vx,\kappa}$ are identified in
$\cI_{\vx,\kappa}$ and $\cC_\kappa$;
compare Figures~\ref{f:incoming_attachment} and~\ref{f:clause_attachment},
right. (We also
perform analogous identifications for variables $\vy$ and $\vz$, of course.) We
also emphasize that the triangle
$\xx_{\vx,\kappa}\yy_{\vx,\kappa}\zz_{\vx,\kappa}$ corresponds to
$\tau_1$ in the definition of thick turbine. (Similarly
$\xx_{\vy,\kappa}\yy_{\vy,\kappa}\zz_{\vy,\kappa}$ and
$\xx_{\vz,\kappa}\yy_{\vz,\kappa}\zz_{\vz,\kappa}$ correspond to
$\tau_2$ and $\tau_3$.)
For
points $p$ and $q$ we use the old notation (as well as for few other points
used only in pictures).

\begin{figure}
  \begin{center}
    \includegraphics[page=27,scale=.86]{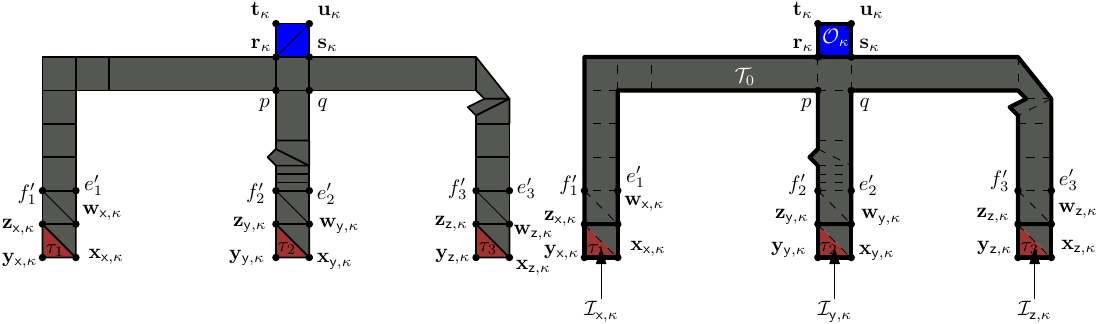}
    \caption{The attachment complex of the clause gadget $\cC_\kappa$.
    Left: The notation of vertices. Right: Intended identifications with other
    gadgets ($\cO_\kappa$ has not been introduced yet).}
    \label{f:clause_attachment}
  \end{center}
\end{figure}

Next we describe the attachment of $\cC_\kappa$ to $\cT_0$. This will be along the
subcomplex of the attachment complex in Figure~\ref{f:gluing_clause}. At
the moment only the edges $\xx_{\vx,\kappa}\yy_{\vx,\kappa}$,
$\xx_{\vy,\kappa}\yy_{\vy,\kappa}$
$\xx_{\vz,\kappa}\yy_{\vz,\kappa}$ of $\cC_\kappa$ are glued to $\cT_0$. (We can assume that all
three edges are in the same height---there was enough flexibility to position
them this way in the constructions of $\cI_{\vx, \kappa}, \cI_{\vy, \kappa}$
and $\cI_{\vz, \kappa}$.)
Next we position the square $pq\ss_\kappa\rr_\kappa$ as a small axis aligned square with
midpoint $\vv_\kappa$. (The point $\vv_\kappa$ is an auxiliary point of the template but
it is not part of the triangulation of $\cT_0$.) Then we simply glue the
remaining parts of the attachment complex (between $pq\ss_\kappa\rr_\kappa$ and
the edges $\xx_{\vx,\kappa}\yy_{\vx,\kappa}$, $\xx_{\vy,\kappa}\yy_{\vy,\kappa}$
$\xx_{\vz,\kappa}\yy_{\vz,\kappa}$) as in Figure~\ref{f:gluing_clause}, possibly again
adjusting the shape slightly (to match the distances between edges and the
right position of the square). Everything occurs inside the rectangle
$\RR_\kappa$.

\begin{figure}
  \begin{center}
    \includegraphics[page=45,scale=.86]{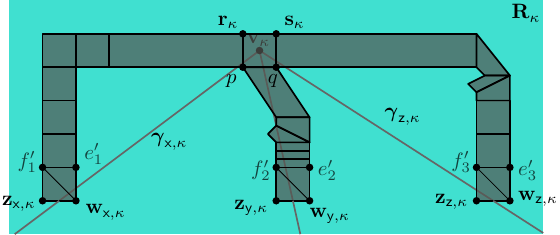}
    \caption{Attaching $\cC_\kappa$ to $\cT_0$.}
    \label{f:gluing_clause}
  \end{center}
\end{figure}

This finishes the attachment of $\cC_\kappa$. The subdivided square
$\rr_\kappa\ss_\kappa\uu_\kappa\ttt_\kappa$ will be attached in the next step to the outgoing
house.

For shelling complexes, recall that there are seven possible shelling
complexes in this case. There are obtained by removing
one, two or three pairs of vertices from the attachment complex among
$\{\xx_{\vx,\kappa},\yy_{\vx,\kappa}\}, \{\xx_{\vy,\kappa},\yy_{\vy,\kappa}\}$
and $\{\xx_{\vz,\kappa},\yy_{\vz,\kappa}\}$.

\subsection{The outgoing house}
\label{ss:outgoing_house}

For every clause we define one outgoing house $\cO_\kappa$. This is the thick
1-house with the lower wall triangulated according to the outgoing case. The
number of squares $k$ will be specified during the construction. 

We perform the following renaming of the vertices: $e \to \rr_\kappa, g' \to
\ss_\kappa,
d' \to \ttt_\kappa, e' \to \uu_\kappa$. 
For the vertices $b, c, f, g, h, o, p, q$, the
rectangle $R$ and the subcomplex $J$ we keep the old notation and we use them
only locally here; see Figure~\ref{f:outgoing_attachment}, left, for the attachment
complex in the new notation (and Figures~\ref{f:F_triangulate}
and~\ref{f:outgoing_wall} for the old notation).
The first four vertices $\rr_\kappa, \ss_\kappa, \ttt_\kappa$ and $\uu_\kappa$ are
already in the clause gadget $\cC_\kappa$. In particular, the subdivided square
$\rr_\kappa\ss_\kappa\uu_\kappa\ttt_\kappa$ of $\cC_\kappa$ is identified with the two triangles of
$\cO_\kappa$; compare Figures~\ref{f:clause_attachment}
and~\ref{f:outgoing_attachment}, right. 
We also emphasize that $\tau$ in the description of thick 1-house translates
as the triangle $\rr_\kappa\uu_\kappa\ttt_\kappa$ in the new notation. 

\begin{figure}
  \begin{center}
    \includegraphics[page=46]{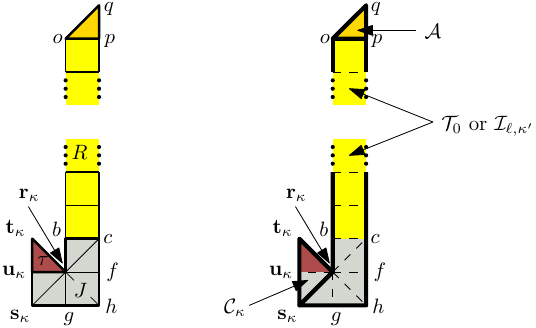}
    \caption{The attachment complex of the outgoing house $\cO_\kappa$ (the
    picture is mirrored when compared with Figure~\ref{f:outgoing_wall}. Left:
    The notation of vertices. Right: Intended identifications with other
    gadgets. In the picture we did not attempt to distinguish where this
    attachment complex is glued to $\cT_0$ and where it is glued to the
    crossing annuli of some $\cI_{\ell,\kappa'}$ (possibly to several such annuli
    for different incoming houses). The gadget $\cA$ has not been introduced
    yet.}
    \label{f:outgoing_attachment}
  \end{center}
\end{figure}

Now we start gluing the attachment complex of $\cO_\kappa$ to other gadgets (mostly
$\cT_0$). We first
glue the remaining triangles of $J$ to $\cT_0$ directly above the edge
$\rr_\kappa\ss_\kappa$ as depicted in Figure~\ref{f:J_incoming}. (This is the familiar
picture which also appeared in the cases of the splitter house and the incoming house.)

\begin{figure}
  \begin{center}
    \includegraphics[page=30]{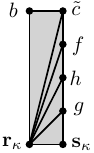}
    \caption{The triangles of $J$ glued to $\cT_0$ except
    $\rr_\kappa\ss_\kappa\uu_\kappa$.}
    \label{f:J_incoming}
  \end{center}
\end{figure}

Next we describe how to glue $R$. Here our description is slightly informal
referring to Figure~\ref{f:R_incoming} in order to simplify at least a small
bit an already complicated notation. We start with placing a certain polygon
$\PP$ into $\cT_0$. We draw an auxiliary small enough circle passing through
$\vv_\pm$ which is tangent to the boundary of $\cT_0$. The topmost point of
this circle is denoted $\ww_+$, the bottommost is denoted $\ww_-$. The vertices
of $\PP$ will be $\vv_\pm$, $\ww_+$, $\ww_-$ and some number of vertices on
the arcs between $\ww_+$ and $\vv_\pm$ and between $\vv_\pm$ and $\ww_-$. 
Namely we add $k_+ -1$ new vertices on the arc between $\ww_+$ and $\vv_\pm$
where $k_+$ is the number of positive clauses and $k_- -1$ new vertices on the arc between 
$\vv_\pm$ and $\ww_-$ where $k_-$ is the number of negative clauses. (In
Figure~\ref{f:R_incoming}, $\PP$ is larger than really desired for easier
visualisation.)
Note that
we can assume that both $k_+$ and $k_-$ are at least one as a formula without
positive or without negative clauses is trivially satisfiable and we can leave
out such a formula from our considerations.\footnote{Purely, formally, we can
set $\cK_\phi$ to be a single tetrahedron for such a formula, for example.}

Then, for a positive clause $\kappa$, we glue $R$, starting from segment $bc$
along the curve $\ddelta_\kappa$: The first square is above $J$ until we reach
the first crossing of $\ddelta_\kappa$ with
some $\ggamma_{\vx,\kappa'}$ (if it exists); more precisely, until we reach
$\cI_{\vx,\kappa'}$ already glued to $\cT_0$. If such a crossing occurs, we glue next
squares to the corresponding crossing annulus of $\cI_{\vx,\kappa'}$ (more precisely
to the bend formed by those squares of the annulus that are not yet glued to
$\cT_0$). Here we crucially use that in the outgoing case we can prescribe the
diagonals of $R$ as the need arises as the diagonals of the crossing annulus
are already prescribed. In Figure~\ref{f:R_incoming}, passing through the
crossing annulus is depicted by interrupting $R$. Then we again
continue up until we reach another crossing (if it exists) and we glue along
another annulus. 

As soon as we pass through all crossings we continue up close to the top
boundary of $\cT_0$ until we reach the first bend of $\ddelta_\kappa$. 
Then we turn right towards the right boundary, and then we turn right once more and glue the other end of 
$R$ to one of the edges of $\PP$ between $\vv_{\pm}$ and $\ww_+$. (After the
second bend, $\ddelta_\kappa$ may leave $R$ but it is not a problem,
$\ddelta_\kappa$ is only an auxiliary curve, not a part of construction.)

When gluing different outgoing houses simultaneously, we require that their $R$
do not cross and they meet only in at most one vertex of $\PP$ as in
Figure~\ref{f:R_incoming}. (In particular, we choose different edges of $\PP$
for different outgoing houses.)

\begin{figure}
  \begin{center}
    \includegraphics[page=29]{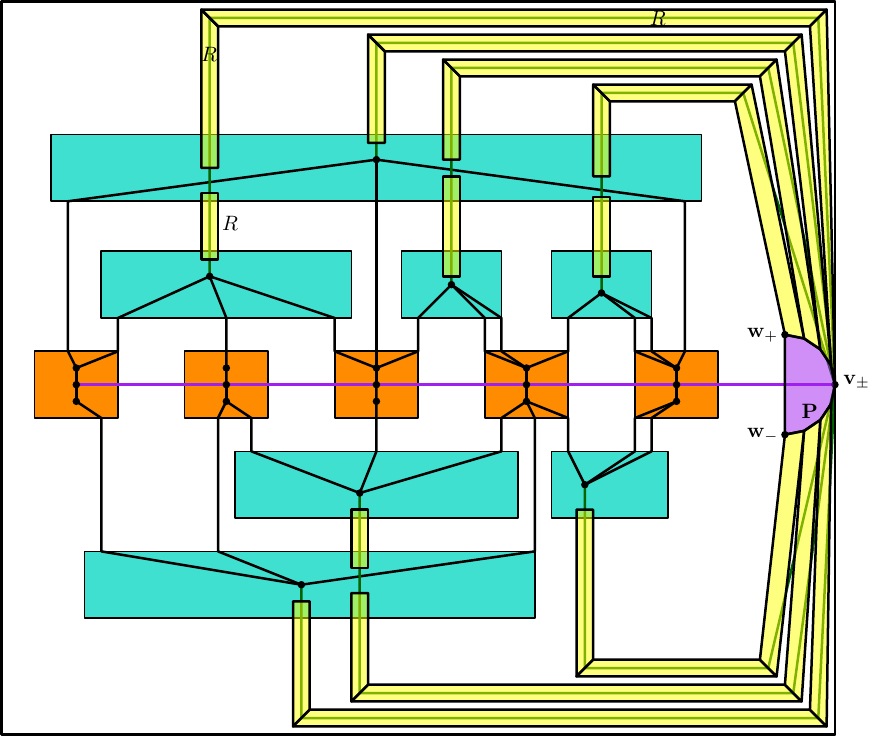}
    \caption{The outgoing houses glued to the template}
    \label{f:R_incoming}
  \end{center}
\end{figure}

If $\kappa$ is a negative clause, the construction is analogous, mirror symmetric.
This finishes gluing the outgoing house(s) to the previous gadgets. The number
of squares in $R$ is chosen so that the gluing described above is possible.
The triangle $opq$ will be glued to the conjunction gadget $\cA$ in the next step.

For the shelling complex, we point out that this is the induced subcomplex of
the attachment complex on all vertices except $\ttt_\kappa$ and $\uu_\kappa$. In other
words, it is a part of the attachment complex glued to $\cT_0$, some
$\cI_{\ell, \kappa'}$ and $\cA$ as in Figure~\ref{f:outgoing_attachment}, right.

\subsection{The conjunction gadget}
\label{ss:conjunction_gadget}

Now we define the conjunction gadget $\cA$ (we use $\cA$ for `and' as $\cC$ is
already taken). This will be the conjunction cone chosen so that we identify
the polygon $P$ of the conjunction cone with $\PP$ in the template (and thus we
triangulate $\PP$ appropriately in the template). The vertex $w_+$ in the
definition of the conjunction cone is identified with $\ww_+$ and $w_-$ is
identified with $\ww_-$. The apex of $\cA$ is denoted
$\dd$. Now each outgoing house $\cO_\kappa$ meets $\PP$ in some edge $\varepsilon$.
Recall that the attachment complex of $\cO_\kappa$ still contains a triangle
not attached to any gadget yet while it contains $\varepsilon$. (These is the triangle $opq$ of
Figure~\ref{f:outgoing_attachment} and $\varepsilon$ is the edge $op$.) 
We identify the remaining vertex (vertex $q$ of
Figure~\ref{f:outgoing_attachment})
with $\dd$. (This, in particular means that each triangle $opq$ is identified
with a triangle of $\cA$.) This finishes the construction of $\cA$.

\subsection{The blocker house}
\label{ss:blocker_house}

Finally, we describe the \emph{blocker house}. Similarly as in the
collapsibility case we temporarily assume for this construction 
that the variables are $\vx_1, \dots, \vx_n$ ordered from left to right on the
template. We will have $n+1$ blocker houses $\cB_i$ for $i \in \{0, \dots,
n\}$.

We will perform the following identifications of vertices of $\cB_i$ with
previous gadgets (see Figures~\ref{f:blocker_i}, \ref{f:blocker_0}
and~\ref{f:blocker_n}).

\begin{figure}
  \begin{center}
    \includegraphics[page=40]{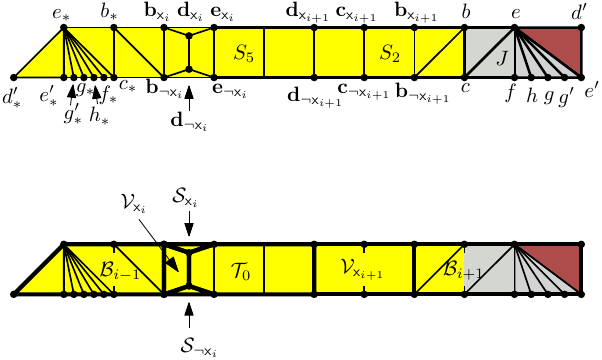}
    \caption{The attachment complex of $\cB_i$ for $1 \leq i \leq n-1$. Top:
   Notation of vertices. Bottom: Neighboring gadgets.}
    \label{f:blocker_i}
  \end{center}
\end{figure}
\begin{figure}
  \begin{center}
    \includegraphics[page=41]{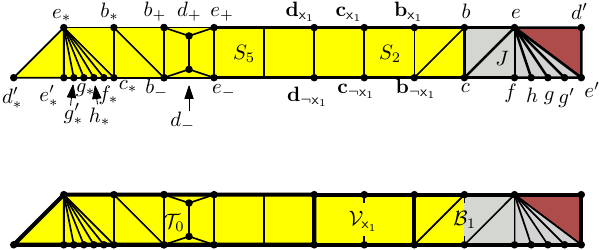}
    \caption{The attachment complex of $\cB_0$. Top:
   Notation of vertices. Bottom: Neighboring gadgets.}
    \label{f:blocker_0}
  \end{center}
\end{figure}
\begin{figure}
  \begin{center}
    \includegraphics[page=42]{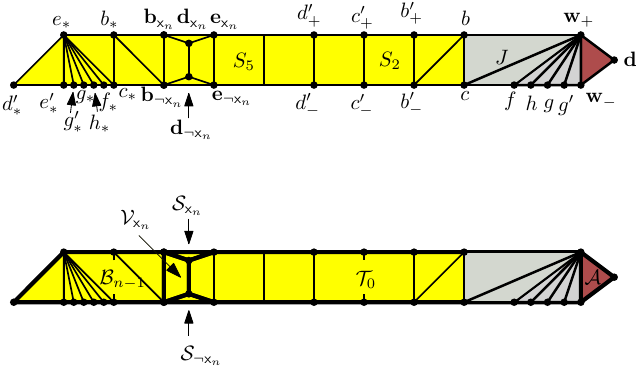}
    \caption{The attachment complex of $\cB_n$. Top:
   Notation of vertices. Bottom: Neighboring gadgets.}
    \label{f:blocker_n}
  \end{center}
\end{figure}

For $i \in \{0, \dots, n-1\}$ we set up:
$b'_+ \to \bb_{\vx_{i+1}}, b'_- \to \bb_{\neg \vx_{i+1}},
c'_+ \to \cc_{\vx_{i+1}}, c'_- \to \cc_{\neg \vx_{i+1}},
d'_+ \to \dd_{\vx_{i+1}}, d'_- \to \dd_{\neg \vx_{i+1}}.
$

For $i \in \{1, \dots, n\}$ we set up:
$
b_+ \to \bb_{\vx_i}, b_- \to \bb_{\neg \vx_i},
d_+ \to \dd_{\vx_i}, d_- \to \dd_{\neg \vx_i},
e_+ \to \ee_{\vx_i}, e_- \to \ee_{\neg \vx_i}.
$

For $i= n$ we set up $e \to \ww_+$, $e' \to \ww_-$, $d' \to \dd$.

In addition for every $i \in \{0, \dots, n-1\}$ we identify the vertex $b$ of
$\cB_i$ with the vertex $b_*$ of $\cB_{i+1}$. Similarly, we identify $c$ with
$c_*$, $d'$ with $d'_*$, $e$ with $e_*$, $e'$ with $e'_*$, $f$ with $f_*$, $g$
with $g_*$, $g'$ with $g'_*$ and $h$ with $h_*$ where the first identified
vertex always come from $\cB_i$ and the second one from $\cB_{i+1}$. 

We also recall that lower wall in the blocker case contains a distinguished
triangulated polytope $P_6$. In order to emphasize $i$, we rename it to $\PP_i$
inside $\cB_i$.

Now we explain/clarify how do we glue the remaining vertices of the attachment
complex of $\cB_i$ and what is the result of all these identifications on
$2$-faces of the attachment complex and the other gadgets. We provide this
explanation roughly in `direction' form $d'$ towards $d'_*$ (see Figures~\ref{f:blocker_i}, \ref{f:blocker_0}
and~\ref{f:blocker_n}) while discussing different cases depending on $i$.

If $i = n$, then the triangle $\ww_+\ww_-\dd$ is identified with the
corresponding triangle of the conjunction gadget. Then, still for $i=n$, the
subcomplex consisting of $J$, the subdivided squares $S_1, \dots,
S_5$ and quadrilateral $\ee_{\vx_n}\dd_{\vx_n}\dd_{\neg \vx_n}\ee_{\vx_n}$ is glued to
$\cT_0$ in between the edges $\ww_+\ww_-$ and $\dd_{\vx_n}\dd_{\neg \vx_n}$. As
the diagonals in the attachment complex are predetermined, they are chosen in
the same way in $\cT_0$. 

If $i \in \{0, \dots, n-1\}$, then the subcomplex consisting of the triangle
$d'ee'$, the subcomplex $J$ and $S_1$ is identified to a subcomplex of
attachment complex of $\cB_{i+1}$. Note that these two subcomplexes exactly
match each other according to the earlier choices of identifications of
vertices. Then, still for $i \in \{0, \dots, n-1\}$, the squares $S_2$ and
$S_3$ are identified with two subdivided squares (or rather rectangles) of the variable gadget
$\cV_{\vx_{i+1}}$; see also Figure~\ref{f:variable_gadget}, right, considering
$\vx = \vx_{i+1}$. Here we use that in the variable gadget, we are flexible to
choose the diagonals of these squares, thus the identification is possible.

Now, if $i = 0$, we glue the remainder of the complex, that is, the subdivided squares 
$S_4, \dots, S_8$ and the triangle $d'_*e'_*e_*$, to $\cT_0$ on the left of the
edge $\dd_{\vx_1}\dd_{\neg \vx_1}$ so that it does not interact with anything else.
(This finishes the explanation if $i=0$.)

If $i \in \{1, \dots, n-1\}$, then we glue the subdivided squares $S_4$,
$S_5$ and $\ee_{\vx_i}\dd_{\vx_i}\dd_{\neg \vx_i}\ee_{\vx_i}$ to $\cT_0$ so that it
fits in between the edges of $\dd_{\vx_{i+1}}\dd_{\neg \vx_{i+1}}$ and
$\dd_{\vx_i}\dd_{\neg \vx_i}$ on $\cT_0$. 

For the remainder of the explanation, we assume $i \in \{1, \dots, n\}$. The
triangle $\dd_{\vx_i}\ee_{\vx_i}\bb_{\vx_i}$ is identified with the corresponding
triangle of the splitter house $\cS_{\vx_i}$ and similarly $\dd_{\neg
\vx_i}\ee_{\neg \vx_i}\bb_{\neg \vx_i}$ with the corresponding triangle of
$\cS_{\neg \vx_i}$; see also Figure~\ref{f:splitter_attachment}, right,
considering $\ell = \vx_i$ and $\ell = \neg \vx_i$. The subdivided square
$\dd_{\vx_i}\bb_{\vx_i}\bb_{\neg \vx_i}\dd_{\neg \vx_i}$ is identified with the
corresponding square of the variable gadget $\cV_{\vx_i}$; see also
Figure~\ref{f:variable_gadget}, right, this time considering $\vx
= \vx_i$. (Again, we use that we can choose the diagonals of these squares 
on a variable gadget.) The rest of the complex, consisting of subdivided
squares $S_7$ and $S_8$ and the triangle $d'_*e'_*e_*$ is glued to $\cB_{i-1}$
according to the earlier identification of vertices between $\cB_i$ and
$\cB_{i-1}$.

For shelling complexes: Type $0$ shelling complex will be considered only for
$\cB_0$ and it is a subcomplex of the attachment complex attached to $\cT_0$ or
$\cV_{\vx_1}$. The type $i$ shelling complex will be considered only for 
$\cB_i$ with $i \in \{1, \dots, n-1\}$ and it is the union of $\PP_i$ and the
subcomplex of the attachment complex attached to $\cT_0$, $\cV_{\vx_i}$,
$\cV_{\vx_{i+1}}$,
$\cB_{i-1}$, $\cS_{\vx_i}$ and $\cS_{\neg \vx_i}$. Finally the 
type $n$ shelling complex will be considered only for
$\cB_n$ and it is the union of $\PP_n$ and the subcomplex of the attachment complex attached to
$\cT_0$, $\cV_{\vx_n}$, $\cB_{n-1}$, $\cS_{\vx_n}$ and $\cS_{\neg \vx_n}$.

\subsection{Triangulating the template}
Gluing other gadgets to $\cT_0$ enforces some of the triangles, edges, or
vertices to be present in a triangulation of $\cT_0$. 
We extend this
arbitrarily to a full triangulation of $\cT_0$; we only require that it remains
of polynomial size in the number of variables. Then we deduce a triangulation
of $\cT$ as explained in Subsection~\ref{ss:template_gadget}. 

\subsection{$\cK_\phi$ is a simplicial complex}
We have to be a bit careful when gluing the gadgets to verify that we get
indeed a simplicial complex. In order to visualize a possible problem consider
the two simplicial complexes in Figure~\ref{f:chords}. If we glue them together
along their boundaries identifying the vertices with the same name, we do not
get a simplicial complex. Indeed the edge $cf$ appears in both of them, thus,
after gluing it is not uniquely determined by its vertices.

\begin{figure}
  \begin{center}
    \includegraphics[page=57]{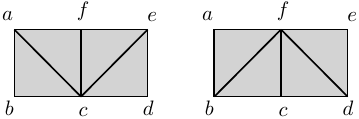}
  \end{center}
  \caption{Two simplicial complexes that do not yield a simplicial complex
  after gluing.}
  \label{f:chords}
\end{figure}

Now we start verifying
that $\cK_\phi$ is a simplicial complex. For purpose of
this verification, we extend the definition of the attachment complex also to
 the template gadget $\cT$, the conjunction gadget $\cA$ and the variable
 gadgets $\cV_\vx$ so that it is the part of the boundary of these gadgets
 which is glued to other gadgets. In order to glue all gadgets to a simplicial
 complex, we need to check that there are no faces $\sigma$ and $\sigma'$
 in two different gadgets $\cG$, $\cG'$ respectively (before the
 identification) so that $\partial \sigma$ and $\partial \sigma'$ are
 identified in $\cK_\phi$ while one of them does not belong to the attachment
 complex of its gadget. (Without loss of generality $\sigma$ does not belong to
 the attachment complex of $\cG$.) Let us remark that, if in the setting above,
 both $\sigma$ and $\sigma'$ belong to the attachment complexes, then they are
 identified as well thus there is no problem.
 
We will inspect all possibly \emph{problematic} faces $\sigma$ in a gadget $\cG$ such
that $\partial \sigma$ belongs to the attachment complex while $\sigma$ does
not belong to it. In each such case we will verify that there is no $\sigma'$
in another gadget $\cG'$ with $\partial \sigma'$ identified with $\partial
\sigma$. Given a simplicial complex $K$ and a subcomplex $L$, we say that $L$
is an induced subcomplex of $K$ if for every face $\vartheta$ of $K$ whenever
all vertices of $\vartheta$ belong to $L$, then $\vartheta$ belongs to $L$ as
well. If the attachment complex of a gadget $\cG$ is an induced subcomplex,
then there is no candidate for a problematic $\sigma$ in $\cG$.

If $\cG$ is the template gadget $\cT$, the attachment complex is a subcomplex
of $\cT_0$. Depending on the chosen triangulation of $\cT_0$, the attachment
complex may be an induced subcomplex of $\cT_0$ and therefore of $\cT$, or not.
In any case, if $\sigma$ is a face of $\cT$ (and therefore of $\cT_0$) such
that the attachment complex contains $\partial \sigma$ while not $\sigma$, then
$\sigma$ belongs to the boundary of $\cK_\phi$ due to the way how was the
attachment complex glued to other gadgets (by inspection). Therefore, there is
no $\sigma'$ in another gadget with $\partial \sigma'$ identified with
$\partial \sigma$.

If $\cG$ is the variable gadget $\cV_\vx$, the attachment complex consists of
almost entire boundary of $\cV_\vx$ except the (interiors of the) faces
$\aa_\vx\cc_\vx$, $\aa_\vx\bb_\vx\cc_\vx$, $\aa_\vx\cc_\vx\dd_\vx$, $\aa_{\neg
\vx}\cc_{\neg \vx}$, $\aa_{\neg \vx}\bb_{\neg \vx}\cc_{\neg \vx}$ and
$\aa_{\neg \vx}\cc_{\neg \vx}\dd_{\neg \vx}$; see
Figure~\ref{f:variable_gadget}, right. By a direct inspection using
that $\partial \cV_\vx$ is triangulated as a cone over the boundary, the only
problematic $\sigma$ are $\aa_\vx\cc_\vx$ and $\aa_{\neg
\vx}\cc_{\neg \vx}$. The only other gadget that contains $\aa_\vx$ or
$\aa_{\neg \vx}$ is $\cS_\vx$ or $\cS_{\neg \vx}$. However non-of them contains
$\cc_\vx$ or $\cc_{\neg \vx}$ thus there is no $\sigma'$ in another
gadget with $\partial \sigma'$ identified with $\partial \sigma$.

If $\cG$ is one of the houses $\cS_\ell$, $\cI_{\ell, \kappa}$, $\cO_{\kappa}$
or $\cB_i$, then the attachment complex is actually an induced subcomplex of
$\cG$. This follows from inspection of each case individually; however, we
emphasize the joint properties referring to the notation of
Subsection~\ref{ss:thick_1house}. The attachment complex always contains
$\tau$, $J$, some subcomplex of the front side of the lower wall and possibly
crossing annuli in the incoming case. The union of $J$ and $\tau$ is an induced
subcomplex of $\partial \cG$ by a direct inspection. This can be further used
to check that the whole attachment complex is an induced subcomplex using that
cubes (or prisms in the lower wall) in the polytopal decomposition of the thick 1-house meet the attachment
complex usually only inside a single face. The exceptions from this rule are
the cube $F$, the cube sharing the edge $ee'$ with $F$ (those two are relevant
for $\tau$ and $J$), the exceptional polytope $P_6$ in the blocker case, and
the cubes at the bends of crossing annuli in the incoming case.  The latter two
exceptions could possibly yield a problematic $\sigma$; however this does not
occur in case of $P_6$ due to its triangulation as a cone with apex $p$ and in
case of the cubes at the bends due to their triangulations using rule (R3) when
triangulating the thick 1-house.

If $\cG$ is the clause gadget $\cC_\kappa$ (i. e. a turbine), we find all problematic $\sigma$ in the
following way. Such a $\sigma$ is inside some prism $Q$ considering the turbine as
a polytopal complex. Each $Q$ that meets the attachment complex inside a single
face does not admit such $\sigma$ (by inspection). It remains to inspect $F_i$,
the cubes, left of $F_i$ and the cubes where the attachment complex bends
(similarly as in the case of crossing annuli); see
Figures~\ref{f:thick_blade},~\ref{f:thick_triangle}
and~\ref{f:F_neighborhood_blade}. The cubes of bends can be ruled out similarly
as in the case of crossing annuli, using the rule (R3) from the triangulation
of the turbine. The inspection of the other cases reveals that the only
problematic $\sigma$ are the edges $f'_ib_i$ and $f'_ig_i$ in the notation of
Figure~\ref{f:F_neighborhood_blade}. In $\cC_\kappa$, these edges translate as
$f'_i\xx_{\ell, \kappa}$ and $f'_i\yy_{\ell, \kappa}$ for some literal $\ell$ in
$\kappa$; compare
Figures~\ref{f:J_turbine} and~\ref{f:clause_attachment}. The only other gadget
containing $\xx_{\ell, \kappa}$ or $\yy_{\ell, \kappa}$ is $\cI_{\ell, \kappa}$
which does not contain $f_1'$. Thus there is no $\sigma'$ in another
gadget with $\partial \sigma'$ identified with $\partial \sigma$.

Finally, if $\cG$ is the conjunction gadget $\cA$, then the direct inspection
reveals that the only problematic $\sigma$ are the triangles
$dw_-\varepsilon_i$ for $i \in \{1, \dots, k-1\}$ in the notation of
Figure~\ref{f:conjunction_cone}. In this case, $\partial \sigma$ does not
belong to any other gadget because $w_-\varepsilon_i = \ww_-\varepsilon_i$ belongs only to $\cT$
while $d = \dd$ does not belong to it.

\section{Satisfiable implies shellable}
\label{s:sat=>shell}

In this section we assume that our formula $\phi$ is satisfiable and we will
show that $\cK_\phi$ is shellable. We fix a satisfying assignment and from this
assignment we derive a shelling of $\cK_\phi$
in several steps which roughly correspond to the order of gluing gadgets when
building $\cK_\phi$. All shellings we describe in this section are
shellings down. The
\emph{intermediate complex} in any stage of our description 
is the subcomplex of $\cK_\phi$ to which $\cK_\phi$ shells via the previously
described shellings. If $\cG$ is a gadget fully contained in the intermediate
complex, then the \emph{remainder} of the intermediate complex (with respect to
$\cG$) is the subcomplex generated by those faces of the intermediate complex
that do not belong to $\cG$; that is, it is formed by those faces which are
contained in a face that is not contained in $\cG$. Note that $\cG$ intersects
the remainder exactly in those faces that belong to $\cG$ but they are contained
in a face that does not belong to $\cG$.

\paragraph{Step 1, first shelling of variable gadgets:}
Let $\vx$ be a variable. Recall that the variable gadget $\cV_{\vx}$ is triangulated as a cone over its boundary; see Figure~\ref{f:variable_gadget}, left.

First assume that $\vx$ is assigned TRUE. In this case we first remove the unique tetrahedra
of $\cV_{\vx}$ containing the triangle $\aa_{\vx}\bb_{\vx}\cc_{\vx}$, then the
one containing the triangle $\aa_{\vx}\cc_{\vx}\dd_{\vx}$ and then the one
containing $\aa_{\vx}\bb_{\vx}\dd_{\vx}$. It is routine to check that this is a shelling. 

Analogously, if $\vx$ is assigned FALSE, we remove the three unique tetrahedra of
$\cV_{\vx}$ containing the triangles $\aa_{\neg \vx}\bb_{\neg \vx}\cc_{\neg \vx}$,
$\aa_{\neg \vx}\cc_{\neg \vx}\dd_{\neg \vx}$ and $\aa_{\neg \vx}\bb_{\neg
\vx}\dd_{\neg \vx}$ in this order.

We perform the shellings above for every variable $\vx$ in an arbitrary order.

\paragraph{Step 2, shelling of positive splitter houses:}
Let $\ell$ be a literal assigned TRUE (that is, either $\ell = \vx$ where $\vx$ is
a variable assigned TRUE, or $\ell = \neg \vx$ where $\vx$ is a variable assigned
FALSE). Now we shell the splitter house $\cS_\ell$ (for each such $\ell$ in
arbitrary order). Note that in the intermediate complex, the splitter house $\cS_\ell$ meets
the remainder of the complex exactly in the first shelling complex in the splitter
case. Indeed, the simplices in the attachment complex of $\cS_\ell$ which are not
in the first shelling complex of $\cS_\ell$ are exactly those that contain vertex
$\aa_\ell$ 
(recall that the attachment complex of $\cS_\ell$
and its attachment to other gadgets is depicted in
Figure~\ref{f:splitter_attachment} and that the first shelling complex in the
splitter case is the subcomplex of the attachment complex glued to $\cB_i$,
$\cT_0$ or $\cI_{\ell,\kappa_j}$). 
 Vertex $\aa_\ell$ appears
only in $\cV_{\vx}$ and $\cS_\ell$; however, shellings in Step 1 caused that any
simplex of the intermediate complex which contains $\aa_\ell$ belongs to
$\cS_\ell$. A fortiori any simplex of the intermediate complex 
containing a face containing $\aa_\ell$
belongs to $\cS_\ell$. Therefore no face of the intermediate complex containing $\aa_\ell$
is in the remainder. On the other hand, all remaining simplices of the attachment
complex belong to the remainder as they also belong to other gadgets which
have not been shelled yet; see Figure~\ref{f:splitter_attachment}, right.

Once we have checked that $\cS_\ell$ meets
the remainder in the (first) shelling complex, we perform the shellings on $\cS_\ell$
according to Lemma~\ref{l:shelling_house} with the intermediate complex as $K$,
the remainder as $L$ and $\cS_\ell$ as $H$. (This essentially removes $\cS_\ell$
from the picture. More precisely, the intermediate complex after this step will
contain only the shelling complex of $\cS_\ell$ which anyway belongs to other
gadgets as well.)

\paragraph{Step 3, shelling of positive incoming houses, clause gadgets and
outgoing houses:}
Similarly as in the collapsibility case, we would like to shell incoming houses
but it cannot be done immediately for all incoming houses as some of them may
be blocked by outgoing houses. Thus we have to interlace these shellings:
shelling some incoming house will release the clause gadget and the
corresponding outgoing house which may release also another incoming house.

We order the clauses in the following way: We start with positive clauses (with
all three literals positive) and we order them according to the $y$-coordinate
of $\vv_\kappa$ starting with the lowest $y$-coordinate. Then we continue with the
negative clauses. We again order them according to the $y$-coordinate
of $\vv_\kappa$, this time starting with the highest $y$-coordinate (i. e. closest
to the $x$-axis). 

We consider the clauses $\kappa$ one by one in the aforementioned order and for each
of them we perform the following shellings. We inductively assume that all
tetrahedra of each $\cO_{\kappa'}$ with $\kappa'$ preceding $\kappa$ have been
shelled before performing the shellings for $\kappa$. We also remark that
now we shell only those incoming $\cI_{\ell, \kappa}$ for which $\ell$ is a literal of $\kappa$ assigned TRUE. For fixed $\kappa$ we shell the gadgets in the order: some incoming houses, the clause gadget and then the outcoming house.

\bigskip

\subparagraph{Incoming houses:} For every literal $\ell$ in $\kappa$ such that
$\ell$ is assigned TRUE we aim to perform shellings on the incoming house
$\cI_{\ell,\kappa}$. (As the assignment is satisfying, there is at least one such
$\ell$ for our $\kappa$. If there are more such $\ell$, we consider them in
arbitrary order.) First, we check that $\cI_{\ell,\kappa}$ meets the remainder of the
intermediate complex exactly in the shelling complex (in the incoming case).

Recall that the shelling complex of $\cI_{\ell,\kappa}$ is the part of the
attachment complex which is glued to $\cT_0$ or $\cC_\kappa$; see
Figure~\ref{f:incoming_attachment}, right. This part of course
belongs to the remainder as neither $\cT_0$ nor $\cC_\kappa$ have been shelled yet.
On the other hand no other face of the shelling complex belongs to the remainder
as $\cS_\ell$ has been shelled in Step~2 and $\cO_{\kappa'}$ with $\kappa'$
preceding $\kappa$
in previous stages of Step~3 due to our inductive assumption. (Note that our
order satisfies that if $\cO_{\kappa'}$ meets $\cI_{\ell,\kappa}$, then
$\kappa'$ precedes $\kappa$.) This finishes the check.

Thus we may perform shelling according to Lemma~\ref{l:shelling_house} with the
intermediate complex as $K$, the remainder as $L$ and $\cI_{\ell,\kappa}$ as
$H$, which we do. This essentially removes $\cI_{\ell,\kappa}$ from the intermediate complex
(up to the parts in other gadgets not shelled yet).

\subparagraph{Clause gadget:}
Now we want to shell the clause gadget $\cC_\kappa$. As usual, we want to
describe how $\cC_\kappa$ intersects the remainder of the intermediate complex
(with respect to $\cC_\kappa$). Assume that $\kappa = (\ell_1 \vee \ell_2 \vee
\ell_3)$ where $\ell_1, \ell_2$ and $\ell_3$ are literals, possibly with
repetitions if $\kappa$
arose from a clause with less than three literals by repetitions. 
This intersection is a subcomplex of the attachment complex of $\cC_\kappa$ which
surely misses vertices $\xx_{\ell_i,\kappa}$ and $\yy_{\ell_i,\kappa}$ such that $\ell_i$
is assigned TRUE because the incoming house $\cI_{\ell_i,\kappa}$ has been already
shelled; see Figure~\ref{f:clause_attachment}, right (with $\ell_1 = \vx$,
$\ell_2 = \vy$ and $\ell_3 = \vz$). On the other hand this intersection surely contains the induced subcomplex of the attachment
complex on all vertices except the aforementioned vertices
$\xx_{\ell_i,\kappa}$ and $\yy_{\ell_i,\kappa}$ for such $\ell_i$ that are
assigned TRUE because none of $\cT_0$, $\cO_\kappa$ nor $\cI_{\ell_i,\kappa}$ with $\ell_i$ assigned FALSE have been shelled yet; see again Figure~\ref{f:clause_attachment}, right.
Because our initial assignment is satisfying, this means that $\cC_\kappa$
intersects the remainder in one of seven possible shelling complexes (for
turbine)---they are reminded in last paragraph of
Subsection~\ref{ss:clause_gadget}.

This means that we can apply Lemma~\ref{l:shelling_turbine} with $K$ as the
intermediate complex, $L$ as the remainder and $T$ as $\cC_\kappa$. Thus we perform
the shelling from this lemma which essentially removes $\cC_\kappa$ from the
intermediate complex.

\subparagraph{Outgoing house:} Our final shelling for $\kappa$ is the shelling of
the outgoing house $\cO_\kappa$. This will also verify our inductive assumption.

Again, we want to describe how $\cO_\kappa$ intersects the
remainder (with respect to $\cO_\kappa$). As usual, we need that this
intersection is the shelling complex of $\cO_\kappa$. This follows immediately from the description of the shelling complex in last paragraph of
Subsection~\ref{ss:outgoing_house} and the fact that $\cC_\kappa$ has been already
shelled while $\cT_0$, $\cA$ and those $\cI_{\ell, \kappa'}$ which meet
$\cO_\kappa$ have
not been shelled yet. (Such $\cI_{\ell, \kappa'}$ satisfies that $\kappa'$ comes later
than $\kappa$ in our order.)

Therefore, we may use Lemma~\ref{l:shelling_house}  with $K$ as the
intermediate complex, $L$ as the remainder and $H$ as $\cO_\kappa$ and we perform
the shelling from this lemma. This essentially removes $\cO_\kappa$ from the
intermediate complex.

\paragraph{Step 4, shelling of the conjunction gadget:}

Now because all outgoing houses $\cO_\kappa$ have been shelled, we may shell the
conjunction gadget $\cA$. Using the notation in the definition of conjunction
cone (see Subsection~\ref{ss:conjunction_cone}) we shell the tetrahedra in the order
$\Delta_{k-1}, \Delta_{k-2}, \dots, \Delta_1$.

\paragraph{Step 5, partial shelling of the blocker houses:}
Now we shell the blocker houses in order $\cB_n, \cB_{n-1}, \dots, \cB_0$. We
recall that the lower wall each $\cB_i$ contains a distinguished polytope
$\PP_i$ (denoted $P_6$ when triangulation the lower wall in the blocker case).
We point out in advance that for $i \in \{1, \dots, n\}$, $\cB_i$ will be
shelled so that $\PP_i$ remains after shelling (while all tetrahedra outside
$\PP_i$ will be removed).

Now we describe shelling of $\cB_n$. As usual, we want to use
Lemma~\ref{l:shelling_house} but we need to be more careful now because we
intend to use it with $3$-dimensional type $n$ shelling complex in blocker
case. Namely, we intend to use Lemma~\ref{l:shelling_house} so that $K$ is the
intermediate complex, $L$ is the union of $\PP_n$ with the remainder (with
respect to $\cB_n$), and $H$ is $\cB_n$. We need to verify that $H \cap L$ is
the type $n$ shelling complex.

Note that $H \cap L$ of the union of $\PP_n$ and the intersection of $H =
\cB_n$ with the remainder. The following gadgets have been already
shelled: $\cA$ (in Step 4) and either $\cS_{\vx_n}$ if $\vx_n$ is assigned TRUE,
or $\cS_{\neg \vx_n}$ if $\vx_n$ is assigned false (in Step 2). We also remark
that $\cV_{\vx_n}$ has been partially shelled in Step 1; however, the
subdivided rectangle $\bb_{\vx_n}\dd_{\vx_n} \dd_{\neg \vx_n}\bb_{\vx_n}$ is
still part of the remainder as no tetrahedron of $\cV_{\vx_n}$ containing a
triangle subdividing $\bb_{\vx_n}\dd_{\vx_n} \dd_{\neg \vx_n}\bb_{\vx_n}$ has been shelled. Therefore, the intersection of $\cB_n$ and the remainder coincides with
subcomplex of the attachment complex which is glued to $\cT_0$, $\cB_{n-1}$,
$\cV_{\vx_n}$ and one of $\cS_{\vx_n}$ or $\cS_{\neg \vx_n}$ (see
Figure~\ref{f:blocker_n}). After taking the union with
$\PP_n$ we exactly get type $n$ shelling complex in the blocker case according
to the description in last paragraph of Subsection~\ref{ss:blocker_house}.
(Note that the intersection of the attachment complex with both $\cS_{\vx_n}$
and $\cS_{\neg \vx_n}$ is also part of $\PP_n$.)

Therefore we have verified the assumptions of Lemma~\ref{l:shelling_house} and
we apply the shelling from this lemma. This essentially removes $\cB_n$ from
the intermediate complex with the exception that $\PP_n$ persists in the intermediate complex.

Next we describe shelling of $\cB_i$ for $i \in \{1, \dots, n-1\}$ assuming
that $\cB_j$ for $j > i$ has been already shelled (to $\PP_j$). The approach is
very analogous to the case of $\cB_n$, thus our description is more brief.

We intend to use Lemma~\ref{l:shelling_house} so that $K$ is the
intermediate complex, $L$ is the union of $\PP_i$ with the remainder (with
respect to $\cB_i$), and $H$ is $\cB_i$. We observe that the intersection of
$\cB_i$ with the remainder is the subcomplex of the attachment complex glued to
$\cT_0$, $\cV_{\vx_i}$, $\cV_{\vx_{i+1}}$, $\cB_{i-1}$ and one of $\cS_{\vx_i}$
or $\cS_{\neg \vx_i}$ (see Figure~\ref{f:blocker_i}). After taking the union with
$\PP_i$ we get exactly type $i$ shelling complex in the blocker case according
to the description in last paragraph of Subsection~\ref{ss:blocker_house}.
Therefore, we can apply Lemma~\ref{l:shelling_house}. This essentially removes
$\cB_i$ from
the intermediate complex with the exception that $\PP_i$ persists in the intermediate complex.

Finally, we shell $\cB_0$. This time, $\PP_0$ is not a part of the type $0$ shelling
complex, thus we only need to verify the intersection of $\cB_0$ with the
remainder (with respect to $\cB_0$) is the type $0$ shelling complex. This
follows from the description of the shelling complex in last paragraph of
Subsection~\ref{ss:blocker_house} as it is easy to check that this intersection
is the subcomplex of the attachment complex glued to $\cT_0$ or $\cV_{\vx_1}$
(because $\cB_1$ has been already shelled; see also Figure~\ref{f:blocker_0}).
Thus by applying shelling from Lemma~\ref{l:shelling_house}, we essentially
remove $\cB_i$ from the intermediate complex (this time including $\PP_0$).

\paragraph{Step 6, second shelling of variable gadgets:}

Now we shell the remainders of variable gadgets in an arbitrary order. See
Figure~\ref{f:variable_gadget} when following the shelling.

Let $\vx$ be a variable. First assume that $\vx$ is assigned TRUE. In this case
the unique tetrahedra of $\cV_{\vx}$ containing 
$\aa_{\vx}\bb_{\vx}\cc_{\vx}$, $\aa_{\vx}\cc_{\vx}\dd_{\vx}$ and
$\aa_{\vx}\bb_{\vx}\dd_{\vx}$ have been already
removed. Let $\cV'_{\vx}$ be the intersection of $\cV_{\vx}$ with the intermediate
complex; that is, $\cV'_{\vx}$ is obtained from $\cV_{\vx}$ after shelling the
aforementioned tetrahedra.

Because the blocker houses (except $\PP_i$s) as well as $\cS_{\vx}$ 
have been shelled in the previous steps, $\cV'_{\vx}$ meets the remainder of
the complex only in the triangle $\aa_{\neg \vx}
\bb_{\neg \vx}\dd_{\neg \vx}$ shared with $\cS_{\neg \vx}$ and the rectangle
$\bb_{\vx}\dd_{\vx}\dd_{\neg \vx}\bb_{\neg \vx}$. (This rectangle is shared with some
$\PP_i$).

Thus we can first greedily shell the tetrahedra meeting one of the rectangles
$\bb_{\vx}\cc_{\vx}\cc_{\neg \vx}\bb_{\neg \vx}$, $\cc_{\vx}\dd_{\vx}\dd_{\neg
\vx}\cc_{\neg \vx}$, in a triangle. (We do not specify the exact order 
because it depends on the choice the diagonals.) Then we shell the unique
tetrahedra containing $\aa_{\neg \vx} \bb_{\neg \vx}\cc_{\neg \vx}$ and
$\aa_{\neg \vx} \cc_{\neg \vx}\dd_{\neg \vx}$. Then we shell the tetrahedra meeting
$\bb_{\vx}\dd_{\vx}\dd_{\neg \vx}\bb_{\neg \vx}$ in a triangle. (Valid order again depends
on the choice of the diagonal). As the last one we remove the unique tetrahedron
containing $\aa_{\neg \vx} \bb_{\neg \vx}\dd_{\neg \vx}$. This essentially removes
the remainder of $\cV_{\vx}$ from the intermediate complex.

If $\vx$ is assigned false, then we apply the same approach as above after
swapping $\vx$ and $\neg \vx$ (except for $\cV_{\vx}$ and $\cV'_{\vx}$).

\paragraph{Step 7, shelling of $\PP_i$:}

In this step, we shell $\PP_i$ for $i \in \{1, \dots, n\}$ one by one in
an arbitrary order. Due to the previous shelling, each $\PP_i$
meets the remainder of the complex (with respect to $\PP_i$) in the rectangle
$\dd_{\vx_i}\ee_{\vx_i}\ee_{\neg \vx_i}\dd_{\neg \vx_i}$ and exactly one of the
triangles $\bb_{\vx_i}\dd_{\vx_i}\ee_{\vx_i}$ or $\bb_{\neg \vx_i}\dd_{\neg
\vx_i}\ee_{\neg \vx_i}$ (see Figures~\ref{f:blocker_i} and~\ref{f:blocker_n}). Because $\PP_i$ itself is a canonically triangulated
polytope (due to its construction) we may apply Lemma~\ref{l:shell_canonical}
where $K$ is the intermediate complex, $P = \PP_i$ and $L$ is the remainder
(with respect to $\PP_i$). This essentially removes $\PP_i$ from the
intermediate complex.

\paragraph{Step 8, shelling of negative splitter houses:}

Now we shell splitter houses $\cS_{\ell}$ for literals $\ell$ assigned FALSE in
arbitrary order. It is routine to check that $\cS_{\ell}$ meets the remainder (with
respect to $\cS_{\ell}$) in the part of the attachment
complex attached to $\cT_0$ or to $\cI_{\ell, \kappa_j}$ (for some $j$); see
Figure~\ref{f:splitter_attachment}. This is
exactly the second shelling complex in the splitter case and thus we may apply
Lemma~\ref{l:shelling_house} to shell $\cS_{\ell}$ as usual.

\paragraph{Step 9, shelling of incoming houses with negative literals:}

Now we shell incoming houses $\cI_{\ell, \kappa_j}$ for literals $\ell$ assigned
FALSE. (Those with $\ell$ assigned true have been shelled in Step 3.)

As everything else has been already shelled, $\cI_{\ell, \kappa_j}$ meets the
remainder (with respect to $\cI_{\ell, \kappa_j}$) only in $\cT_0$; see
Figure~\ref{f:incoming_attachment}. This is the
second shelling complex in the incoming case, thus we shell according to
Lemma~\ref{l:shelling_house}.

\paragraph{Step 10, shelling of $\cT$.}
After Step 9, the intermediate complex coincides with $\cT$ as we have
shelled all other gadgets. Recall that $\cT$ is obtained by taking a
triangulated disk $T$ (which coincides with $\cT_0$ up to isomorphism); 
considering the product $T \times I$ obtaining first a 
polytopal decomposition of $T \times I$ into prisms; and the triangulating by
taking some canonical triangulation.

It is well known (and easy to prove) that every triangulated disk is shellable; see,
e.g.~\cite{danaraj-klee78algo} for an even stronger statement. Removing the
triangular prisms $P$ of $T \times I$ in the order corresponding to a shelling
means that such a prism $P$ satisfies the assumptions of
Lemma~\ref{l:shell_canonical}. Thus, $\cT_0$ is shellable by a repeated
application of Lemma~\ref{l:shell_canonical}. This shelling finishes shelling
of $\cK_\phi$.

\section{Intermezzo: $\cK_\phi$ is a ball}
Before proving the second implication, we will show that $\cK_\phi$ is a ball.

\begin{proposition}
\label{p:ball}
  $\cK_\phi$ is a $3$-ball.
\end{proposition}

In fact, if $\phi$ is satisfiable, then it can be checked that any shelling used
in the previous section is a shelling in PL sense, thus it can be deduced (by
a little thought) that $\cK_\phi$ is a ball from Lemma~\ref{l:pseudomanifold}.
However, we need to know that $\cK_\phi$ is a ball even if $\phi$ is not
satisfiable. We will show this by providing another shelling in PL sense which
does not distinguish whether $\phi$ is satisfiable or not. Of course, this
shelling cannot be a classical shelling in simplicial sense if we aim to show
that $\cK_\phi$ is not shellable for $\phi$ not satisfiable. We will be using
that $\cK_\phi$ is a $3$-pseudomanifold which follows from our construction.
Whenever we glued two gadgets together the triangles of the intersection are
only in those two gadgets.

\begin{lemma}
\label{l:shell_U}
  Let $\cU$ be the union of all blocker houses and variable gadgets inside
  $\cK_\phi$. Let $\cR$ be the subcomplex of $\cK_\phi$ formed by all
  tetrahedra not in $\cU$. Then $\cU$ is a ball which meets $\cR$ in a disk.  
\end{lemma}

\begin{proof}
  It is easy to check that $\cU$ meets $\cR$ in a disk. This disk consists of
  the parts of the blocker houses attached to $\cT_0$, $\cA$, $\cS_{\vx}$,
  $\cS_{\neg \vx}$
  (for some $x$) and parts of the variable gadgets attached to $\cS_{\vx}$,
  $\cS_{\neg \vx}$ (for some $x$); see
  Figures~\ref{f:variable_gadget},~\ref{f:blocker_i},~\ref{f:blocker_0}
  and~\ref{f:blocker_n}. It remains to show that $\cU$ is a ball. 

  For the remainder of this proof, all the gadgets are considered only inside
  $\cU$. We again use a notion of intermediate complex and the remainder
  analogously as in the beginning of Section~\ref{s:sat=>shell} with exception
  that everything is taken inside $\cU$ rather than $\cK_\phi$.

  First we observe that $\cB_0$ intersects the remainder in a disk (formed by
  the parts of the attachment complex glued to $\cV_{\vx_1}$ or $\cB_1$; see
  Figure~\ref{f:blocker_0}). Thus we may elementarily shell $\cB_0$ in PL sense. Then
  $\cV_{\vx_1}$ intersects the remainder in a disk formed by the intersection of
  $\cV_{\vx_1}$ with $\cB_1$ (see Figure~\ref{f:variable_gadget}). This allows us
  to elementarily shell $\cV_{\vx_1}$ in PL sense. Next, $\cB_1$ intersects the remainder in a
  disk (formed by the parts of the attachment complex glued to $\cV_{\vx_2}$ or
  $\cB_2$; see  Figure~\ref{f:blocker_i}). This allows us to perform an
  elementary shelling of $\cB_1$ in PL sense. We continue by elementary
  shellings of $\cV_{\vx_2}, \cB_2, \dots, \cV_{\vx_n}$ in this order yielding
  $\cB_n$. In other words $\cU$ shells in PL sense to $\cB_n$. Because
  $\cK_\phi$ is a $3$-pseudomanifold, $\cU$ is a $3$-pseudomanifold as well.
  Thus, because $\cB_n$ is a PL $3$-ball, we deduce that $\cU$ is a PL $3$-ball
  as well by Lemma~\ref{l:pseudomanifold}.
\end{proof}

\begin{proof}[Proof of Proposition~\ref{p:ball}] Now Lemma~\ref{l:shell_U} implies that $\cK_\phi$ elementarily shells in PL
sense to $\cR$. Then we can shell elementarily in PL sense all splitter houses
one by one. This is similar to Step~2 in Section~\ref{s:sat=>shell} with
exception that we can perform this for all literals (because the variable
gadgets were already shelled). Another exception is that the blockers are
already missing. But this is not a problem; we do not have to verify the
assumptions of Lemma~\ref{l:shelling_house}. We only check that
  each splitter meets the remainder in a disk which is sufficient for an elementary shelling in
PL sense. Now we can shell incoming houses, clause gadgets and outgoing
houses. Here we can use the approach from Step 3 again used for all literals
but otherwise essentially in verbatim (because the attachments are the
same---in case of turbine only one type of shelling complex occurs). We also
point out that for the purposes of this proof, Step~3 could be simplified but
we do not attempt so as we can use something that has been already done.
Finally we shell the conjunction gadget by an elementary shelling in PL sense
obtaining $\cT$ which is a $3$-ball. Altogether, $\cK_\phi$ shells in PL sense
to $\cT$. Therefore, Lemma~\ref{l:pseudomanifold}
implies that $\cK_\phi$ is a $3$-ball as well.
\end{proof}

\section{Shellable implies satisfiable}
\label{s:shell=>sat}
Assume that our triangulated ball $\cK_\phi$ is shellable. We aim to show that
$\phi$ is satisfiable. We fix an arbitrary shelling down of $\cK_\phi$.
For tetrahedra $\Lambda$ and $\Gamma$ of $\cK_\phi$ we write $\Lambda \prec
\Gamma$ if $\Lambda$ is appears before $\Gamma$ in this shelling. For a tetrahedron $\Lambda$ of
$\cK_\phi$ we
also denote $\cK_\Lambda$ the subcomplex of $\cK_\phi$ formed by $\Lambda$ and
the tetrahedra that will follow in our shelling including the one that remains
after finishing the shelling; $\cK_\Lambda := (\cK_\phi)_\Lambda$ in the notation introduced above
Lemma~\ref{l:all_balls}. We know that $\cK_\Lambda$ is
a ball by Lemma~\ref{l:all_balls} and Proposition~\ref{p:ball} and we also know
that $\Lambda$ is free in $\cK_\Lambda$ by Lemma~\ref{l:free_facet}.

Given a gadget $\cG$, let $\Gamma(\cG)$ be the first tetrahedron removed from
$\cG$ during our shelling. In addition, let $\cH$ be a splitter house, an
incoming house, an outgoing house, or a blocker house. Recall that $\cH$
contains an important triangle $\tau$, which we further denote $\tau(\cH)$
whenever we want to emphasize $\cH$. In each case, $\tau(\cH)$ is in two
tetrahedra of $\cK_\phi$; exactly one of them is inside $\cH$. The tetrahedron
containing $\tau(\cH)$ inside $\cH$ will be denoted by $\Theta(\cH)$ whereas
the tetrahedron containing $\tau(\cH)$ outside $\cH$ will be denoted
$\Xi(\cH)$.

\stepcounter{theorem}

\begin{claim}
  \label{c:order_H}
  For every thick 1-house $\cH$ (i.e. $\cH$ is one of $\cS_\ell$,
  $\cI_{\ell,\kappa}$, $\cO_\kappa$ or $\cB_i$) we have
\[
  \Xi(\cH) \prec \Gamma(\cH) \preceq \Theta(\cH).
\]
\end{claim}

\begin{proof}
  The inequality $\Gamma(\cH) \preceq \Theta(\cH)$ follows immediately from the
  fact that $\Gamma(\cH)$ is the first tetrahedron removed from $\cH$ during
  our shelling, thus it remains to show $\Xi(\cH) \prec \Gamma(\cH)$.

 For simplicity, let $\Gamma := \Gamma(\cH)$. We use Lemma~\ref{l:1house_blocked} with $B
= \cK_{\Gamma}$ and $H = \cH$. The assumption on faces of $\partial
H$ not in the attachment complex is satisfied as every face of $\partial H =
\partial \cH$ not in the attachment complex is also a face of $\partial
\cK_\phi$ due to our construction. But then it also has to be a face of $\partial
\cK_\Gamma$ as $\cK_\Gamma$ is a subcomplex of $\partial \cK_\phi$ (and both
  are balls). Now for contradiction assume that $\Gamma \prec \Xi(\cH)$.
  As we also have $\Gamma \preceq \Theta(\cH)$ we deduce that 
  both $\Theta(\cH)$ and $\Xi(\cH)$ belong to $\cK_\Gamma$. This
  means that the assumption of Lemma~\ref{l:1house_blocked} 
  that the triangle $\tau(\cH)$ is not in 
  $\partial B = \partial \cK_{\Gamma}$ is also satisfied. From
  Lemma~\ref{l:1house_blocked} we
deduce that there is no free tetrahedron of $\cK_{\Gamma}$ contained in $\cH$
which contradicts the fact that $\Gamma$ is a free tetrahedron of
$\cK_{\Gamma}$ in $\cH$. This finishes the proof of the claim.
\end{proof}

Now, for a clause $\kappa$, let $\Xi_\kappa$ be the unique tetrahedron of $\cO_\kappa$ which meets $\cA$ in a
triangle. 
\begin{claim}
  \label{c:order_A}
 For every clause $\kappa$ we have
\[
  \Xi_\kappa \prec \Xi(\cB_n).
\]
\end{claim}

\begin{proof}
We use Lemma~\ref{l:conjunction_cone} with $B = \cK_\phi$ and $C = \cA$.
The two initial assumptions are satisfied due to the construction. Note that
the tetrahedron $\Delta_1$ in the statement of the lemma is the tetrahedron
$\Xi(\cB_n)$ and $\Delta'_0$ is the tetrahedron $\Theta(\cB_n)$. 
We remark that conclusion that $\Delta'_0 = \Theta(\cB_n)$ is shelled
before $\Delta_1 = \Xi(\cB_n)$ cannot be satisfied due to Claim~\ref{c:order_H}.
  Thus all the tetrahedra $\Delta'_i$ (for $i \in \{1, \dots, k-1\}$) in the statement of
Lemma~\ref{l:conjunction_cone} have to be shelled before $\Xi(\cB_n)$ while
$\Xi_\kappa$ is one of them. This finishes the proof of the claim.
\end{proof}

Now, for a clause $\kappa$ and a literal $\ell$ in $\kappa$, let
$\Xi_\ell(\cC_\kappa)$ be the
unique tetrahedron outside $\cC_\kappa$ which contains the triangle $\xx_{\ell,
\kappa} \yy_{\ell,\kappa}\zz_{\ell,\kappa}$. We recall that this triangle is
one of the triangles $\tau_1$, $\tau_2$ or $\tau_3$ of $\cC_\kappa$; see
Figure~\ref{f:clause_attachment}. We also remark that $\Xi_\ell(\cC_\kappa)$ belongs
to the incoming house $\cI_{\ell,\kappa}$.

\begin{claim}
  \label{c:order_C}
 For every clause $\kappa$ there is a literal $\ell$ in this clause such that
\[
  \Xi_\ell(\cC_\kappa) \prec \Gamma(\cC_\kappa).
\]
\end{claim}

\begin{proof}
  The claim follows from Lemma~\ref{l:turbine_blocked} in a similar way as
  Claim~\ref{c:order_H} from Lemma~\ref{l:1house_blocked}. Namely, we use
  Lemma~\ref{l:turbine_blocked} with $B = \cK_\Gamma$ and $T = \cC_\kappa$ where
  $\Gamma = \Gamma(\cC_\kappa)$. The
  assumption on faces of $T$ not in the attachment complex is satisfied 
  as every face of $\partial T =
\partial \cC_\kappa$ not in the attachment complex is also a face of $\partial
\cK_\phi$ due to our construction. But then it also has to be a face of $\partial
\cK_\Gamma$ as $\cK_\Gamma$ is a subcomplex of $\cK_\phi$ (and both
  are balls).
  
  Now for contradiction assume that $\Gamma \prec \Xi_\ell(\cC_\kappa)$ for every
  $\ell$ in $\kappa$. We deduce that the tetrahedra $\Xi_\ell(\cC_\kappa)$ belong to
  $\cK_\Gamma$. As whole $\cC_\kappa$ belongs to $\cK_\Gamma$ as well (due to
  definition of $\Gamma = \Gamma(\cC_\kappa)$), we deduce that none of the three
  triangles $\tau_1, \tau_2, \tau_3$ of the thick turbine $\cC_\kappa$ belongs to
  $\partial \cK_\Gamma$. Lemma~\ref{l:turbine_blocked} implies that there is no
  free tetrahedron of $B = \cK_\Gamma$ contained in $T = \cC_\kappa$. However
  $\Gamma$ is a free tetrahedron of $B = \cK_\Gamma$ contained in $T =
  \cC_\kappa$
  which is the required contradiction.
\end{proof}

Now we aim to define our assignment of variables of the formula $\phi$. If 
$\Xi_\ell(\cC_\kappa) \prec \Gamma(\cC_\kappa)$ we set $\ell$ to TRUE 
(that is, if $\ell = \vx$ for a variable $\vx$, then we set $\vx$ to TRUE and
if $\ell = \neg \vx$, we set $\vx$ to FALSE). If $\vy$ is a variable which did not get any assignment from this rule, we set it arbitrarily TRUE or FALSE. As soon as we verify that there are no conflicts, that is, no literal has been assigned both TRUE and FALSE, we get a satisfying assignment due to
Claim~\ref{c:order_C}.

Thus it remains to check that there are no conflicts. For contradiction assume
that $\ell$ has been set both TRUE and FALSE. Without loss of generality $\ell
= \vx$ for some variable $\vx$ (otherwise we swap $\ell$ with $\neg \ell$). Then we
have $\Xi_{\vx}(\cC_\kappa) \prec \Gamma(\cC_\kappa)$ for some clause $\kappa$
containing $\vx$ 
and $\Xi_{\neg \vx}(\cC_{\kappa'}) \prec \Gamma(\cC_{\kappa'})$ for another
clause $\kappa'$ containing $\neg \vx$. Using the former inequality, 
Claims~\ref{c:order_H} and \ref{c:order_A} and the
facts that $\Xi(\cI_{\vx,\kappa})$ belongs to $\cS_{\vx}$,
$\Xi_{\vx}(\cC_\kappa)$ belongs to $\cI_{\vx,\kappa}$, $\Xi(\cO_\kappa)$
belongs to $\cC_\kappa$, and $\Xi_\kappa$ belongs to $\cO_\kappa$ we
deduce:

\begin{equation}
\label{e:Sx}
   \Xi(\cS_{\vx}) \prec \Gamma(\cS_{\vx}) \preceq \Xi(\cI_{\vx,\kappa}) \prec
   \Gamma(\cI_{\vx,\kappa}) \preceq \Xi_{\vx}(\cC_\kappa) \prec
   \Gamma(\cC_\kappa) \preceq \Xi(\cO_\kappa) \prec
    \Gamma(\cO_\kappa) \preceq \Xi_\kappa \prec \Xi(\cB_n).
\end{equation}

Analogously, we deduce:
\begin{equation}
\label{e:Snx}
   \Xi(\cS_{\neg \vx}) \prec \Gamma(\cS_{\neg \vx}) \preceq \Xi(\cI_{\neg
   \vx,\kappa'}) \prec \Gamma(\cI_{\neg \vx,\kappa'})
      \preceq \Xi_{\neg \vx}(\cC_{\kappa'})
      \prec \Gamma(\cC_{\kappa'}) \preceq \Xi(\cO_{\kappa'}) \prec
      \Gamma(\cO_{\kappa'}) \preceq \Xi_{\kappa'} \prec \Xi(\cB_n).
\end{equation}

In addition, using Claim~\ref{c:order_H} and the fact that $\Xi(\cB_i)$ belongs
to $\cB_{i+1}$ for $i \in \{0, \dots, n-1\}$ we deduce
\begin{equation}
\label{e:Bs}
  \Xi(\cB_n) \prec \Gamma(\cB_n) \preceq \Xi(\cB_{n-1}) \prec \Gamma(\cB_{n-1})
  \preceq \cdots \preceq \Xi(\cB_0) \prec \Gamma(\cB_0).
\end{equation}

Note that $\Xi(\cS_{\vx})$ and $\Xi(\cS_{\neg \vx})$ both belong to the variable
gadget $\cV_{\vx}$ and they are the unique tetrahedron of $\cV_{\vx}$ containing the
triangle $\aa_{\vx}\bb_{\vx}\dd_{\vx}$ and $\aa_{\neg \vx}\bb_{\neg
\vx}\dd_{\neg \vx}$, respectively. (To recall the notation see
Figure~\ref{f:variable_gadget}.) Because $\Xi(\cS_{\vx}) \neq \Xi(\cS_{\neg
\vx})$, we either have $\Xi(\cS_{\vx}) \prec \Xi(\cS_{\neg \vx})$ or
$\Xi(\cS_{\neg \vx}) \prec \Xi(\cS_{\vx})$. 

First, we assume $\Xi(\cS_{\vx}) \prec \Xi(\cS_{\neg \vx})$ and we will bring this
case to the contradiction. The other case will be analogous. Let $\Psi$ be the
tetrahedron that immediately follows after $\Xi(\cS_{\neg \vx})$ in our shelling.

\begin{claim}
\label{c:six_tetrahedra}
  The complex $\cK_\Psi$ contains all six tetrahedra of $\cK_{\phi}$ which do not
  belong to $\cV_{\vx}$ while they meet one of the rectangles
  $\bb_{\vx}\cc_{\vx}\cc_{\neg \vx}\bb_{\neg \vx}$, $\bb_{\vx}\dd_{\vx}\dd_{\neg
  \vx}\bb_{\neg \vx}$ or $\cc_{\vx}\dd_{\vx}\dd_{\neg \vx}\cc_{\neg \vx}$ in a triangle. 
\end{claim}

\begin{proof}
  Note that these six tetrahedra are contained in $\cB_i$ for some $i$ (see
  Figure~\ref{f:variable_gadget}). Therefore the claim follows from the fact
  that $\Psi \prec \Gamma(\cB_i)$ for all $i$ which follows from~\eqref{e:Snx}
  and~\eqref{e:Bs} and the definition of $\Psi$. 
\end{proof}

Now we use Lemma~\ref{l:tp_blocked}
with $P = \cV_{\vx}$, $B = \cK_\phi$ and $L = \cK_\Psi$. The assumption on $P
\cap \partial B$ is satisfied due to the construction of $\cK_\phi$.
Also the assumption on six tetrahedra is satisfied due to
Claim~\ref{c:six_tetrahedra}. It follows from the lemma that either
$\Xi(\cS_{\vx})$ or $\Xi(\cS_{\neg \vx})$ belongs to $\cK_\Psi$. (Note that
$\Xi(\cS_{\vx})$ is one of the first three tetrahedra while
$\Xi(\cS_{\neg \vx})$ is one of the second three tetrahedra in the conclusion of
the lemma.) This contradicts that $\Xi(\cS_{\vx}) \prec \Xi(\cS_{\neg \vx}) \prec \Psi$.

The other case $\Xi(\cS_{\neg \vx}) \prec \Xi(\cS_{\vx})$ yields a contradiction by
a symmetric argument replacing $\vx$ and $\neg \vx$. In an analogy of
Claim~\ref{c:six_tetrahedra}, we use~\eqref{e:Sx} and~\eqref{e:Bs}.

This finishes the proof that $\phi$ is satisfiable and thereby the proof of
Theorem~\ref{t:main} as well.

\section{Hardness of shellability for embedded $2$-complexes.}
\label{s:embedded_hard}
In this section we use the notation from Section~\ref{s:thin_construction}.
Let us recall that $\cK'_\phi$ is the complex built in
Subsection~\ref{ss:gluing_thin}. We consider a $2$-complex $\cK''_\phi$ which
is a subcomplex obtained from $\cK'_\phi$ by removing the interior of each
variable gadget. (In other words we remove the edge $\aa_{\vx}\aa_{\neg \vx}$
from each variable gadget $\cV_\vx$ and all simplices containing this edge.)
We aim to show that $\cK''_\phi$ is shellable if and only if $\phi$ is
satisfiable. Note that $\cK''_\phi$ is a 2-complex and it is embeddable into
$3$-space because it is a subcomplex of $\cK'_\phi$. Therefore the
aforementioned equivalence will prove Theorem~\ref{t:embedded_hard}.

\paragraph{Satisfiable implies shellable.}
For the proof of the first implication, we need a few auxiliary tools.

Let $K$ be a simplicial complex. We say that the complex $K$ collapses to a
complex $K'$ via an $(i,j)$ collapse, if it arises from $K$ by an elementary
collapse removing a free face of dimension $j$ contained in a unique
face of dimension $i$. 

\begin{observation}
  \label{o:2i}
  A $2$-complex is shellable if it collapses to a triangle using
  only $(2,0)$ and $(2,1)$ collapses.
\end{observation}

\begin{proof}
It is sufficient to observe that the order of removing triangles by $(2,0)$ and
  $(2,1)$ collapses is a shelling down.
\end{proof}

The following lemma is a strengthening of Lemma~16 in~\cite{gpptw19}.

\begin{lemma} 
\label{l:stronger_16}
  Let $D$ be a triangulated disk and $T$ be a tree in 1-skeleton of
  $D$ with at least $2$ vertices (i.e., $T$ is 1-dimensional). Then $D$
  collapses to $T$ using only $(2,0)$ and $(2,1)$ collapses.
\end{lemma}

\begin{proof}
  We will prove the claim by induction in the number of triangles. It is easy
  to resolve the case when $D$ contains exactly one triangle.

  For the second induction step, we first consider the case that $D$ collapses
  to a complex $D'$ via a $(2,0)$ or $(2,1)$ collapse so that $D'$ is again a
  disk which contains $T$. Then we perform this collapse and we provide the
  required collapses of $D'$ by induction.

  Thus, it remains to consider the case where such a collapse is not possible.
  Let $\varepsilon = uv$ be an edge on $\partial D$ which is not contained in $T$. Let
  $uvw$ be the unique triangle of $D$ containing $\varepsilon$. We deduce that
  $w$ belongs to $\partial D$, otherwise a $(2,1)$ collapse through
  $\varepsilon$ would yield a disk. 

  Let $D'$ be complex obtained by a $(2,1)$ collapse through
  $\varepsilon$. 
  Because $w \in \partial D$, we deduce that
  $D'$ is a wedge of two subcomplexes $D_u$ and $D_v$ which overlap in $w$;
	$D_u$ contains $u$ and $D_v$ contains $v$. Each of $D_u, D_v$ may either be an edge or a disk but two edges are impossible as $D$ contains at least two triangles.
	Let $T_u := D_u \cap T$ and $T_v := D_v \cap T$; see Figure~\ref{f:disk_collapse}. 

\begin{figure}
	\begin{center}
		\includegraphics{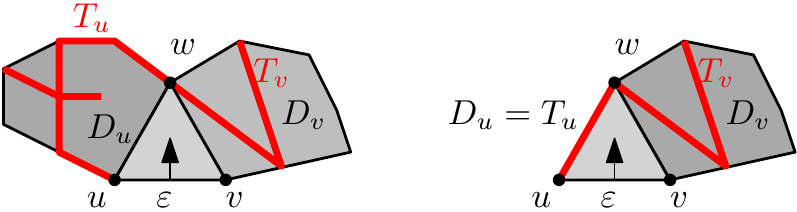}
		\caption{$D_u$ and $D_v$. Left: Both of them are disks. Right: one of them is an edge.}
		\label{f:disk_collapse}
	\end{center}
\end{figure}

  First we deduce that each of $T_u, T_v$ has at least two vertices (i.e. at least one
  edge). For contradiction assume, for example, that $T_u$ contains at most one
  vertex. (If it contains one vertex, then this vertex has to be $w$ as $T$ is
  connected.) If $D_u$ contains no triangle, then we may perform the $(2,0)$
  collapse in $D$ through $u$ yielding a disk contradicting the fact that such
  a collapse is impossible. Similarly, if $D_u$ contains single triangle, we
  again may perform a $(2,0)$ collapse in $D$ removing that triangle yielding a
  disk, a contradiction. Thus we may assume that $D_u$ contains at least two
  triangles. If $D_u$ contains a vertex in the interior of $D_u$, then there is
  also a vertex $x$ in the interior of $D_u$ which forms a triangle with some
  edge $\varepsilon'$ on $\partial D_u$. Then the $(2,1)$ collapse in $D$
  through $\varepsilon'$ yields a disk, contradiction. Thus we may assume that
  $D_u$ contains no vertex in the interior. In other words, $D_u$ is a polygon
  with at least four vertices triangulated by adding diagonals. Such a polygon
  contains at least two triangles with two edges on the boundary (consider the
  dual graph). One of these triangles can be removed from $D$ by a $(2,0)$
  collapse yielding a disk, a contradiction. (The only way how to block this is
  if the two edges on the boundary of the triangle share $w$; this may happen
  only for one triangle.) This finishes the proof that both $T_u$ and $T_v$
  have at least two vertices.

 Now let us assume that one of $D_u$, $D_v$, say $D_u$, is just an edge. Then
  necessarily $D_u = T_u$ is the tree formed by the edge $uw$. By induction, 
  $D_v$ collapses to $T_v$ using only $(2,0)$ and $(2,1)$ collapses. These
  collapses also work in $D'$ as the edge $uw$ cannot block any of them.
  Therefore $D$ collapses to $T$ using only $(2,0)$ and $(2,1)$ collapses by
  first collapsing to $D'$ and using the aforementioned collapses.

 Finally assume that both $D_u$ and $D_v$ are disks. By induction,
  $D_u$ collapses to $T_u$ and $D_v$ collapses to $T_v$ 
  using only $(2,0)$ and $(2,1)$ collapses. These collapses can be performed in
  $D'$ (say first those in $D_u$, then those in $D_v$) because no simplex of
  $D_u$ may block a collapse in $D_v$ except possible $w$ but $w$ necessarily belongs to
  $T$. Similarly, no simplex of $D_u$ may block a collapse in $D_u$. Therefore $D$ collapses to $T$ using only $(2,0)$ and $(2,1)$ collapses by 
  first collapsing to $D'$ and then using the aforementioned collapses.
\end{proof}

Now we are ready to prove that if $\phi$ is satisfiable, then $\cK''_\phi$ is
shellable. The proof is very similar to the proof in
Subsection~\ref{ss:sat=>col} and we will often refer there. We start describing
a shelling down of $\cK''_\phi$ (assuming that $\phi$ is satisfiable and after
fixing an assignment).

For every variable $\vx$ that is assigned TRUE we start with shelling of
triangles $\uu_\vx\vv_\vx\aa_\vx$, $\uu_\vx\ww_\vx\aa_\vx$ and
$\vv_\vx\ww_\vx\aa_\vx$ in this order. 
Similarly, if $\vx$ that is assigned FALSE we start with shelling of
triangles $\uu_\vx\vv_\vx\aa_{\neg \vx}$ and $\uu_\vx\ww_\vx\aa_{\neg \vx}$,
$\vv_\vx\ww_\vx\aa_{\neg \vx}$ in this order. This we do for each variable
independently for an arbitrary order of the variables. We obtain an auxiliary
complex $\cK'''$. (This is $\cK''_\phi[\tau_i, \dots, \tau_m]$ where $\tau_i,
\dots, \tau_m$ are triangles of $\cK''_\phi$ which do not appear in the
shelling above.) Note that $\vv_\vx\aa_\vx$ is free in $\cK'''$ if $\vx$ is
assigned TRUE and $\vv_\vx\aa_{\neg \vx}$ is free $\cK'''$ if $\vx$ is assigned
FALSE.

Now we start collapsing $\cK'''$ using only $(2,0)$ and $(2,1)$ collapses (to
a triangle). From Observation~\ref{o:2i} we will get that $\cK'''$ is shellable and therefore
$\cK''_\phi$ is shellable as well. Now we perform the collapses as in
Subsection~\ref{ss:sat=>col} starting with $\cS_\ell$ essentially in verbatim, we just need to be
careful about a few details:

\begin{itemize}
  \item We need to perform only $(2,0)$ or $(2,1)$ collapses. The collapses in
    Subsection~\ref{ss:sat=>col} are given either by Lemma~\ref{l:thin_1house}
    or Lemma~\ref{l:thin_turbine}. In both cases, by short backtracking of the
    references, the proof relies on Lemma~16 in~\cite{gpptw19}. Once we use the
    stronger version, Lemma~\ref{l:stronger_16}, we get versions of Lemma~\ref{l:thin_1house}
    or Lemma~\ref{l:thin_turbine} (the part regarding collapsibility) 
    using only $(2,0)$ or $(2,1)$ collapses.
  \item When collapsing the remainder of the variable gadget (one but last
    paragraph of Subsection~\ref{ss:sat=>col}), we collapse only the remainder
    of the boundary, because the interior is not in $\cK''_\phi$.
  \item In the very last step, we do not collapse the template to a point but
    to a triangle. (We can use Lemma~\ref{l:stronger_16} to collapse to an edge
    and then undo the last step.)
\end{itemize}

\paragraph{Shellable implies satisfiable.}
We again need a few auxiliary claims first.

\begin{lemma}
  \label{l:cK'_contractible}
  $\cK'_\phi$ is contractible
\end{lemma}

\begin{proof}
  If we build $\cK'_\phi$ by gluing the gadgets in the order as in
  Subsection~\ref{ss:gluing_thin}, we observe that every new gadget is added to
  the previous gadgets along a tree. Also every gadget (including the template)
  is contractible (due to Lemmas~\ref{l:thin_1house} and~~\ref{l:thin_turbine}
  as collapses preserve the homotopy type, or trivially by the construction in
  case of the variable gadget). This means that $\cK'_\phi$ is contractible.
  (This follows for example from~\cite[Proposition~4.1.5]{matousek03} by
  contracting the shared tree in every step.
\end{proof}

\begin{corollary}
  \label{c:wedge}
  $\cK''_\phi$ is homotopy equivalent to 
  a wedge of $n$ 2-spheres, where $n$ is the number of variables of
  $\phi$.
\end{corollary}

\begin{proof}
  Let $X'$ be the topological space obtained by contracting each variable
  gadget in $\cK'_\phi$ and $X''$ be the topological space obtained by
  contracting the boundary minus a triangle of each variable gadget. By
  construction, $X'$ and $X''$ differ so that $X''$ contains additionally a $2$-sphere
  attached to each contracted variable gadget.
  By~\cite[Proposition~4.1.5]{matousek03}, $X'$ is homotopy equivalent to
  $\cK'_\phi$, therefore contractible by Lemma~\ref{l:cK'_contractible}. 
  This implies that $X''$ is homotopy
  equivalent to the wedge of $n$ 2-spheres. By~\cite[Proposition~4.1.5]{matousek03},
  again, $X''$ is homotopy equivalent to $\cK''_\phi$.
\end{proof}

Now we assume that $\cK''_\phi$ is shellable. Let $n$ be the number of
variables of $\phi$. By Corollary~\ref{c:wedge},
$\cK''_\phi$ is homotopy equivalent to a wedge of $n$ 2-spheres. Therefore its
reduced Euler characteristic is equal to $n$. By~\cite[Theorem~8,
(i)$\Rightarrow$(iii)]{hachimori08}, $\cK''_\phi$ is collapsible after removing
some $n$ triangles $\tau_1, \dots, \tau_n$. Let $\bar \cK$ be the complex
obtained from $\cK''_\phi$ by removing these triangles. We remark that in each
$\partial \cV_{\vx}$ we have to remove at least one triangle otherwise it $\bar
\cK$ cannot be contractible, a fortiori it cannot be collapsible. Because we
only have $n$ triangles, we have to remove exactly one from each $\partial
\cV_{\vx}$. It is not hard to observe that $\cK'_\phi$ collapses to $\bar \cK$
as each variable gadget $\cV_\vx$ collapses to its boundary minus a triangle.
Therefore, because $\bar \cK$ is collapsible, we deduce that $\cK'_\phi$ is collapsible as
well. It follows that $\phi$ is satisfiable; see Subsection~\ref{ss:col=>sat}.
This finishes the proof of Theorem~\ref{t:embedded_hard}.

\bibliographystyle{alpha}
\bibliography{shard3D}
\end{document}